%% file: main.tex
\documentclass[12pt]{l4dc2021}

\usepackage[utf8]{inputenc} 
\usepackage[T1]{fontenc}    
\usepackage{booktabs}       
\usepackage{amsfonts}       
\usepackage{nicefrac}       
\usepackage{microtype}      
\usepackage{wrapfig}
\usepackage{algorithm} 
\usepackage{appendix}
\usepackage{xcolor}
\usepackage{import}

\usepackage[utf8]{inputenc}
\usepackage[english]{babel}
\usepackage{kky}

\usepackage{cleveref}
\usepackage{bm}
\usepackage{subcaption}

\renewcommand{\d}[1]{\ensuremath{\operatorname{d}\!{#1}}}
\crefname{theorem}{theorem}{theorem}
\crefname{corollary}{corollary}{corollary}
\crefname{assumption}{assumption}{assumption}
\crefname{lemma}{lemma}{lemma}
\crefname{remark}{remark}{remark}
\crefname{proposition}{proposition}{proposition}
\crefname{conjecture}{conjecture}{conjecture}
\crefname{definition}{definition}{definition}
\newtheorem*{theorem*}{Theorem}
\makeatletter
\newtheorem*{rep@theorem}{\rep@title}
\newcommand{\newreptheorem}[2]{%
\newenvironment{rep#1}[1]{%
 \def\rep@title{#2 \ref{##1}}%
 \begin{rep@theorem}}%
 {\end{rep@theorem}}}
\makeatother
\newreptheorem{theorem}{Theorem}

\DeclareFontFamily{U}{wncy}{}
\DeclareFontShape{U}{wncy}{m}{n}{<->wncyr10}{}
\DeclareSymbolFont{mcy}{U}{wncy}{m}{n}
\DeclareMathSymbol{\Sh}{\mathord}{mcy}{"58} 

\newcommand{\cmax}{c_\mathrm{max}}

\newcommand{\ar}[1]{{\color{blue} Arash: #1}}
\newcommand{\mgh}[1]{{\color{orange} MGH: #1}}

\renewcommand{\ar}[1]{{\color{blue} {}}}
\renewcommand{\mgh}[1]{{\color{cyan} {}}}

\newcommand{\sat}{{\rm SAT}}

\DeclareMathSymbol{:}{\mathrel}{operators}{"3A}

\title[Neural Lyapunov Redesign]{Neural Lyapunov Redesign}
\usepackage{times}

\author{%
 \Name{Arash Mehrjou} \Email{amehrjou@ethz.ch}\\
 \addr Max Planck Institute for Intelligent Systems \& ETH Z\"urich
 \AND
 \Name{Mohammad Ghavamzadeh} \Email{ghavamza@google.com}\\
 \addr Google Research
 \AND
 \Name{Bernhard Sch\"olkopf} \Email{bs@tue.mpg.de}\\
 \addr Max Planck Institute for Intelligent Systems
 }

\begin{document}

\maketitle

\begin{abstract}%
Learning controllers merely based on a performance metric has been proven effective in many physical and non-physical tasks in both control theory and reinforcement learning. However, in practice, the controller must guarantee some notion of safety to ensure that it does not harm either the agent or the environment. Stability is a crucial notion of safety, whose violation can certainly cause unsafe behaviors. Lyapunov functions are effective tools to assess stability in nonlinear dynamical systems. In this paper, we combine an improving Lyapunov function with automatic controller synthesis in an iterative fashion to obtain control policies with large safe regions. We propose a two-player collaborative algorithm that alternates between estimating a Lyapunov function and deriving a controller that gradually enlarges the stability region of the closed-loop system. We provide theoretical results on the class of systems that can be treated with the proposed algorithm and empirically evaluate the effectiveness of our method using an exemplary dynamical system.
\end{abstract}

\begin{keywords}%
 Lyapunov function, Controller Synthesis, Actor-Critic, Neural Networks%
\end{keywords}

\input{intro_arxiv.tex}

\input{problem_statement_arxiv.tex}

\input{method_arxiv.tex}

\input{experiments_arxiv.tex}

\input{related_works_arxiv.tex}

\input{supplement_arxiv.tex}

\clearpage
\bibliography{refs_arxiv} 

\end{document}

%% file: intro_arxiv.tex
\section{Introduction}

\label{sec:intro}

Studying the stability region of autonomous systems and designing controllers (policies) to drive a non-autonomous system towards a target behavior are of fundamental importance in many disciplines, such as aviations~\citep{liao2004analysis}, autonomous driving~\citep{wen2018new}, and robotics~\citep{pierson2017deep}. An indisputable goal for a controller is to stabilize the system. Unlike the global nature of linear systems, stability is a local property in nonlinear systems. Knowledge of the stability region is essential in many applications, such as stability of power systems~\citep{xin2007methods}, design of associative memory in artificial neural networks~\citep{hopfield1994neurons}, robotics~\citep{westervelt2003switching}, and biology~\citep{baer2006multiple}.

Controllers that enhance the stability region of a system, also known as Region of Attraction (RoA), are highly desired as they make more clever use of the inherent nonlinear structure of the system. For example, an autonomous driving system will remain safe under more diverse and potentially harsh conditions~\citep{imani2018region}.

A great number of methods for designing a controller for nonlinear systems and determining their RoA~\citep{isidori2014nonlinear,khalil2002nonlinear} have been proposed in the literature mainly based on Lyapunov's theory of stability~\citep{liapounoff1907probleme}. However, the design of a stabilizing controller has been ad-hoc for every class of nonlinear systems. Our work is inspired by a classic method called~\emph{Lyapunov Redesign} in control theory that uses a given Lyapunov function to design a controller such that the closed-loop system becomes stable when assessed by that Lyapunov function. Here, we relax the necessity to know the Lyapunov function by learning it together with the controller.

Machine learning and control theory have contributed to each other in various ways. Function approximators such as kernel methods~\citep{pillonetto2014kernel} and neural networks~\citep{chu1990neural} have been successfully used for system identification. Moreover, ideas from information theory and statistics have been adopted for model selection in system identification~\citep{solowjow2018minimum,mehrjou2018local}. Moreover, recent advances in implicit generative models have proven effective in filtering applications where the states of the system are inferred from noisy observations~\citep{mehrjou2018deep}. The development of automatic differentiation packages for machine learning has opened up a new approach to stability analysis and estimation of functions of interest in control theory such as Hamiltonian~\cite{greydanus2019hamiltonian} or Lyapunov function~\cite{mehrjou2019deep}.

The analytical tools from control theory have helped to analyze the dynamics of machine learning algorithms. For example, the algorithms such as generative adversarial networks~\citep{goodfellow2014generative} and $n$-player cooperative games have been analyzed as dynamical systems to obtain some insight into the dynamics of learning~\citep{mehrjou2018nonstationary,mehrjou2018analysis,mescheder2017numerics,mehrjou2019kernel} and the emergence of the playing strategies~\citep{balduzzi2018mechanics}. In the present work, analytical tools from control theory are employed to study the strength of the learning signal during the course of training.

The popularity of multi-player learning algorithms in machine learning after Generative Adversarial Networks~\citep{goodfellow2014generative} have been inspired in areas of control theory where such multi-player setting can be defined.~\cite{mehrjou2020learning} have proposed a two-player setting where one player learns the Lyapunov function and the region of attraction while the other player learns the governing equations of the system. These players inform each other and their collaboration gives a faster convergence compared with learning each of them separately. The present work can be seen as a two-player algorithm where one player learns the Lyapunov function and the stability region while the other player improves the controller. The information exchange between the two players is the key factor of the proposed algorithm.

As the Lyapunov function is one of the most important objects in control theory, machine learning tools have been employed to learn it for nonlinear systems~\citep{richards2018lyapunov}. A given Lyapunov function is used in~\citep{berkenkamp2016safe, berkenkamp2017safe} to derive a controller with a safety guarantee.~\cite{chang2019neural} solve an optimization problem at every stage to find the states that violate the Lyapunov condition. Our work proposes improvements in various directions. We improve the Lyapunov learning algorithm of~\cite{richards2018lyapunov} by a theoretically motivated method and show that~\emph{hyperbolicity} is a needed feature for nonlinear systems whose controllers are learned iteratively. Moreover,~\cite{berkenkamp2016safe, berkenkamp2017safe} assume the Lyapunov function is given while our work co-learns the Lyapunov function and the controller.~\cite{chang2019neural}'s algorithm needs to be solved until the end that can be too costly as a global optimization problem has to be solved multiple times. However, our work is a growing method that improves the controller while the system is in action.

We bring together various tools from control theory and machine learning to build an iterative algorithm that redesigns both controller and the Lyapunov function to enlarge the stability region. Our contributions are as follows: {\bf 1)} Improving the learning of the Lyapunov function, {\bf 2)} Interlacing the Lyapunov learning with the policy update to enlarge the RoA iteratively, and {\bf 3)} Providing theoretical results for the tractable class of systems and analyzing the learning signal. \Cref{sec:preliminaries} goes over the preliminary materials and \Cref{sec:problem_statement} describe the problem. The proposed algorithm and its theoretical discussions are presented in~\Cref{sec:method}. Finally, empirical evaluation comes in~\Cref{sec:experiments}, followed by related work and conclusion in~\Cref{sec:related_works}. To stay within the space limit, some proofs, theoretical discussions, and extended experimental results are provided in the Appendix.

\section{Preliminaries}

\label{sec:preliminaries}
Here we present the definitions, notations, and theoretical results that are used in later sections.

\noindent{\bf\em System:} A discrete-time\footnote{Including discretized continuous-time systems} time-invariant disturbance-free nonlinear dynamical system is described by
\begin{equation}
    \label{eq:discrete_system}
    \xb_{k+1} = f(\xb_k, \ub_k),
\end{equation}
where $k\in\ZZ$ is the discrete time index, $\xb_k\in\Xcal\subseteq \RR^d$ and $\ub_k\in \Ucal \subseteq \RR^p$ are the state and control signals. We consider a fully observable regime, where the states are available to a feedback controller, i.e., $\ub_k=\pi(\xb_k)$, and $\pi$ is the feedback law or policy. Hence,~\labelcref{eq:discrete_system} can be written as a time-invariant autonomous (TIA) system $\xb_{k+1}=f_\pi(\xb_k)$ where $f_\pi$ is the time-independent \emph{dynamics} function. By assuming Lipschitz continuity for $f$ and $\pi$, a unique solution to this system for every initial state exists that is captured by the so-called ~\emph{flow} function $\Phi(\xb, \cdot):\ZZ\to\Xcal$, with $\Phi(\xb, 0)=\xb$.

\noindent
{\bf\em Sets:} For a TIA system with dynamics function $f$, a state vector $\bar{\xb}$ is called an~\emph{equilibrium point} if it is a fixed-point for $f$, i.e.,~$\bar{x}=f(\bar{x})$. A state vector is called a~\emph{regular point} if it is not an~\emph{equilibrium point}. Let $J_f(\xb)$ be the Jacobian of $f$ at $\xb$. If $J_f(\xb)$ has no eigenvalue with modulus one, $\xb$ is called a~\emph{hyperbolic} equilibrium point. A hyperbolic equilibrium point is asymptotically stable when the eigenvalues of its corresponding Jacobian have modulus less than one; otherwise, it is an~\emph{unstable} equilibrium point. A system whose all equilibrium points are hyperbolic is called a hyperbolic system. A set $M\subseteq\Xcal$ is called an~\emph{invariant} set, if $f(M)=M$, i.e., every trajectory starting in $M$ remains in $M$, for $k\in\ZZ$. A point $\pb\in\RR^d$ is said to be in the $\omega$-limit set (or $\alpha$-limit set) of $\xb$, if for every $\epsilon>0$ and $N>0$ ($N<0$), there exists a $k>N$ ($k<N$) such that $\lVert \xb_k - \pb\rVert < \epsilon$. The stable and unstable manifolds of $\bar{\xb}$ are defined as the set of points whose $\omega$-limit ($\alpha$-limit) set is $\bar{\xb}$ and denoted by $W^s$($W^u$).

Both $W^s$ and $W^u$ are proved to be invariant sets~\citep{palis2012geometric}.

\noindent
{\bf\em Stability:} A fixed-point $\bar{\xb}$ is said to be an asymptotically stable equilibrium, if $\lim_{k\to\infty}\xb_k=\bar{\xb}$. Nonlinear systems often have a local stability region (RoA) that is defined for the stable equilibrium $\bar{\xb}$ as $\Rcal^{\bar{\xb}}= \{\xb\in\Xcal: \lim_{k\to\infty} \Phi(\xb, k)=\bar{\xb}\}$. Topologically speaking, when $f$ is continuous, $\Rcal^{\bar{\xb}}$ is an open, invariant set (see~\citealp{chiang2015stability} for exact definitions). The stability boundary $\partial \Rcal^{\bar{\xb}}$ is a closed positively invariant set and is of dimension $n-1$, if $\Rcal^{\bar{\xb}}$ is not dense in $\RR^n$.

\paragraph{Lyapunov stability.} Let $f$ be a locally Lipschitz continuous dynamics function with an equilibrium point at the origin $\bar{\xb}=\mathbf{0}$. Suppose there exists a locally Lipschitz continuous function $V=\Xcal\to\RR$ and a domain $\Dcal\subseteq\Xcal$, such that
\begin{align}
    \label{eq:lyapunov_conditions}
    &V(\mathbf{0})=0\quad\text{and}\quad V(\xb) > 0\quad&\forall \xb\in\Dcal\backslash \{\mathbf{0}\}\\
    &\Delta V(\xb):=V(f(\xb)) - V(\xb)<0\quad&\forall\xb\in\Dcal\backslash \{\mathbf{0}\}\label{eq:lyapunov_decrease_condition}
\end{align}
Then, $\bar{\xb}$ is asymptotically stable and $V$ is a~\emph{Lyapunov function (lf)}. The domain $\Dcal$ in which~\eqref{eq:lyapunov_decrease_condition} is satisfied is called the~\emph{Lyapunov decrease region}. Every sublevel set $\Scal_c(V) = \{\xb\in\Xcal: V(\xb)<c\}$, for $c\in\RR_+$, that is contained within $\Dcal$ is invariant under the dynamics $f$.

%% file: problem_statement_arxiv.tex
\section{Problem Statement}
\label{sec:problem_statement}

We consider a discrete-time TIA system as in~\labelcref{eq:discrete_system}, where the control signal is produced by a feedback-controller $\pi(\cdot;\psi):\Xcal \to \Ucal$ parameterized by $\psi$. Therefore, the closed-loop dynamics denoted by $f_\pi$ is a functional of the controller and is consequently parameterized by $\psi$ as $\xb_{k+1}=f(\xb_k,\pi(\xb_k;\psi))=f_\pi(\xb_k;\psi)$. Without loss of generality, we assume that the equilibrium point of interest is located at the origin $\bar{\xb}=\mathbf{0}$. The policy $\pi$ induces a RoA around the equilibrium point denoted by $\Rcal_\pi$.\footnote{We drop $\bar{\xb}$ from the superscript of $\Rcal_\pi^{\bar{\xb}}$, since we always assume $\bar{\xb}=\mathbf{0}$, unless otherwise stated.}

Each control task can be broken down into two subtasks: {\bf 1)} Controller synthesis and {\bf 2)} Closed-loop response evaluation. The controller is designed to optimize some measure of performance. In this work, the performance measure is the size of the stability region. 

We endow the state space $\Xcal$ with a measure $\mu$ to obtain the measure space $(\Xcal, \Bcal(\Xcal), \mu)$ with Borel sigma-algebra $\Bcal(\Xcal)$. To prevent pathological cases, we assume $\Xcal$ to be compact with $\mu(\Xcal)<\mu_\infty<\infty$. Let $\Pi=\{\pi:\Xcal\to\Ucal: \pi\in C^1\;\text{and bounded}\}$ be the set of all functions from which the policy is chosen. The goal is to find a member $\pi^*$ of the equivalence class of optimal policies $\Pi^*\subseteq\Pi$, where $\Pi^*$ is defined as $\Pi^*=\{\pi^*\in \Pi: \mu(\Rcal_{\pi^*}) = \max_{\pi\in\Pi} \mu(\Rcal_\pi)\}$. 

The main challenge in this optimization problem is the fact that there is no analytical or straightforward way to infer how changing $\pi$ changes $\mu(\Rcal_\pi)$. If there exists a differentiable map from $\pi$ to $\mu(\Rcal_\pi)$, one could locally increase the stability region by perturbing the policy in the direction of $\partial \mu(\Rcal_\pi)/\partial \pi$. However, except for extremely simple systems, such a map cannot be derived analytically. In this work, we construct a bridge between these two objects by an auxiliary function, which is an evolving Lyapunov function that is learned alongside the policy. 

Our goal is to construct a sequence of policies $(\pi_1, \pi_2, \ldots)$ that gives a sequence of RoAs $(\Rcal_{\pi_1}, \Rcal_{\pi_2}, \ldots)$, such that $\Rcal_{\pi_n}\xrightarrow{\mu} \bar{\Rcal}$ as $n\to \infty$, where $\bar{\Rcal}\subseteq\Dcal$ is the largest achievable RoA that is constrained by the physical limitations of the system (see~\Cref{sec:control_lyapunov_function} for the characterization of $\bar{\Rcal}$ using the concept of control Lyapunov function). To achieve this goal, we need to address two challenges: {\bf 1)} Approximating $\Rcal_{\pi_n}$ for a fixed $\pi_n$ and {\bf 2)} Using $(\Rcal_{\pi_n}, f_{\pi_{n}})$ to find $\pi_{n+1}$. Next section, explains our proposed method to address these two challenges.

%% file: method_arxiv.tex
\section{Proposed Method}
\label{sec:method}

\indent
For a function $g$, let's define its sublevel set with level value $a$ as $\Scal_a(g):=\{\xb\in\RR^d: g(\xb)<a\}$. The index $n\in\NN\cup \{0\}$ refers to a phase of the algorithm. At phase $n$, let $\Rcal_{\pi_n}$ be the RoA of the closed-loop system~\labelcref{eq:discrete_system} that is induced by the state-feedback policy $\pi_n$. It can be shown that there exists an optimal Lyapunov function $V_{\pi_n}$ with a level value $c_n$, such that $\Rcal_{\pi_n}=S_{c_n}(V_{\pi_n})$~\citep{vannelli1985maximal}. Therefore, the information of $\Rcal_{\pi_n}$ is encoded in $(V_{\pi_n}, c_n)$. 

Starting with a conservative controller (e.g., a quadratic controller for the linearized system), our method inductively constructs a sequence of policies $(\pi_1,\pi_2,\ldots)$ that eventually converges to a policy with maximal achievable RoA. See~\Cref{fig:cartoon_and_graphical_RoA}(Left) for an overall schematic of the algorithm.

Each step of this inductive process is called a~\emph{phase} of the algorithm. Each phase consists of two sub-phases: {\bf 1)} Learning the Lyapunov function and the RoA corresponding to the policy of that phase {\bf 2)} Updating the policy to enlarge the RoA. The RoA estimation sub-phase finds a Lyapunov function $V_{\pi_n}$ and level value $c_n$, such that $\Scal_{c_n}(V_{\pi_{n}}) = \Rcal_{\pi_{n}}$. Then, the policy update sub-phase learns a new policy $\pi_{n+1}$, such that $\Scal_{c_n}(V_{\pi_{n}})\subseteq\Rcal_{\pi_{n+1}}$. These conditions need to be satisfied for every $n$ that consequently limits the class of treatable systems. Before delving into the algorithmic implementation,~\Cref{sec:theoretical_discussion} provides necesary theoretical insights into this favorable class of systems along with other theoretical considertations for the proposed multi-phase growing algorithm.

\subsection{Theoretical Discussion}
\label{sec:theoretical_discussion}
We start with the assumptions that are necessary for the practical applicability of the method. Then, we show for which class of systems these assumptions are satisfied.

\begin{assumption}
    \label[assumption]{assumption:continuously_changing_roa_with_policies}
    Let $R:\Pi\to 2^\Xcal$ be a set-valued function defined as $R(\pi) = \Rcal_\pi$. Let $d_\Pi:\Pi\times\Pi\to \RR_+$ and $d_\Xcal:2^\Xcal \times 2^\Xcal\to \RR_+$ be some specified metrics in the space of policies and the space of all subsets of the state space, respectively. Then, the map $R$ is assumed to be continuous with respect to the topologies induced by the metrics $d_\pi$ and $d_\Xcal$.
\end{assumption}

This assumption indicates that a small change in the policy leads to a small change in the RoA that it induces. Formally speaking, let $\Rcal_{\pi_{n}}$ and $\pi_n$ be the RoA and the policy of phase $n$.~\Cref{assumption:continuously_changing_roa_with_policies} states that for every $\epsilon>0$, one can choose small enough $\delta>0$ such that if $\pi$ satisfies $d_\Pi(\pi_n, \pi)<\delta$, then $\mu(\Rcal_{\pi_{n}}\triangle \Rcal_{\pi})<\epsilon$. Next, we show that hyperbolic systems fulfill this assumption\footnote{The operator $\triangle$ shows the symmetric difference between two sets.}. 

\renewcommand*{\proofname}{Proof sketch}
\begin{theorem}[Persistance of the stability boundary with variations in the policy]
\label{thm:persistance_of_roa_with_policy_variations}
Consider the closed-loop hyperbolic system $\xb_{k+1}=f_\pi(\xb_k)$ with policy $\pi$. For a certain policy $\pi=\tilde{\pi}$, let $(\bar{\xb}_{\tilde{\pi}}$, $\Rcal_{\tilde{\pi}})$ be an asymptotically stable equilibrium point and its corresponding RoA. Then, for every $\epsilon>0$, there exists $\delta>0$ such that for every $\pi'$ with $d_\Pi(\tilde{\pi}, \pi')<\delta$, we have $\mu(\Rcal_{\tilde{\pi}} \triangle \Rcal_{\pi'})<\epsilon$.
\end{theorem}

\begin{proof}
The first step to prove this result is to characterize the RoA in terms of the properties of the dynamics function $f_\pi$. As $\Rcal_{\tilde{\pi}}$ is characterized by its boundary, we focus our attention on the stability boundary $\partial\Rcal_{\tilde{\pi}}$. The critical elements\footnote{For theoretical discussion, we focus on critical elements that are equilibrium points. Similar results exist for other types of critical elements such as limit cycles but are left out of this work for brevity} of the dynamical system determine the structure of the stability boundary. Suppose $f_\pi$ is a diffeomorphism and all equilibrium points on $\partial\Rcal_{\tilde{\pi}}$ are hyperbolic. Moreover, let the stable and unstable manifolds of the equilibrium points on $\partial\Rcal_{\tilde{\pi}}$ intersect transversally\footnote{Roughly speaking, the manifolds intersect in a generic way.}. Finally, assume that every trajectory on $\partial\Rcal_{\tilde{\pi}}$ approaches one of the equilibrium points. If $\{\xb_1, \xb_2, \ldots\}$ are the hyperbolic equilibrium points on $\partial\Rcal_{\tilde{\pi}}$, $\partial\Rcal_{\tilde{\pi}}$ is completely characterized by
{\small
\begin{equation}
\partial \Rcal_{\tilde{\pi}} = \cup_i W^s(\xb_i).
\end{equation}
}
See Theorems 9-11 in~\citet{chiang2015stability} for the detailed proof. Hence, to show the persistance of $\partial \Rcal_{\tilde{\pi}}$, it is enough to show the persistance of the equilibrium points that live on $\partial \Rcal_{\tilde{\pi}}$ and the persistance of their stability condition. As a result of the continuity of $f(\xb, \pi(\xb))$ and $\pi(\xb)$ w.r.t.~their arguments, implicit function theorem guarantees that small perturbations to $\pi$ cause small changes in the hyperbolic equilibrium points~\citep{krantz2012implicit}. If $\xb^*_{\tilde{\pi}}$ is a hyperbolic equilibrium point of $\xb_{k+1}=f_\pi(\xb_k)$ for the policy $\pi=\tilde{\pi}$, there exists a $\delta>0$ and a neighborhood $U$ of $\xb^*_{\tilde{\pi}}$ that contains a unique hyperbolic equilibrium point $\xb^*_{\pi'}$ for every $\pi'\in \{\pi\in\Pi: d_\Pi(\pi, \tilde{\pi})<\delta\}$. Similarly, the continuity of the eigenvalues of $J_{f_{\pi}}(\xb_\pi)$ w.r.t.~$\pi$ affirms that the perturbed equilibrium point $\xb^*_{\pi'}$ has the same stability condition as $\xb^*_{\tilde{\pi}}$.

As stated above, $\partial\Rcal_{\tilde{\pi}}$ is characterized by the stable manifolds of equilibrium points that live on it. It was also shown that the hyperbolic equilibrium points on $\partial \Rcal_{\tilde{\pi}}$ together with their stable and unstable manifolds change continuously with $\tilde{\pi}$. This results in a continuous change of $\partial \Rcal_{\tilde{\pi}}$, and consequently $\Rcal_{\tilde{\pi}}$, w.r.t.~small variations in $\tilde{\pi}$.
\end{proof}

This result suggests restricting the change of the policy at the policy update sub-phase of each growing phase of the algorithm to bound the change of its induced RoA. For parametric policies such as neural networks, this is achieved by cropping parameter values after every phase of training.

\newcommand{\sfsize}{0.12\textwidth}
\begin{figure}[t!]
  \begin{subfigure}[t]{0.35\textwidth}
    \centering
    \includegraphics[width=\textwidth]{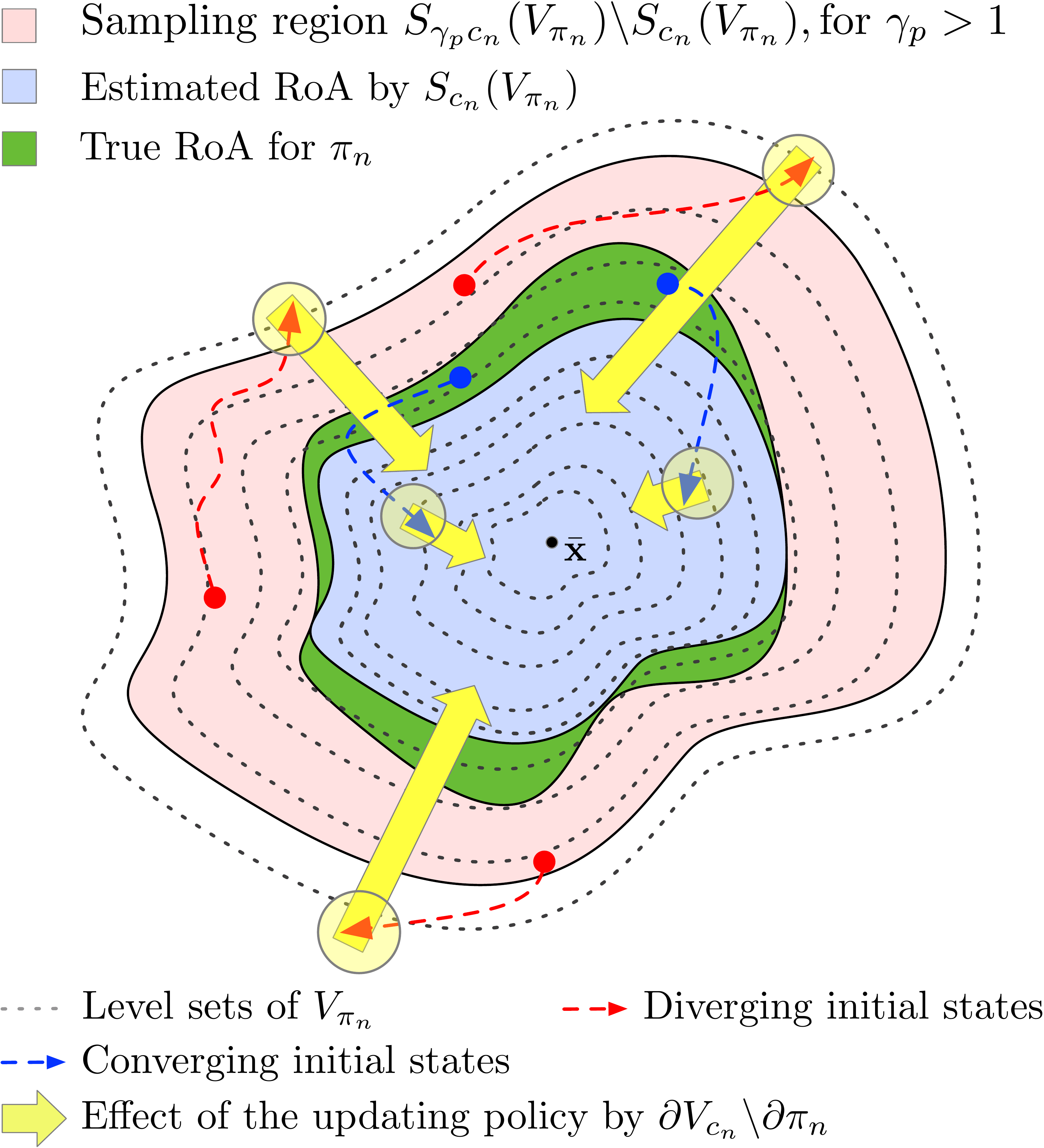}
  \end{subfigure}
  \begin{subfigure}[t]{0.65\textwidth}
    \centering
    \vspace{-33ex}
    \input{./figs/increasing_roas/tex_files/increasing_nested_roa_figure_stage_1.tex}\\
    \input{./figs/increasing_roas/tex_files/increasing_nested_roa_figure_stage_2.tex}\\
    \input{./figs/increasing_roas/tex_files/increasing_nested_roa_figure_stage_3.tex}\\
    \input{./figs/increasing_roas/tex_files/increasing_nested_roa_figure_stage_4.tex}\\
  \end{subfigure}
  \caption{(Left) Illustration of the policy update sub-phase. Given the estimated RoA, the policy update sub-phase (yellow arrows) tries to pull the diverging trajectories towards the level sets of the estimated Lyapunov function that reside inside the RoA. (Right) Visualizing the true ROA which is enlarged by the improved policy and is chased by a learned Lyapunov function. Green boundary: True RoA, Blue: $S_{c_n}(V_{\pi_n})$, Pink: $S_{\gamma c_n}(V_{\pi_n})$ for $\gamma=4$. The pink area shows the region from which the samples outside the estimated RoA is taken for both estimating the RoA and updating the policy. \vspace{-4ex}}\label{fig:cartoon_and_graphical_RoA}
\end{figure}

\subsection{Algorithm}
\label{sec:algorithm}
The implementation of the algorithm that was outlined in~\Cref{sec:method} comes in the following. Notice that the algorithm is a multi-phase inductive method. Hence, we only present two sub-phases of a single phase. The full algorithm consists of iterations over this phase.

\paragraph{RoA Estimation Sub-Phase: Learn $\Rcal_{\pi_{n}}$ from $(\Rcal_{\pi_{n-1}}, f_{\pi_{n}})$.}
This sub-phase takes the current policy ${\pi_{n}}$ and the previous RoA estimate $\Scal_{c_{n-1}}(V_{\pi_{n-1}})$ and outputs $\Scal_{c_{n}}(V_{\pi_{n}})$ that approximates $\Rcal_{\pi_{n}}$.
For this sub-phase, we improve the growing algorithm of~\citep{richards2018lyapunov} to learn an inner estimate of RoA. Our RoA estimation algorithm takes advantage of the theoretical results of~\citep{chiang1989stability} to improve the stability of training. The proposed loss function is more robust against local minima observed by~\cite{richards2018lyapunov} as is empirically shown in~\Cref{sec:experiments}.

The learning algorithm is verbally described here and the pseudo-code can be seen in~\Cref{alg:roa_estimation} in the Appendix. It starts with an initial conservative estimate of RoA and grows it during multiple iterations (phases). A neural network, as a universal function approximator, is employed to estimate a Lyapunov function for the system. At each iteration, the Lyapunov function is improved such that one of its sublevel sets gives a better estimate of the RoA. For this purpose, a fixed number of initial states are taken from a gap surrounding the current estimate of the RoA. The initial states are integrated forward in time by the closed-loop dynamics to produce the solution $\Phi(\xb, k)$ for each initial state $\xb$. The final state of each trajectory determines the stability of its corresponding initial state. We train the Lyapunov function as a classifier. The current estimate of the RoA produced by the current Lyapunov function is called the largest stable sublevel set. The initial states are labelled~\emph{stable} if their trajectory enters the largest stable sublevel set and are labelled~\emph{unstable} otherwise. We denote these two sets of initial states as $\XX^\text{IN}$ and $\XX^\text{OUT}$ respectively and minimize the loss function 

{\footnotesize
\begin{equation}
  \label{eq:roa_update_objective}
  \Lcal(V) = \sum_{\xb\in\XX^\text{IN}}[V(\xb) - \bar{c}] - \sum_{\xb\in\XX^\text{OUT}}[V(\xb) - \bar{c}] +\lambda_\text{RoA}\sum_{\xb\in\XX^\text{IN}} \Delta  V(\xb) + \lambda_\text{monot}\sum_{\xb\in\XX^\text{IN}} [V(\xb) - V_{\pi_{n-1}}(f_{\pi_{n-1}}(\xb))]^2
\end{equation}
}to update the Lyapunov function where $\bar{c}$ is fixed to a constant value ($1$ throughout this work). The idea is to absorb this degree of freedom in $V$ to ease training. The loss function is minimized by automatic differentiation and Stochastic Gradient Descent (SGD) with respect to the parameters of the neural network that realizes $V$. Once it is minimized, a line search is carried out on the level value to obtain $c_n$ such that $\Scal_{c_n}(V_{\pi_n})$ does not exceed the true RoA. The same process repeats for a certain number of iterations until a good inner estimate of the RoA is achieved. An important component in this process is the size of the gap around the largest stable sublevel set of each iteration of the algorithm. This gap is produced by $\Gcal_{n} = \Scal_{\gamma_r c_n}(V_{\pi_{n}}) \backslash \Scal_{c_n}(V_{\pi_{n}})$ with $\gamma_r > 1$. The size of this gap is controlled by $\gamma_r$. Larger values of $\gamma_r$ gives a faster convergence but less stable training.

To give an intuitive idea of the terms in~\Cref{eq:roa_update_objective}, the first two terms encourage the Lyapunov function to change in a way that its sublevel set $S_{\bar{c}}(V)$ includes the stable initial states $\XX^\text{IN}$ and excludes the unstable initial states $\XX^\text{OUT}$. The third term weighted by $\lambda_{\rm RoA}$ encourages the negative definiteness of $\Delta V$ on the RoA. The last term is inspired by the constructive method of~\citep{chiang1989stability} that accelerates capturing the entire RoA (see~\Cref{sec:improved_roa_estimation_constructive_term}). To satisfy Lyapunov conditions, $V$ needs to be positive definite on its domain. Rather than treating this condition in the loss function, we encode it in the architecture of the neural network using the construction proposed by~\cite{richards2018lyapunov}. Check~\Cref{sec:lyapunov_architecture} for the detailed description of the architecture.

\paragraph{Policy Update Sub-Phase: Learn $\pi_{n+1}$ from $(\Rcal_{\pi_n}, f_{\pi_{n}})$:}
This sub-phase of the algorithm uses the estimated Lyapunov function $V_{\pi_n}$ to update the policy so that the new policy gives rise to a larger RoA. The idea is to change the policy in a way that the unstable trajectories starting from around the current RoA enter the RoA and become stable. Let $\Dcal\subseteq\Xcal$ be the working domain of the system around the equilibrium $\bar{\xb}$. Given a hypothesis class of feasible policies $\Pi$, only a subset of $\Dcal$ is stabilizable. Assume $\bar{\Bcal}\subseteq\Dcal$ is the largest stabilizable subset with $\mu(\bar{\Bcal}) = \bar{\mu}$. Therefore, an attempt to improve $\pi_n$ amounts to appending points from $\bar{\Bcal}\backslash \Rcal_{\pi_n}$ to $\Rcal_{\pi_n}$. The set $\bar{\Bcal}$ is not fully known in advance but some of its properties can be derived. Especially, for system~\labelcref{eq:discrete_system}, if $f, \pi\in C^\infty$, the maximum stabilizable set $\bar{\Bcal}$ whose measure materializes as $\bar{\mu}$ is compact and connected. Using this theoretical result, if $\Rcal_{\pi_n}\subsetneq\bar{\Bcal}$, the stabilizable states can be chosen from a gap around $\Rcal_{\pi_n}$.

Because the RoA estimation sub-phase estimates it as a sublevel set of a Lyapunov function, i.e. $\Rcal_{\pi_n} \approx S_{c_n}(V_{\pi_n})$), the sampling gap is constructed as $\Gcal_n = S_{\gamma_p c_n}(V_{\pi_n}) \backslash S_{c_n}(V_{\pi_n})$ for a $\gamma_p>1$. To make sure the policy does not destabilize the regions that are already stabilized in the previous phases, the algorithm also samples initial states from within $\Scal_{c_n}(V_{\pi_n})$. All sampled initial states are integrated forward for $L_p$ steps and the policy is updated via minimizing the loss function 

\begin{equation}
  \label{eq:policy_update_objective}
  \Lcal(\pi) = \hspace{-2ex}\sum_{\xb\in\Gcal_n\cup S_{c_n}(V_{\pi_n})}\hspace{-2em}[1_{[V_{\pi_n}(\Phi_\pi(\xb, L_p)) < c_n]} + \lambda_{\rm u}1_{[V_{\pi_n}(\Phi_\pi(\xb, L_p)) > c_n]}] V_{\pi_n}(\Phi_\pi(\xb, L_p)).
\end{equation}

We implement the policy as a differentiable function such as a neural network and use automatic differentiation and SGD for minimization. The parameters of the policy appear in~\labelcref{eq:policy_update_objective} via the end state $\Phi_\pi(\xb, L_p)$ of the closed-loop trajectories. It is clear in~\labelcref{eq:policy_update_objective} that minimizing $\Lcal(\pi)$ with respect to $\pi$ while the Lyapunov function is fixed, pushes the trajectories towards the areas where the Lyapunov function assumes smaller values. This affects both stable and unstable trajectories while its influence on unstable trajectories can be magnified by $\lambda_u>1$. The detailed pseudo-code of this sub-phase can be found in~\Cref{alg:policy_update} in the Appendix. The quality of the training signal is analysed next.

\textit{Theoretical analysis of the learning signal---} The effect of the policy $\pi$ on $\Lcal(\pi)$ passes through the Lyapunov function $V$ as can be seen in~\labelcref{eq:policy_update_objective}. Unlike the conventional Lyapunov redesign method in control theory where the Lyapunov function is fixed, here the Lyapunov function itself is learned by the RoA estimation sub-phase of the algorithm. Hence, an ill-conditioned $V$ can harm the policy update phase. To take a closer look at this issue, we expand the learning signal analytically. Let the policy $\pi$ be a function of the states parameterised by $\psi$. Let $T=L_p$ be the time step of the final state of the trajectory. The learning signal to update the policy is proportional to $\partial \Lcal / \partial \psi$ expanded as

{\small
\begin{equation}
  \frac { \partial \Lcal } { \partial \psi } = \frac { \partial \Lcal  } { \partial \mathbf { x } _ { T } } \sum _ { 1 \leq k \leq T } \left(  \frac { \partial \mathbf { x } _ { T } } { \partial \mathbf { x } _ { k } } \frac { \partial ^ { + } \mathbf { x } _ { k } } { \partial \psi } \right).\label{eq:finaltime_loss_derivative}
\end{equation}
}
by applying the chain rule for differentiation. 

The term $\partial ^ { + } \mathbf { x } _ { k } / \partial \psi$ is the single-step effect of $\psi$ on $\xb_k$ when $\xb_{k-1}$ is fixed. We discuss every term in this equation in the following. Observe that $\partial \Lcal / \partial \xb_T$ is multiplied the summation, i.e., its small value diminishes the entire signal. It can be expanded as

{\small
\begin{equation}\label{eq:dl_dxT}
  \frac{ \partial \Lcal(\xb) }{\partial \xb}|_{\xb = \xb_T} = \frac{\partial \Lcal(V)}{\partial V}|_{V=V(\xb_T)} \frac{\partial V(\xb)}{\partial \xb}|_{\xb = \xb_T}.
\end{equation}}
The first term on the r.h.s. does not vanish as it is $1$ or $\lambda_u$ for the loss function defined by~\Cref{eq:policy_update_objective}. The second term of the r.h.s depends on the slope of $V$ evaluated at the final state of the trajectory. One potential pathological condition occurs for stable and long trajectories. The reason is that $\nabla_\xb V(\xb)$ continuously vanishes at the equilibrium (see~\Cref{lem:dv_dx_at_equilibrium} in the Appendix). Therefore, for long stable trajectories where $\xb_T$ is too close to the equilibrium, the learning signal to update the policy will be too small. Too long trajectories are detrimental for unstable trajectories as well because the states may grow exponentially and cause damages to the system. Hence, the length of the trajectory is an important design parameter that needs special attention when applying our proposed method in practice. The terms inside the sum only depend on the properties of the dynamics, not the Lyapunov function. Specifically, the first term shows how long the system keeps the memory of the past states and the second term shows how sensitive the system is with respect to the controller parameters. A more detailed discussion is deferred to~\Cref{sec:weak_learning_signal}.

Before presenting the performance of the method empirically, it is important to note that our method assumes the model of the system is given. The following remark discusses to what extent the model of the system is needed and this requirement can be relaxed in the future.

\begin{remark}
  \label[remark]{rem:model_based_assumption}
  The RoA estimation phase does not need the model of the system. The system can be launched from sampled initial states and the generated trajectories are all we need in~\labelcref{eq:roa_update_objective}. The policy update phase of the algorithm requires a local estimate of the system to be able to compute $\partial V/\partial \pi$. The locality of the model is inversely proportional to the length of the trajectory $L_p$ in~\labelcref{eq:policy_update_objective}. A detailed theoretical discussion on this point is deferred to~\Cref{sec:model_based_assumtion}.
\end{remark}

%% file: figs/increasing_roas/tex_files/increasing_nested_roa_figure_stage_1.tex
\begin{subfigure}[t]{\sfsize}
    \centering
    \includegraphics[width=1\textwidth]{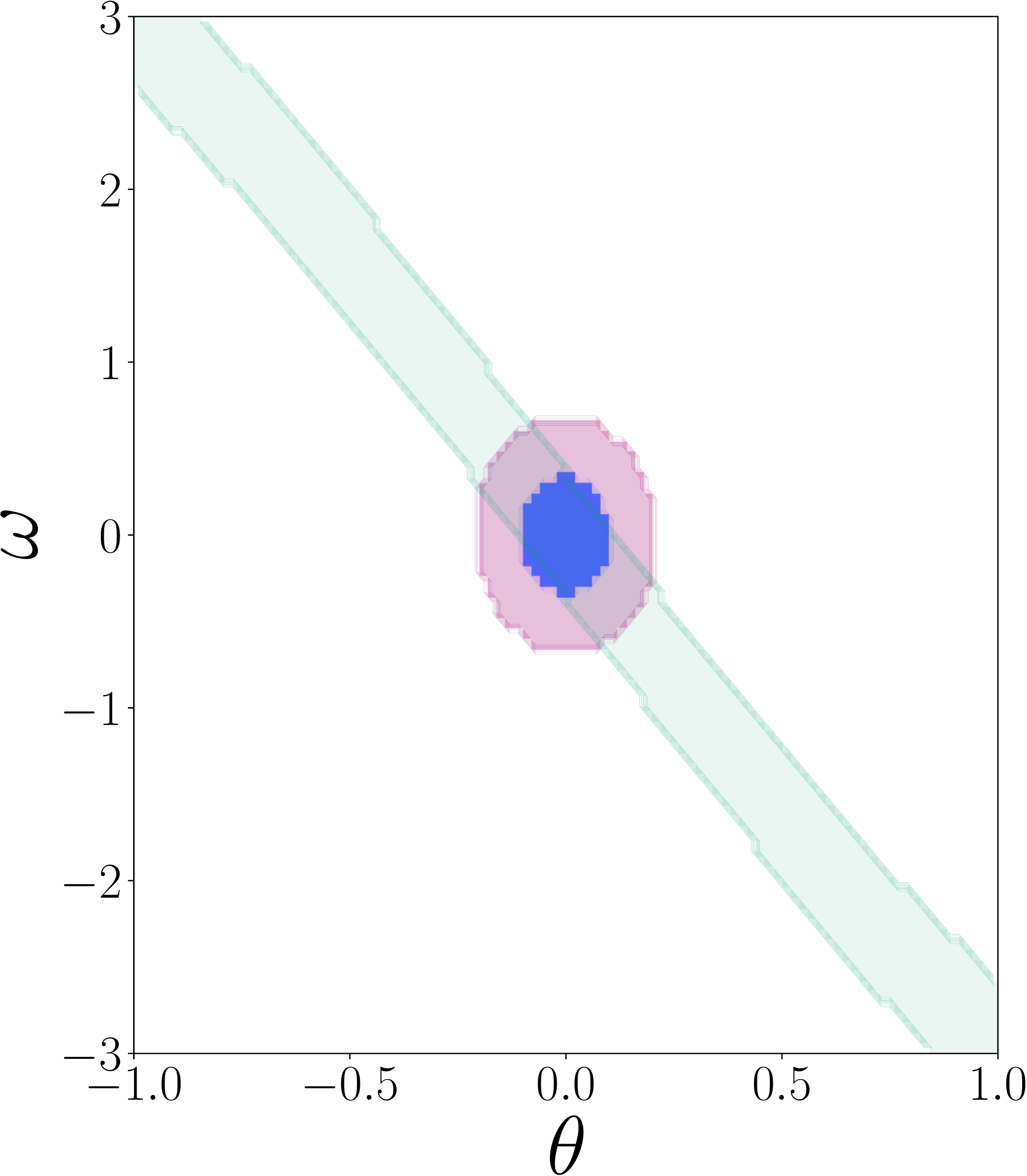}
\end{subfigure}
\begin{subfigure}[t]{\sfsize}
    \centering
    \includegraphics[width=1\textwidth]{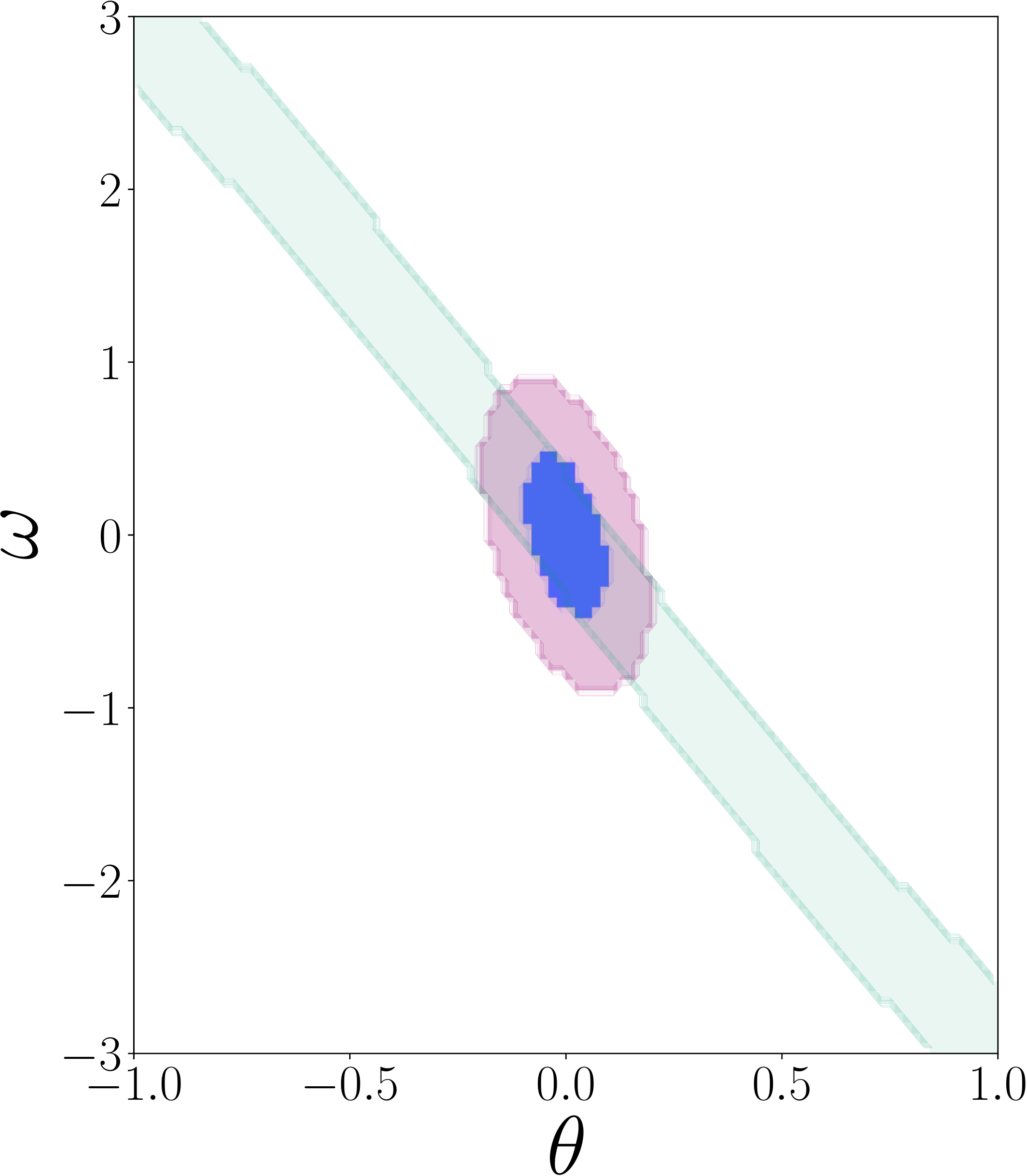}
\end{subfigure}
\begin{subfigure}[t]{\sfsize}
    \centering
    \includegraphics[width=1\textwidth]{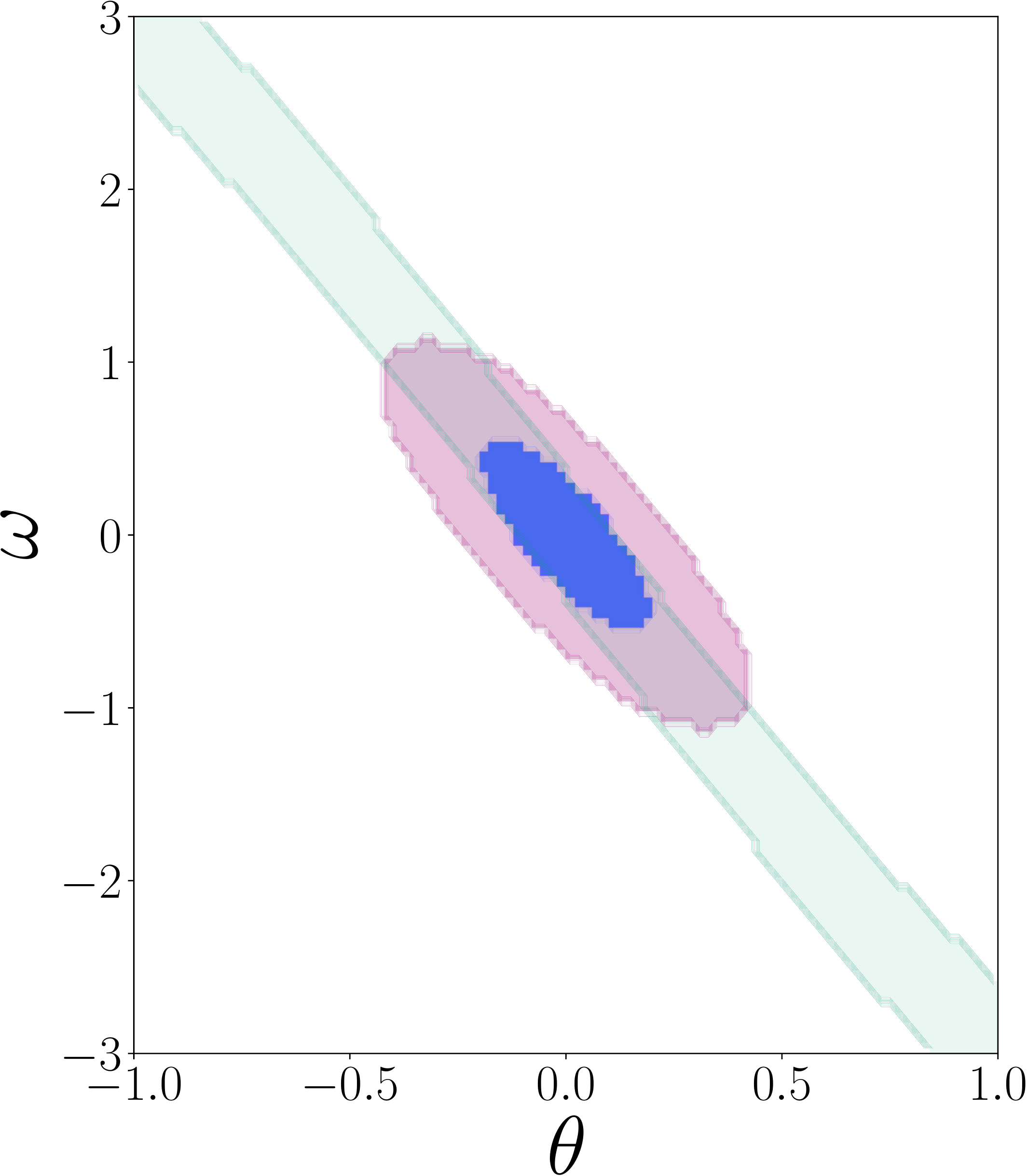}
\end{subfigure}
\begin{subfigure}[t]{\sfsize}
    \centering
    \includegraphics[width=1\textwidth]{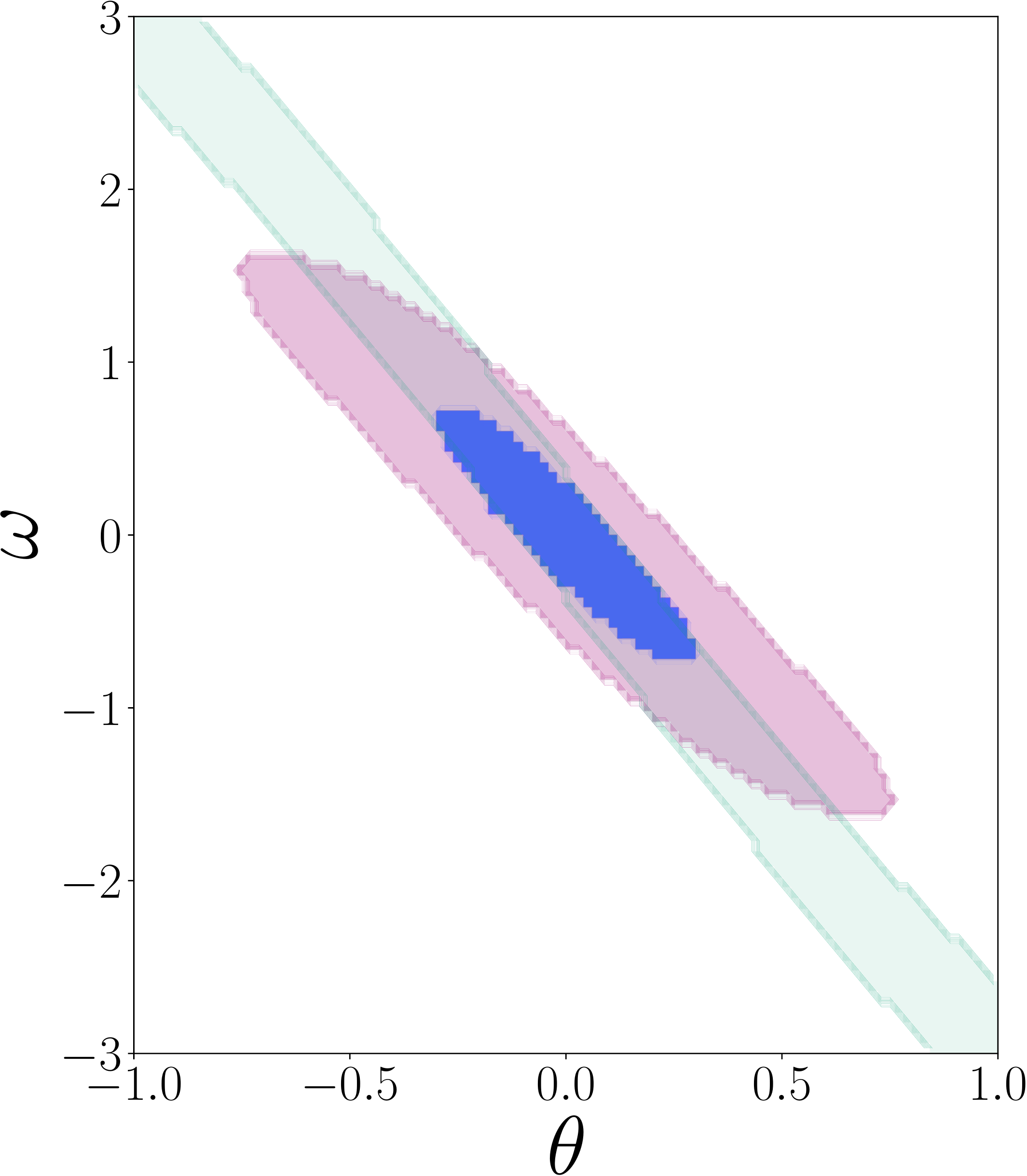}
\end{subfigure}
\begin{subfigure}[t]{\sfsize}
    \centering
    \includegraphics[width=1\textwidth]{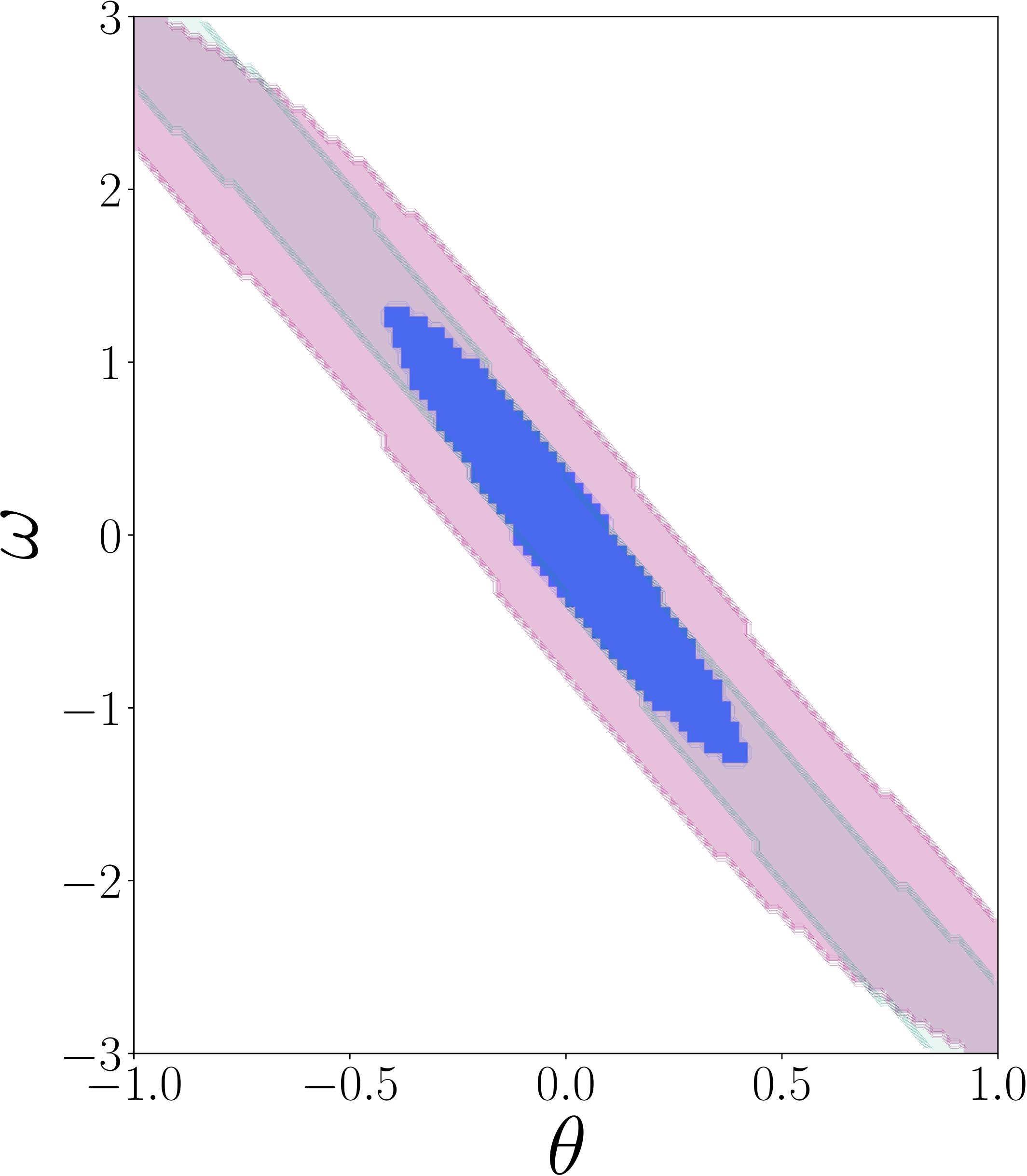}
\end{subfigure}
\begin{subfigure}[t]{\sfsize}
    \centering
    \includegraphics[width=1\textwidth]{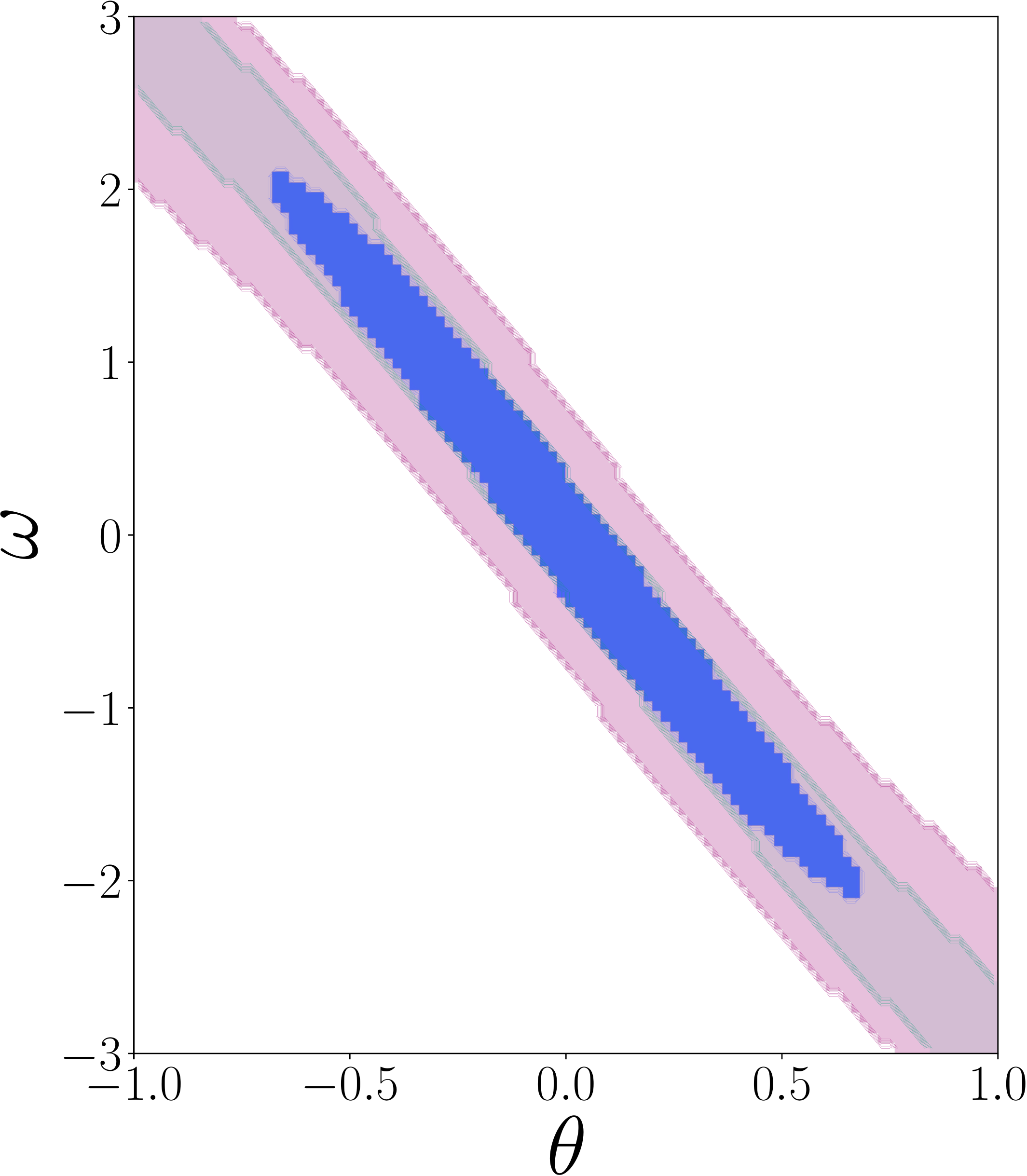}
\end{subfigure}

%% file: figs/increasing_roas/tex_files/increasing_nested_roa_figure_stage_2.tex
    \begin{subfigure}[t]{\sfsize}
        \centering
        \includegraphics[width=1\textwidth]{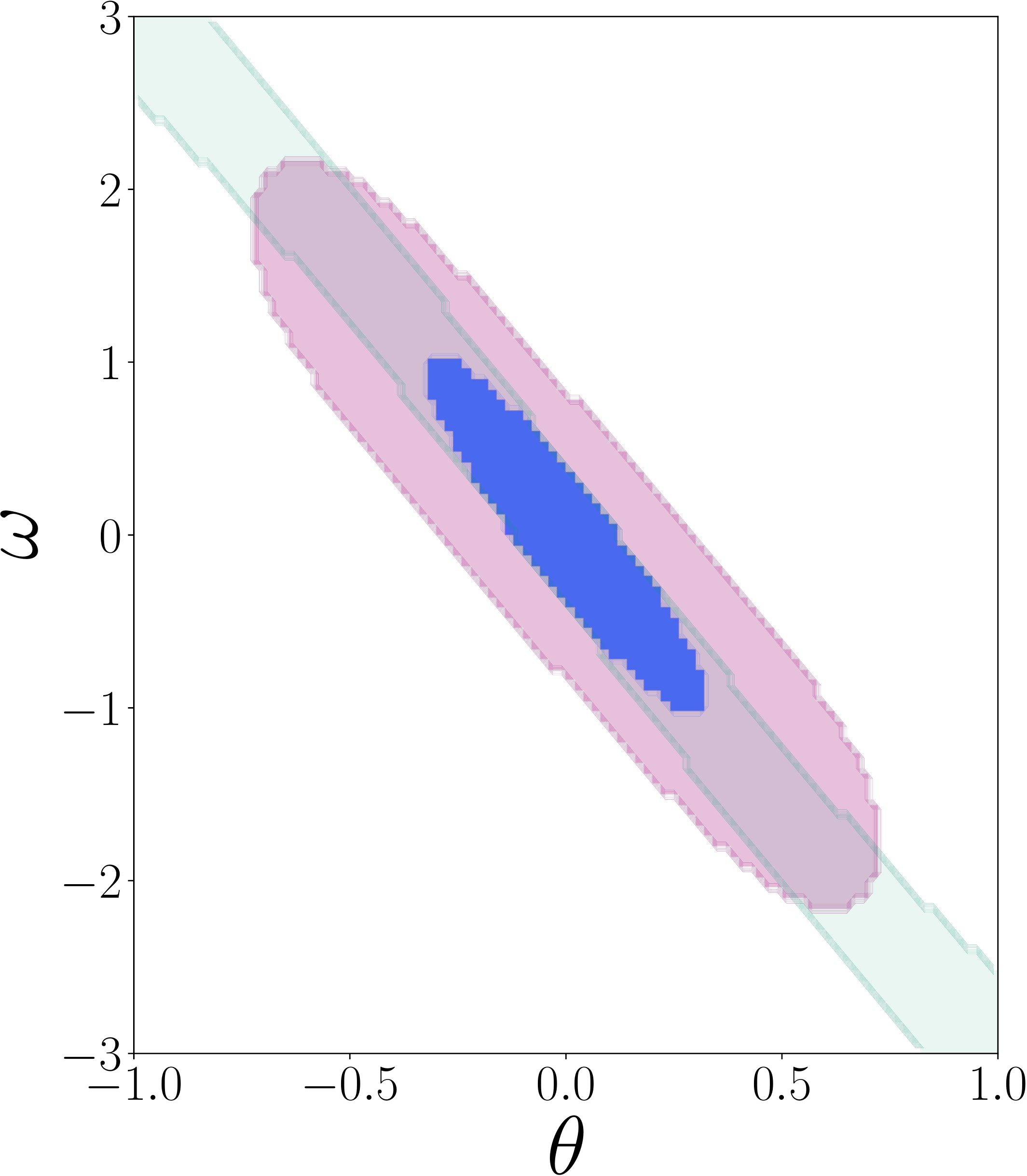}
    \end{subfigure}
    \begin{subfigure}[t]{\sfsize}
        \centering
        \includegraphics[width=1\textwidth]{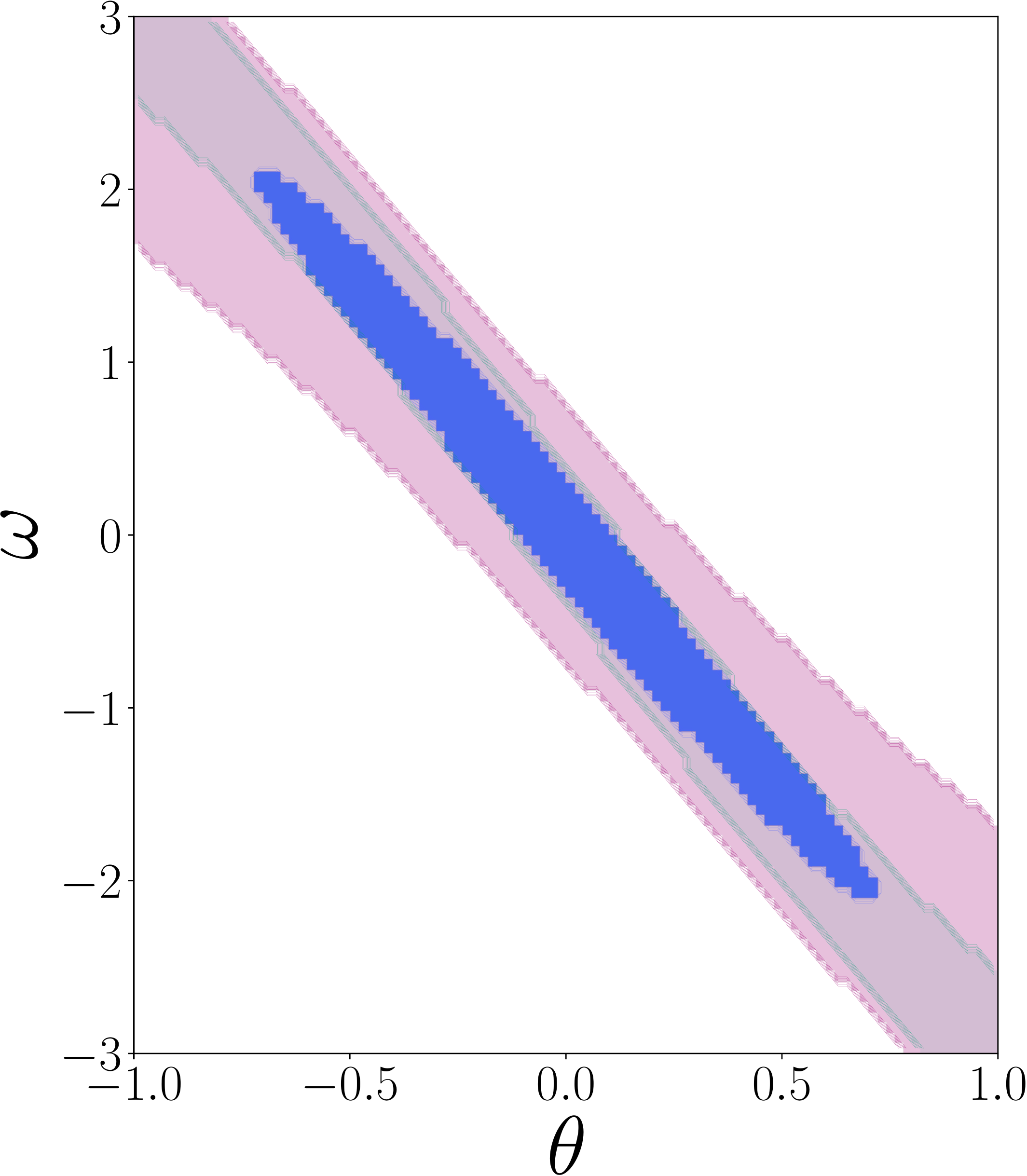}
    \end{subfigure}
    \begin{subfigure}[t]{\sfsize}
        \centering
        \includegraphics[width=1\textwidth]{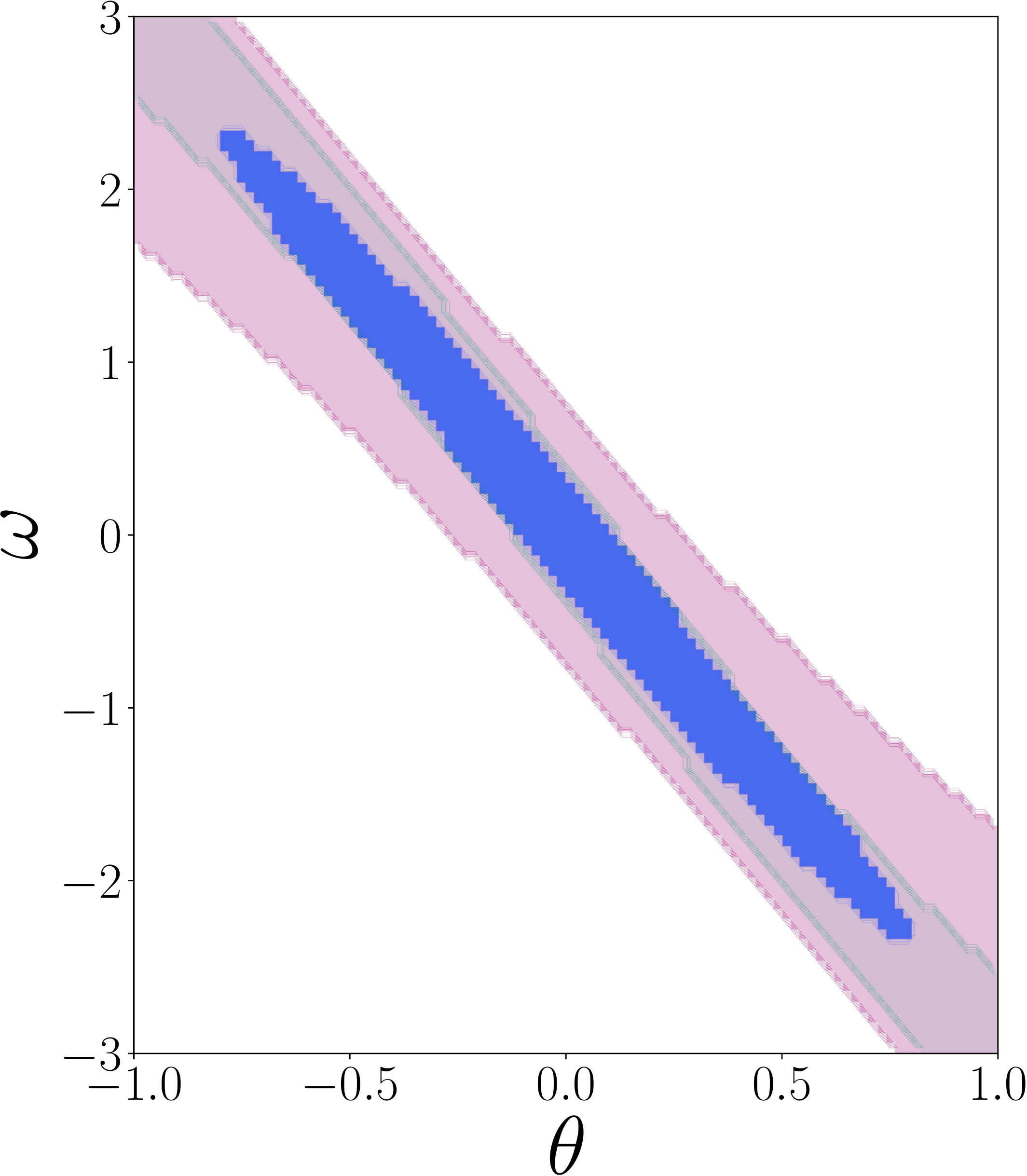}
    \end{subfigure}
    \begin{subfigure}[t]{\sfsize}
        \centering
        \includegraphics[width=1\textwidth]{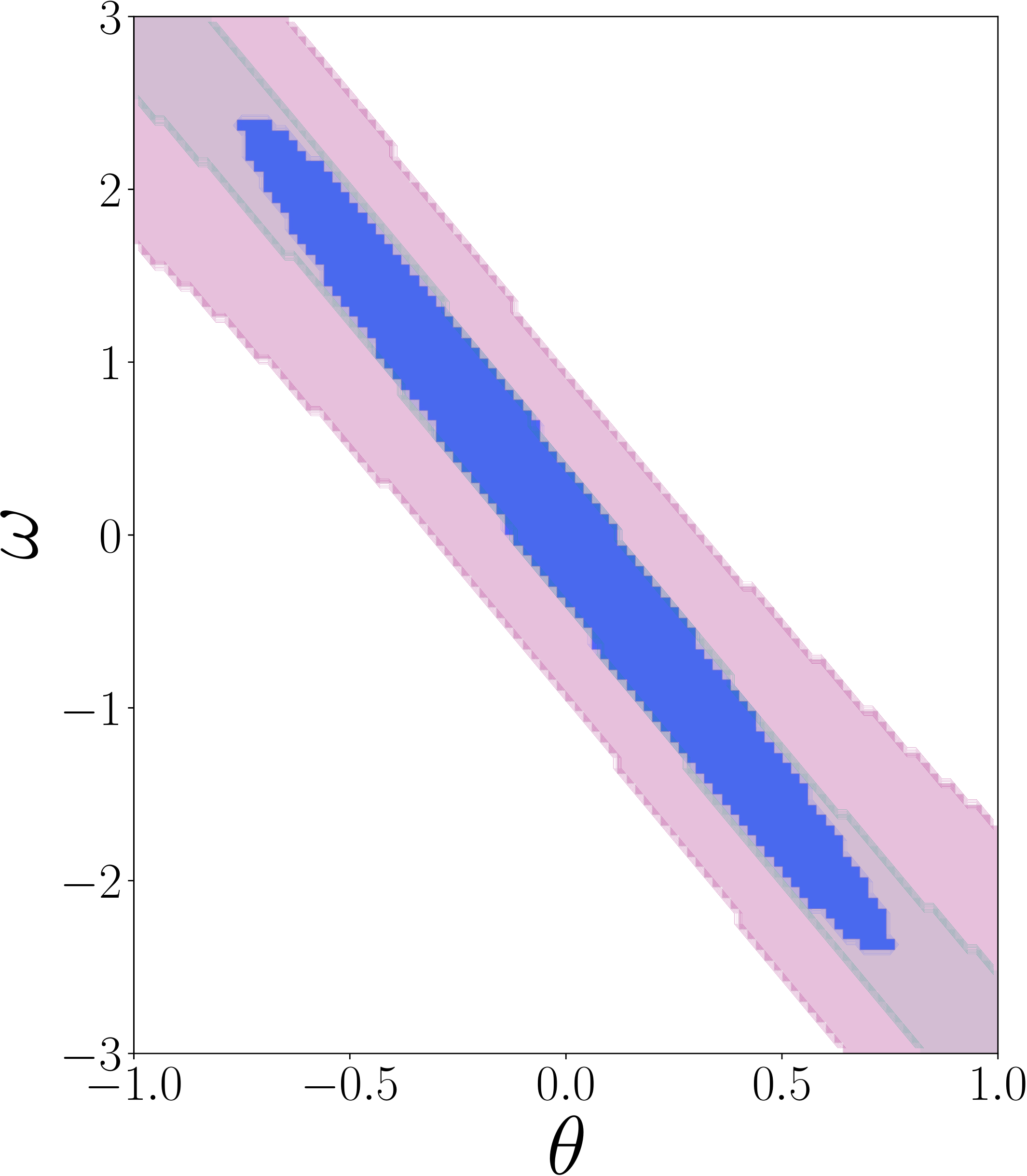}
    \end{subfigure}
    \begin{subfigure}[t]{\sfsize}
        \centering
        \includegraphics[width=1\textwidth]{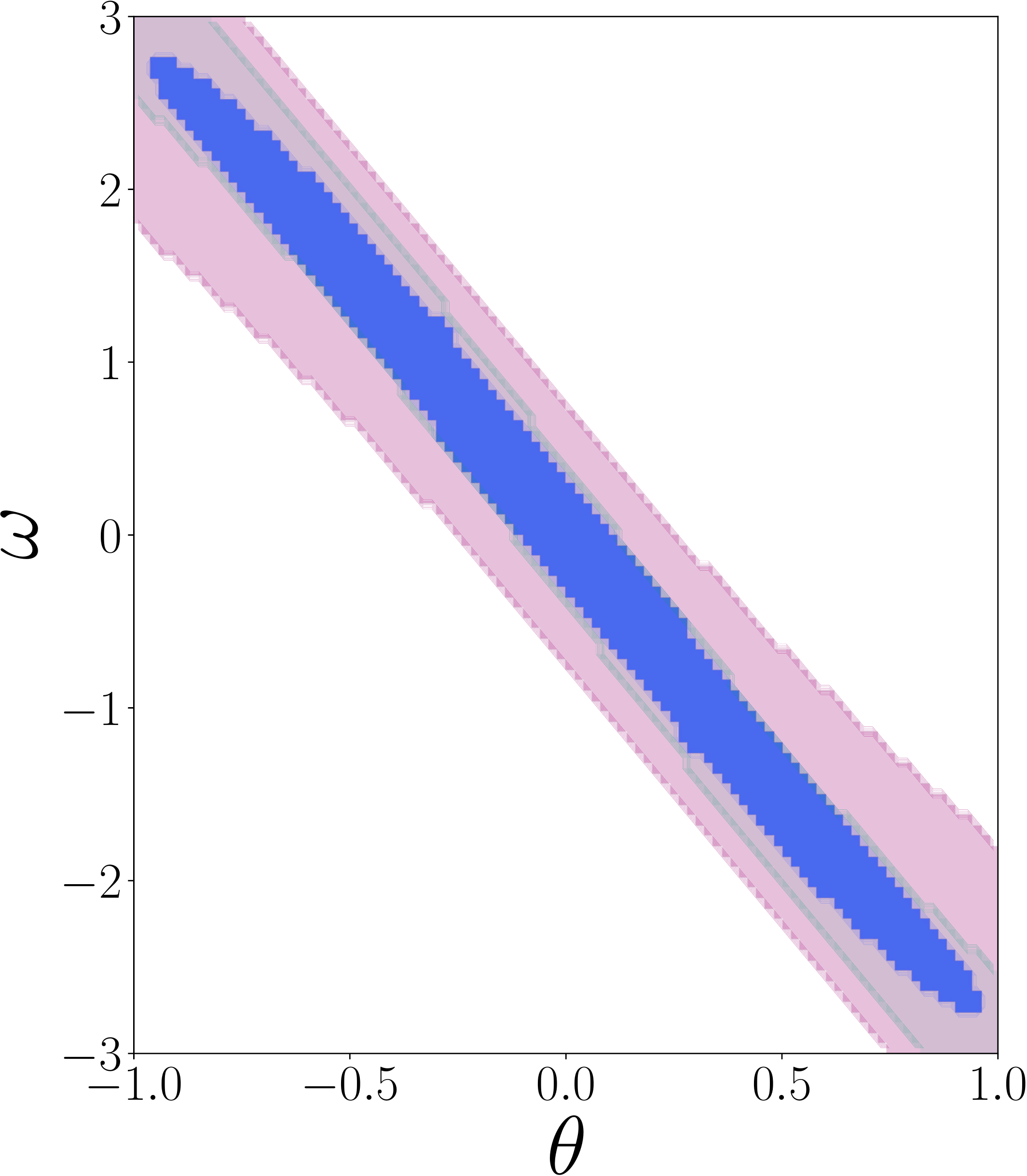}
    \end{subfigure}
    \begin{subfigure}[t]{\sfsize}
        \centering
        \includegraphics[width=1\textwidth]{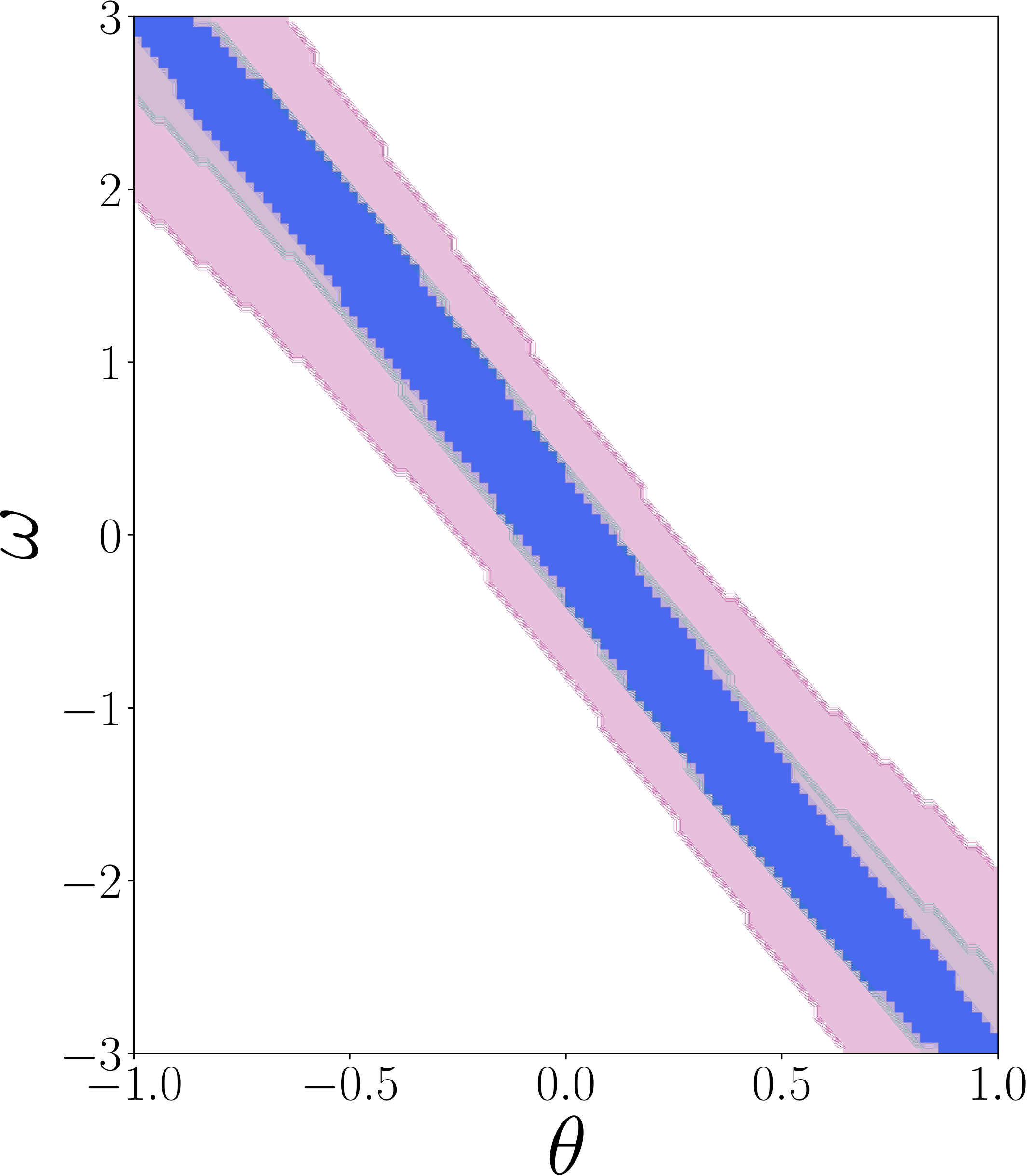}
    \end{subfigure}

%% file: figs/increasing_roas/tex_files/increasing_nested_roa_figure_stage_3.tex
\begin{subfigure}[t]{\sfsize}
    \centering
    \includegraphics[width=1\textwidth]{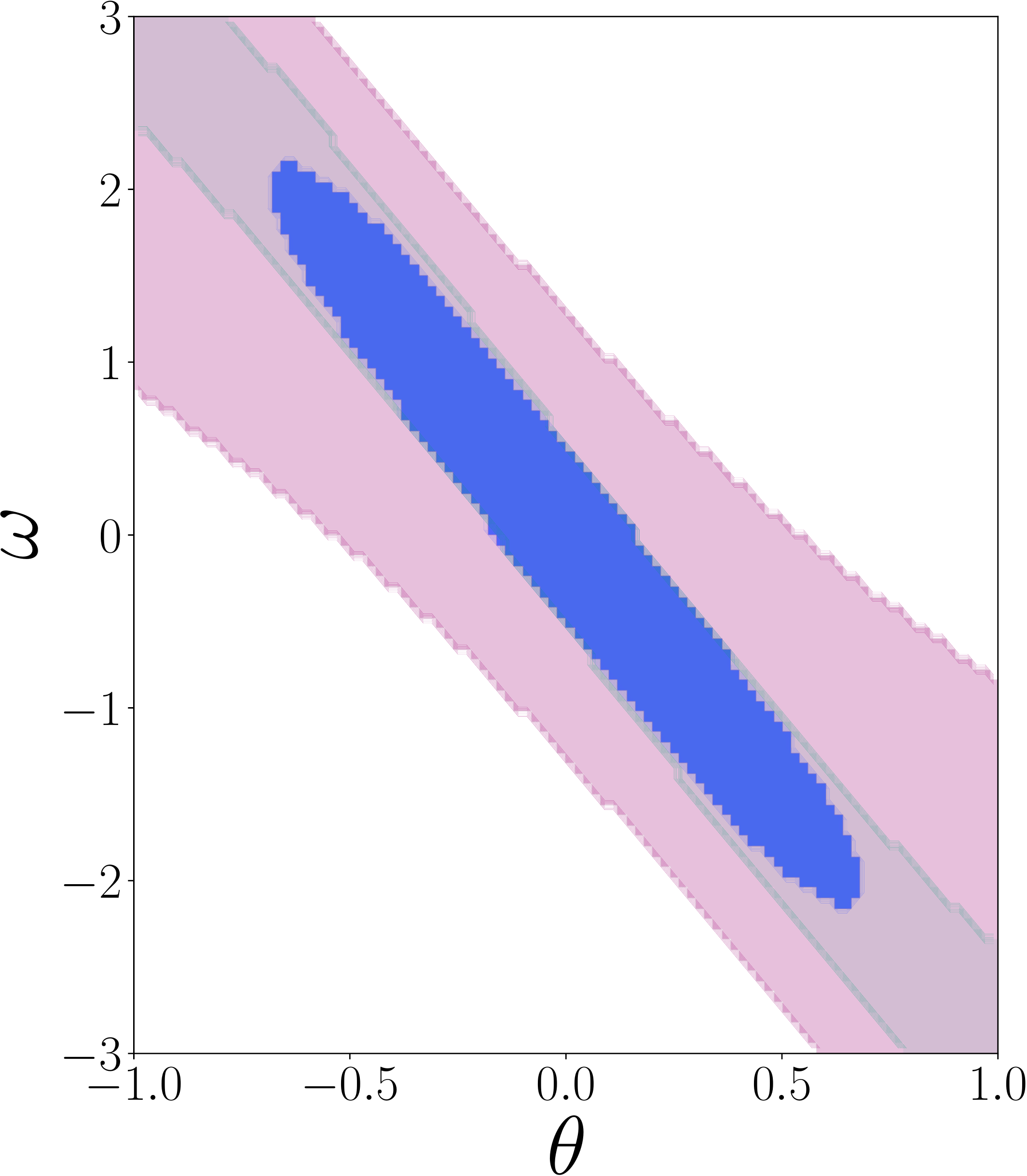}
\end{subfigure}
\begin{subfigure}[t]{\sfsize}
    \centering
    \includegraphics[width=1\textwidth]{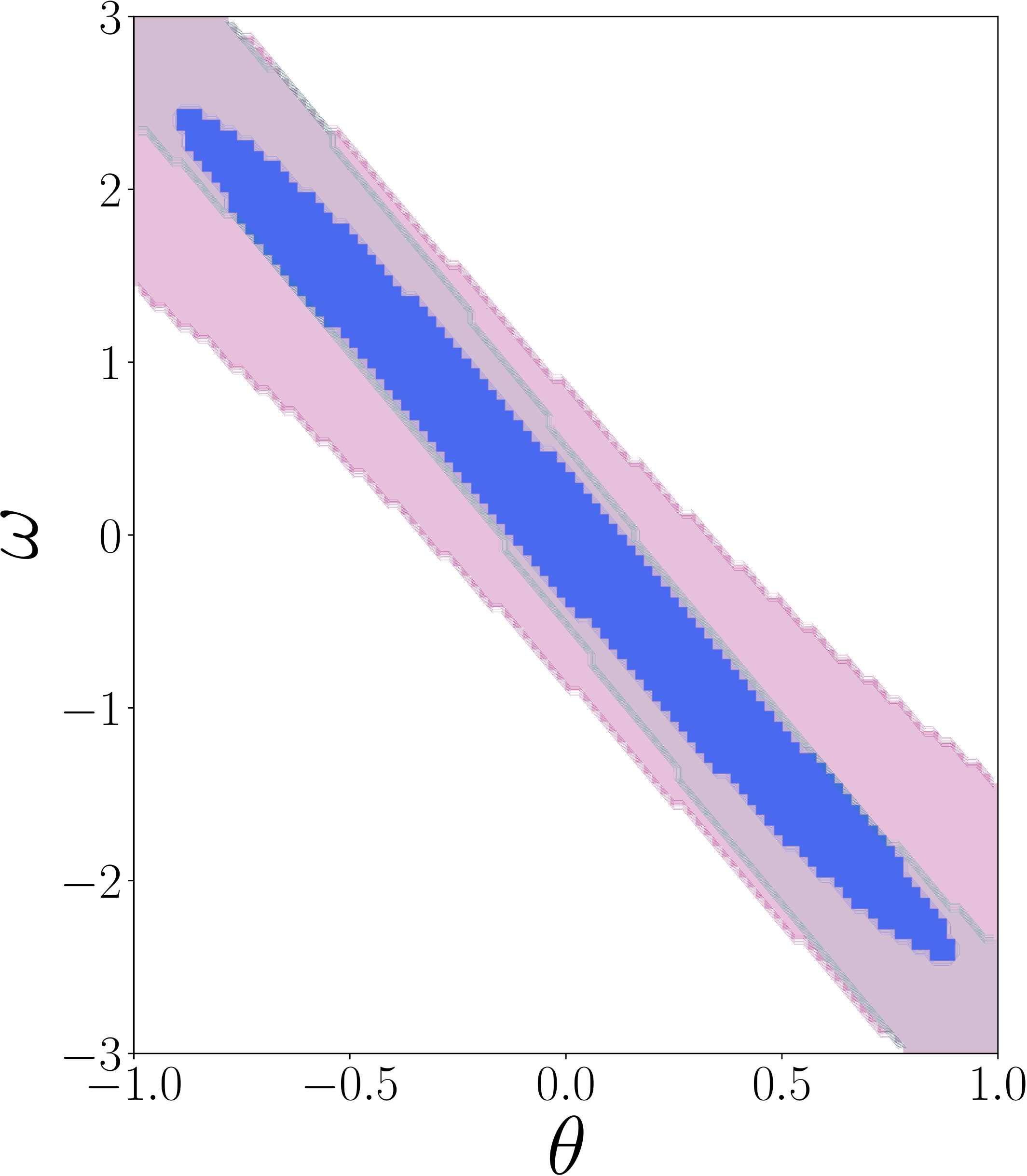}
\end{subfigure}
\begin{subfigure}[t]{\sfsize}
    \centering
    \includegraphics[width=1\textwidth]{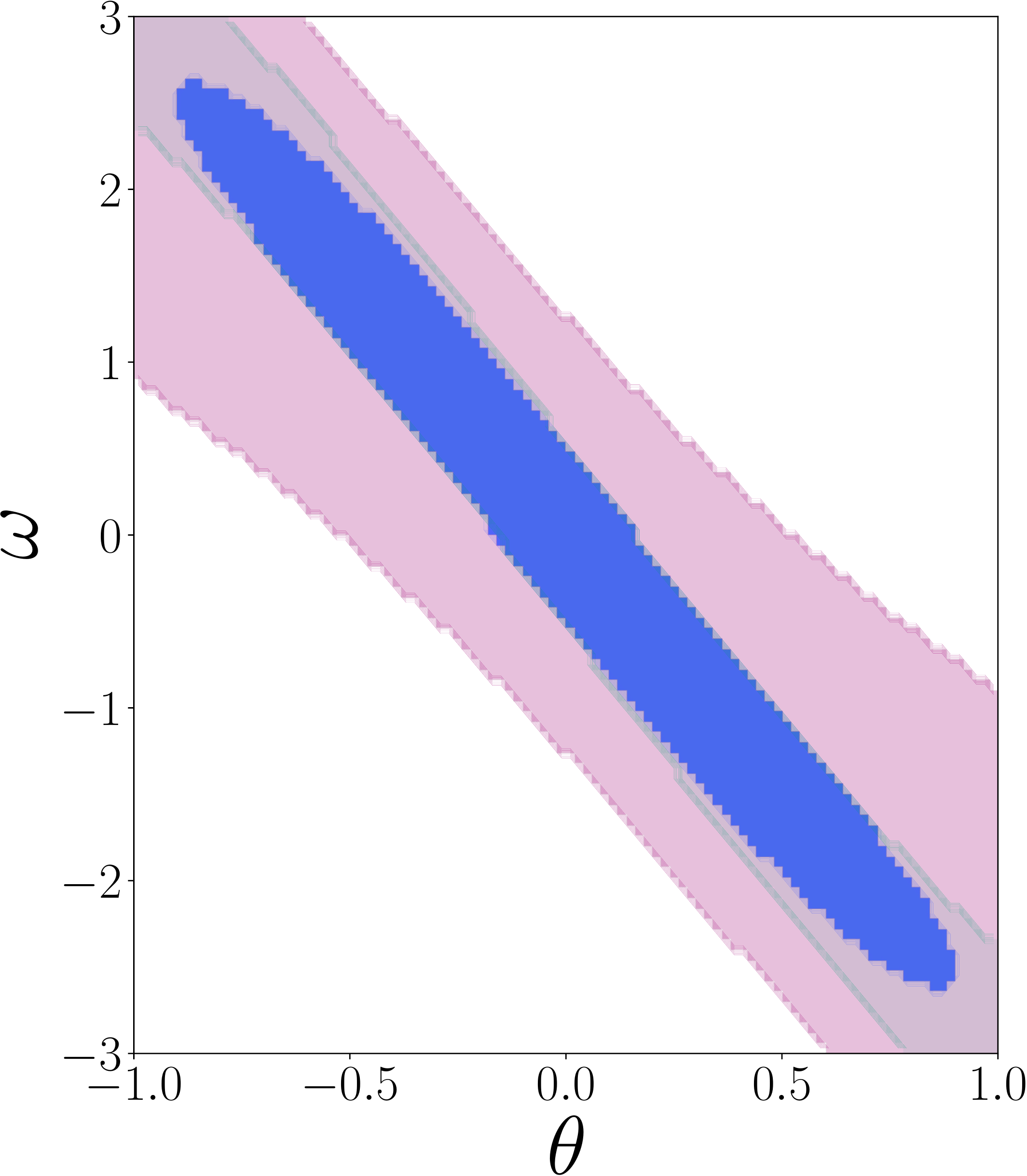}
\end{subfigure}
\begin{subfigure}[t]{\sfsize}
    \centering
    \includegraphics[width=1\textwidth]{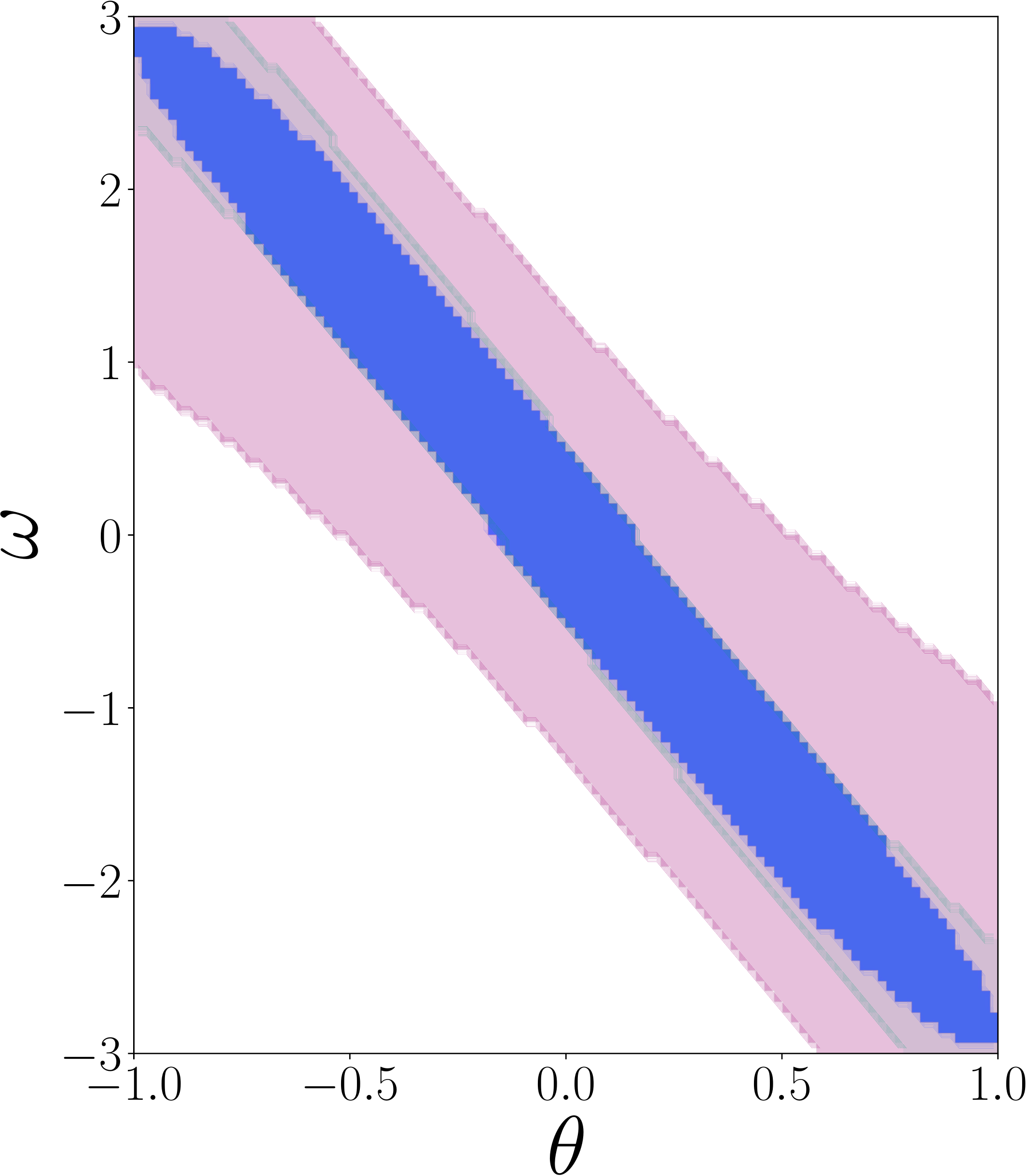}
\end{subfigure}
\begin{subfigure}[t]{\sfsize}
    \centering
    \includegraphics[width=1\textwidth]{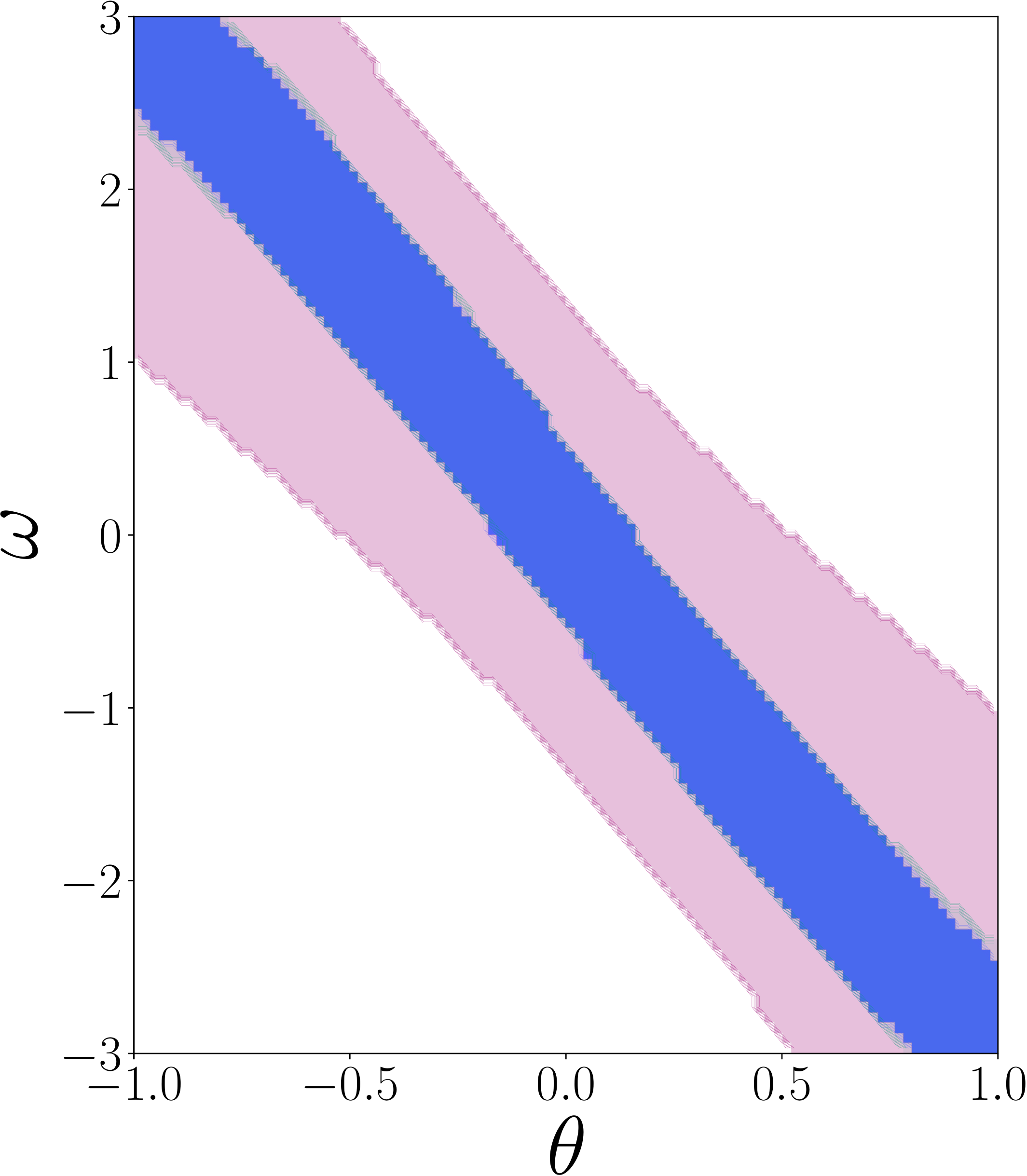}
\end{subfigure}
\begin{subfigure}[t]{\sfsize}
    \centering
    \includegraphics[width=1\textwidth]{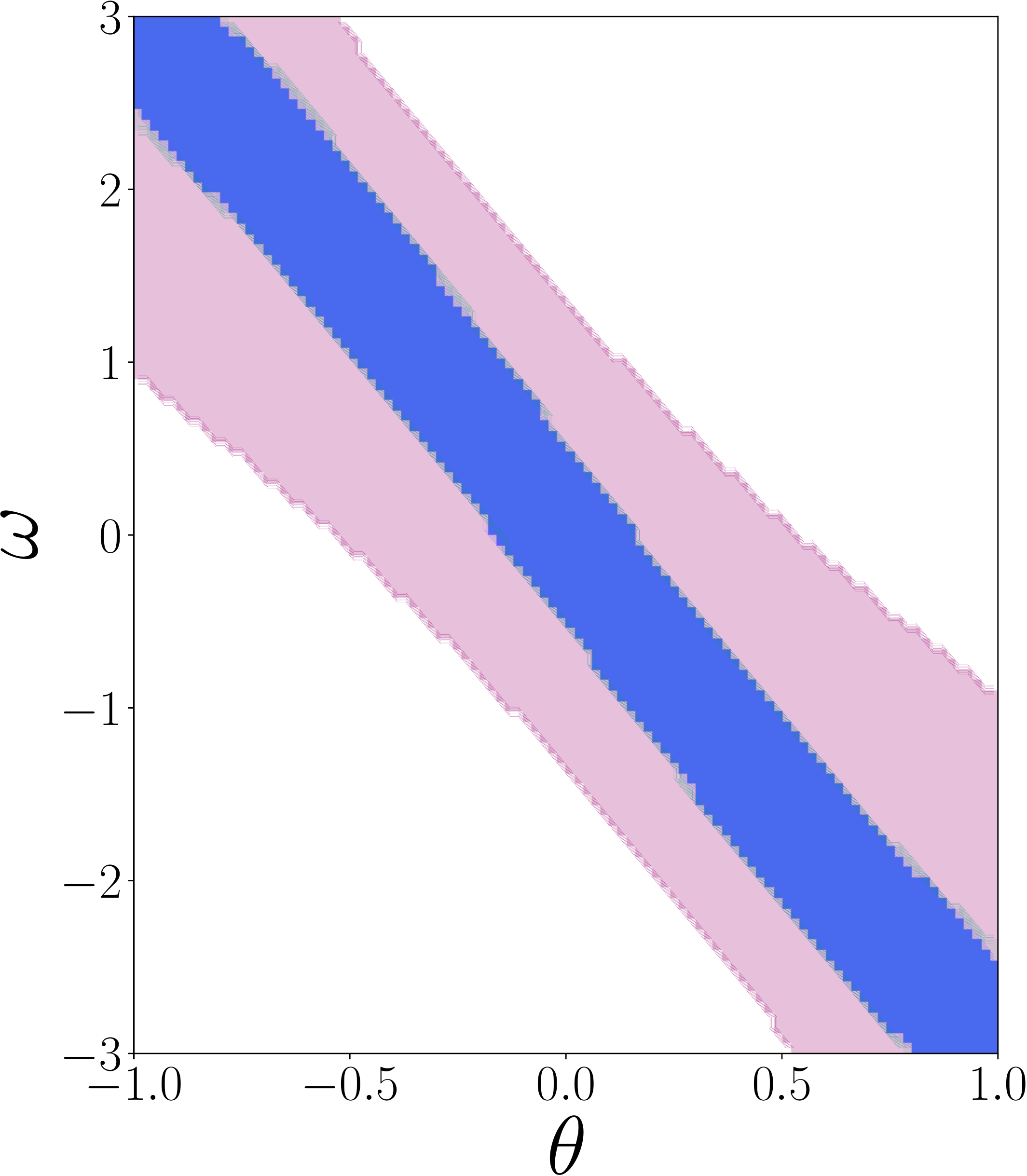}
\end{subfigure}

%% file: figs/increasing_roas/tex_files/increasing_nested_roa_figure_stage_4.tex
\begin{subfigure}[t]{\sfsize}
    \centering
    \includegraphics[width=1\textwidth]{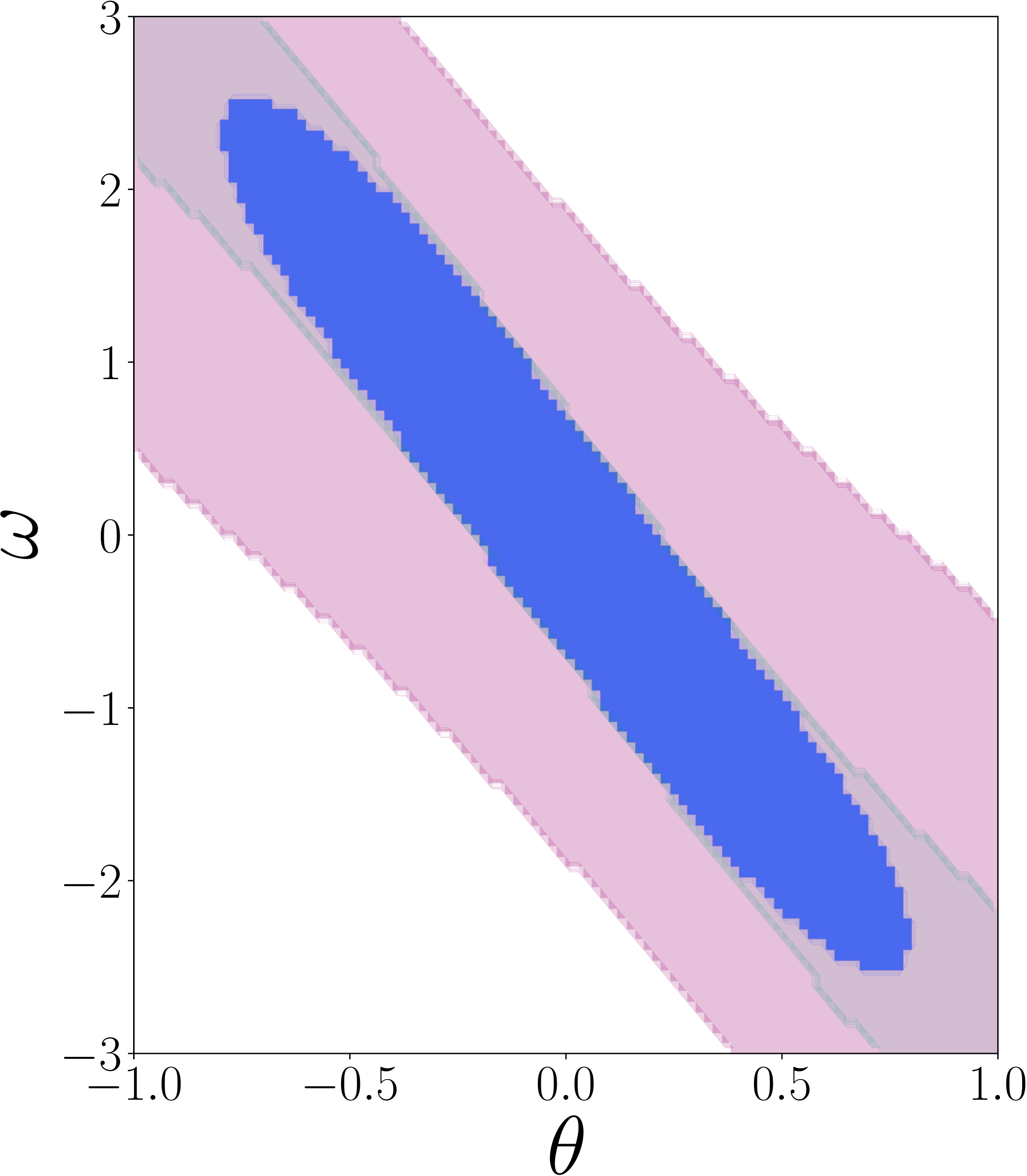}
\end{subfigure}
\begin{subfigure}[t]{\sfsize}
    \centering
    \includegraphics[width=1\textwidth]{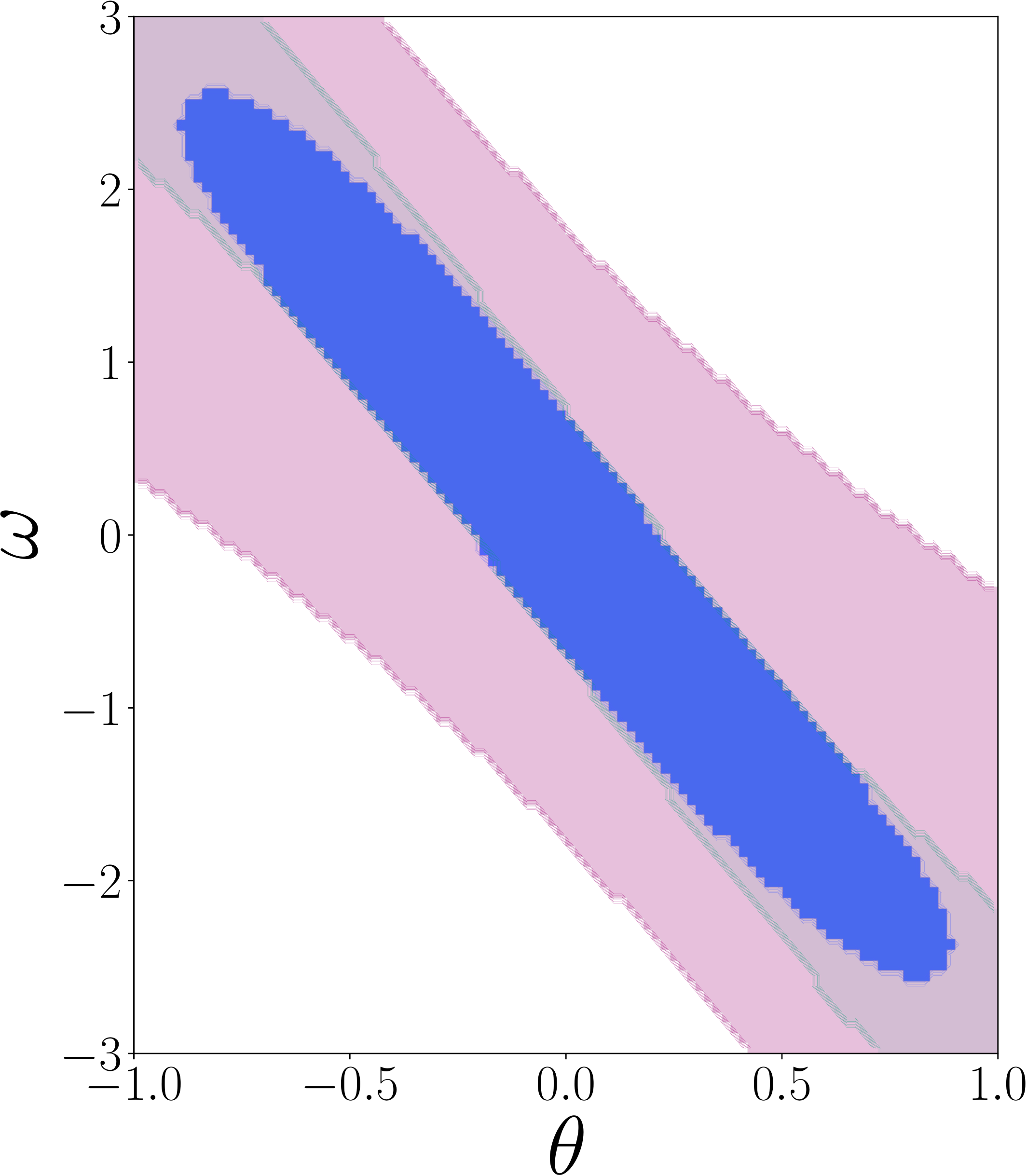}
\end{subfigure}
\begin{subfigure}[t]{\sfsize}
    \centering
    \includegraphics[width=1\textwidth]{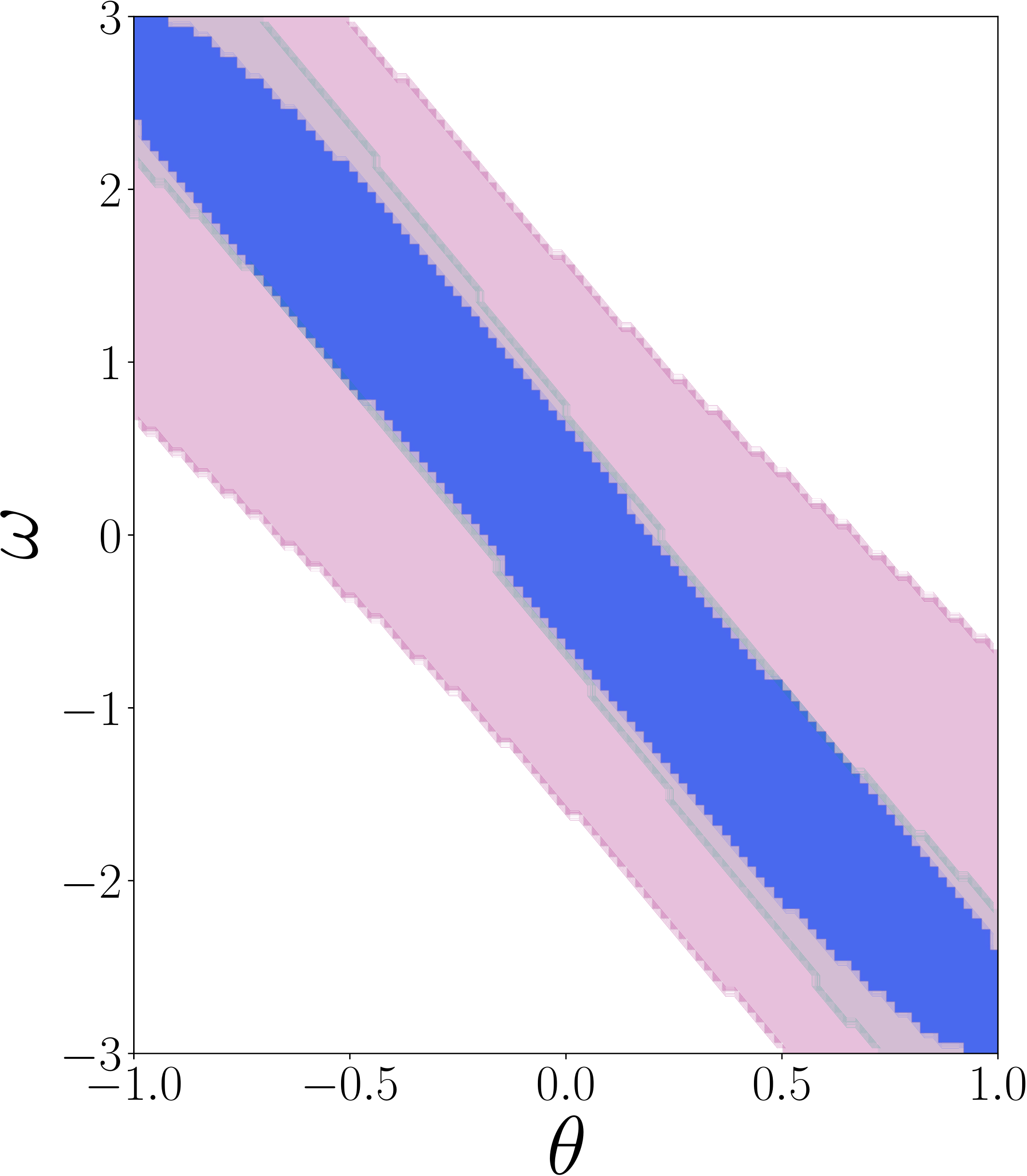}
\end{subfigure}
\begin{subfigure}[t]{\sfsize}
    \centering
    \includegraphics[width=1\textwidth]{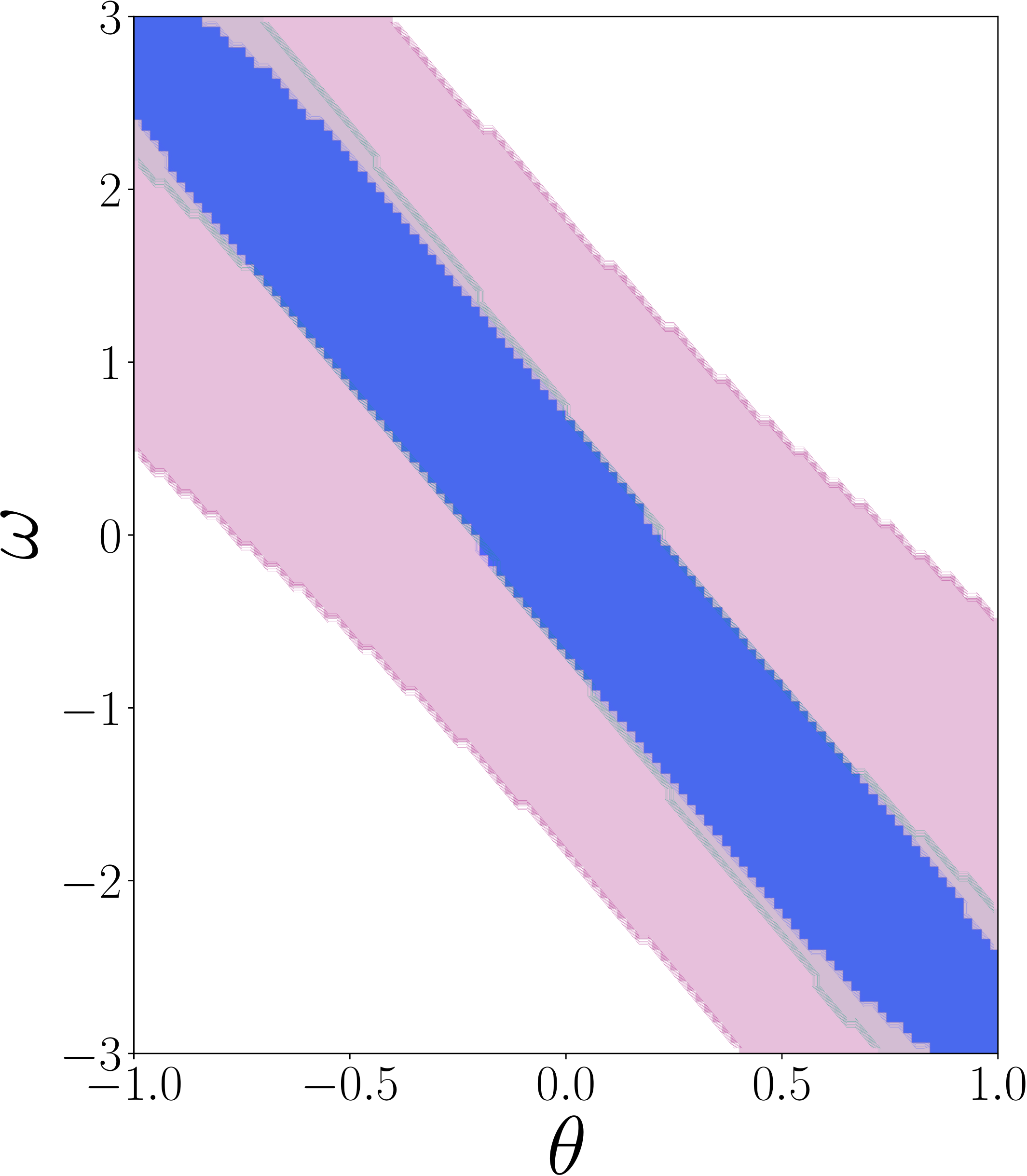}
\end{subfigure}
\begin{subfigure}[t]{\sfsize}
    \centering
    \includegraphics[width=1\textwidth]{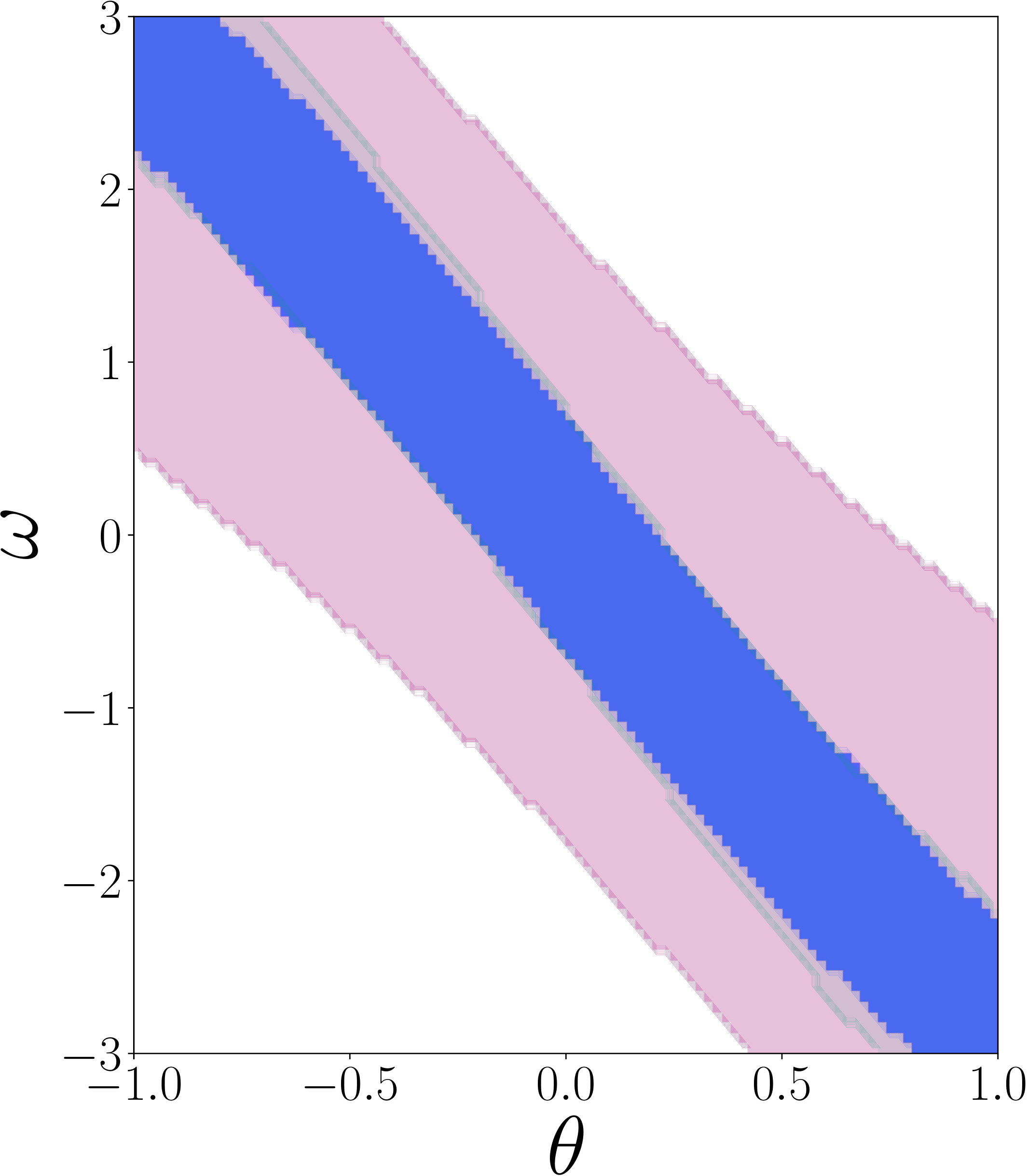}
\end{subfigure}
\begin{subfigure}[t]{\sfsize}
    \centering
    \includegraphics[width=1\textwidth]{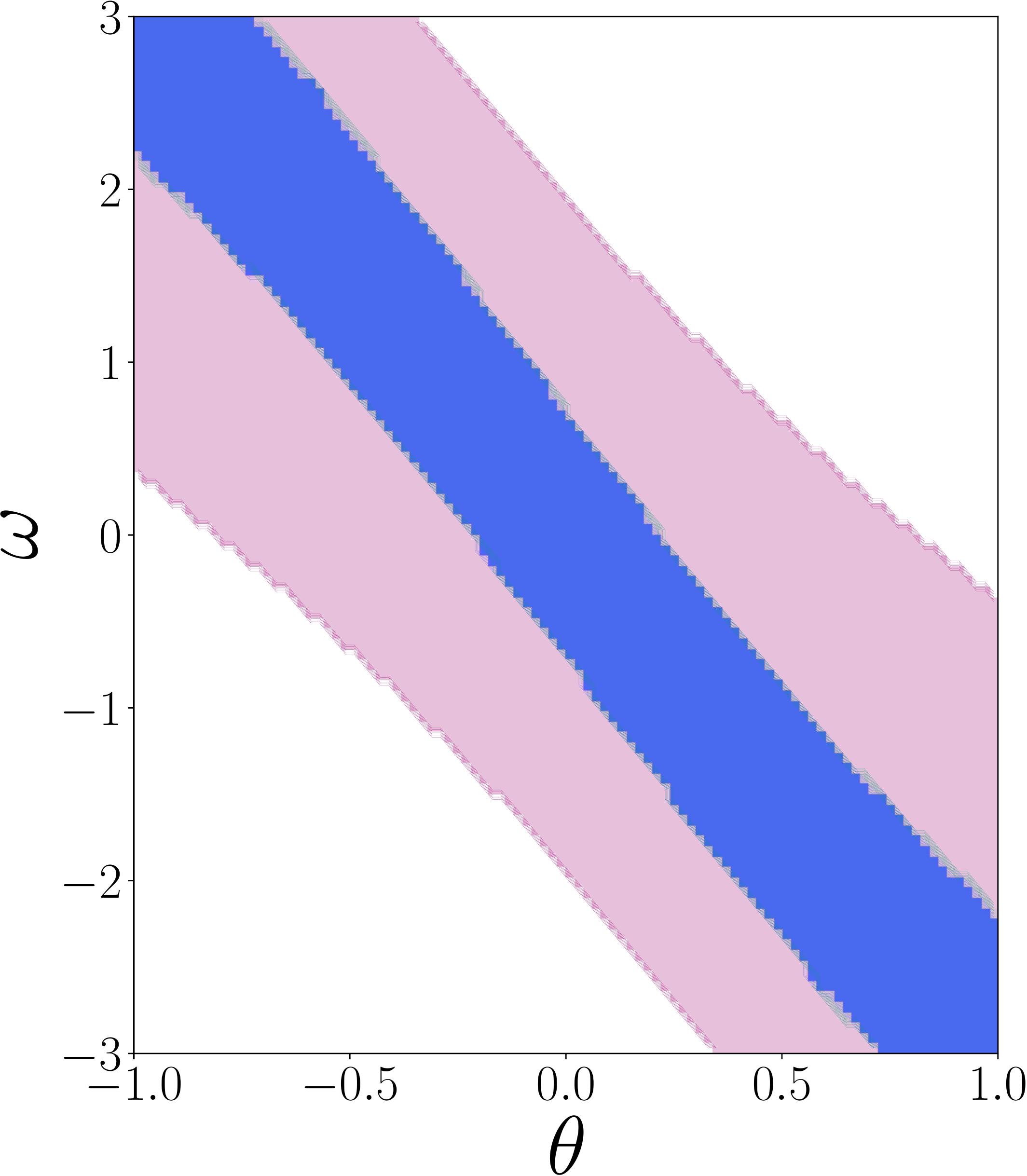}
\end{subfigure}

%% file: experiments_arxiv.tex
\clearpage
\section{Experiments}
\label{sec:experiments}

\newcommand{\hspacesubfigs}{1ex}
\begin{figure}[t!]
    \centering
    \begin{subfigure}[t]{0.23\textwidth}
        \centering
        \includegraphics[width=1\textwidth]{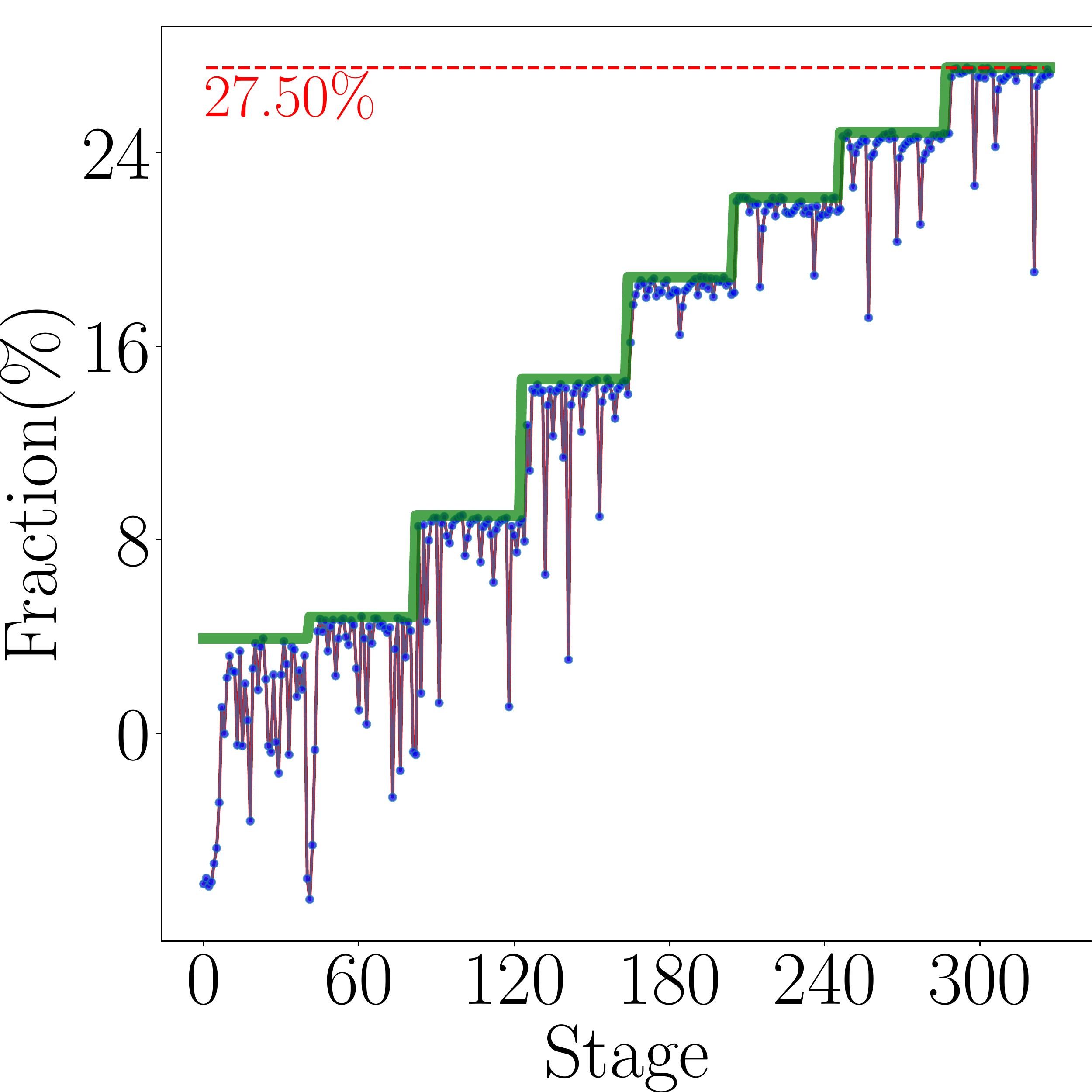}
        \caption{\footnotesize Unimproved process}\label{fig:pendulum_roa_and_policy_fractions_plots_original}        
    \end{subfigure}
    \hspace{\hspacesubfigs}
    \begin{subfigure}[t]{0.23\textwidth}
        \centering
        \includegraphics[width=1\textwidth]{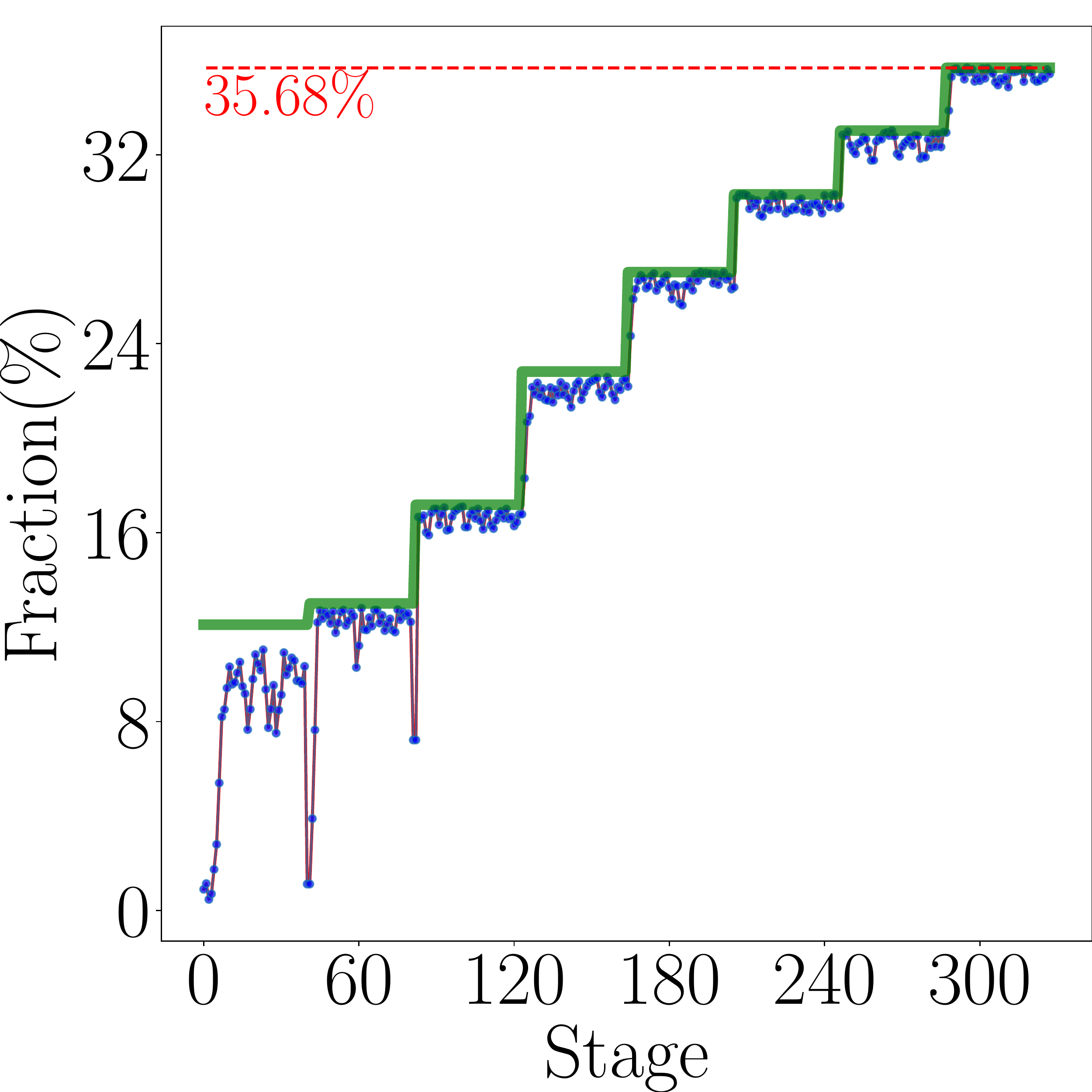}
        \caption{\footnotesize Improved process}\label{fig:pendulum_roa_and_policy_fractions_plots_monotone}        
    \end{subfigure}
    \hspace{\hspacesubfigs}
    \begin{subfigure}[t]{0.23\textwidth}
        \centering
        \includegraphics[width=1\textwidth]{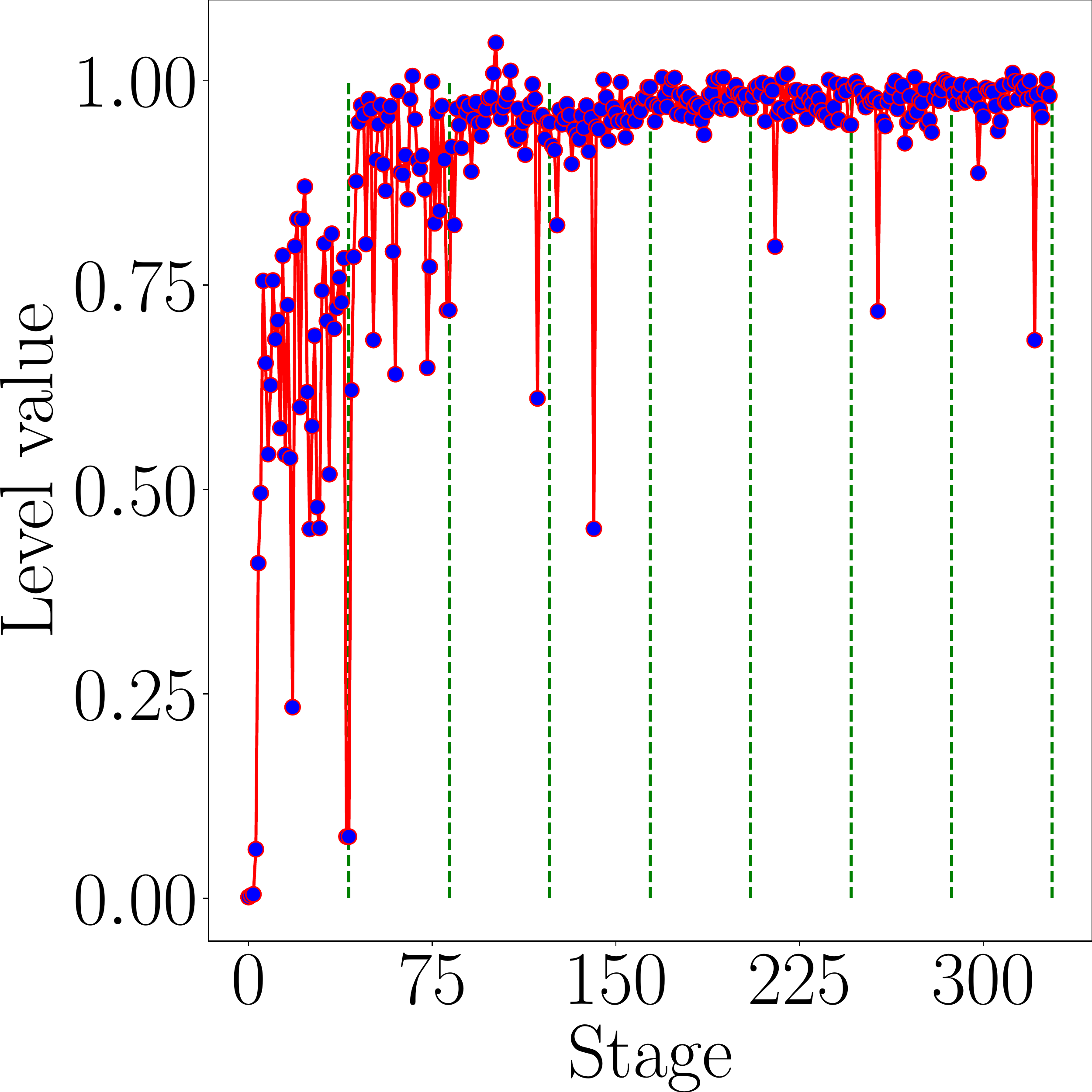}
        \caption{\footnotesize Level values}\label{fig:pendulum_level_values}       
    \end{subfigure}
    \hspace{\hspacesubfigs}
    \begin{subfigure}[t]{0.23\textwidth}
        \centering
        \includegraphics[width=1\textwidth]{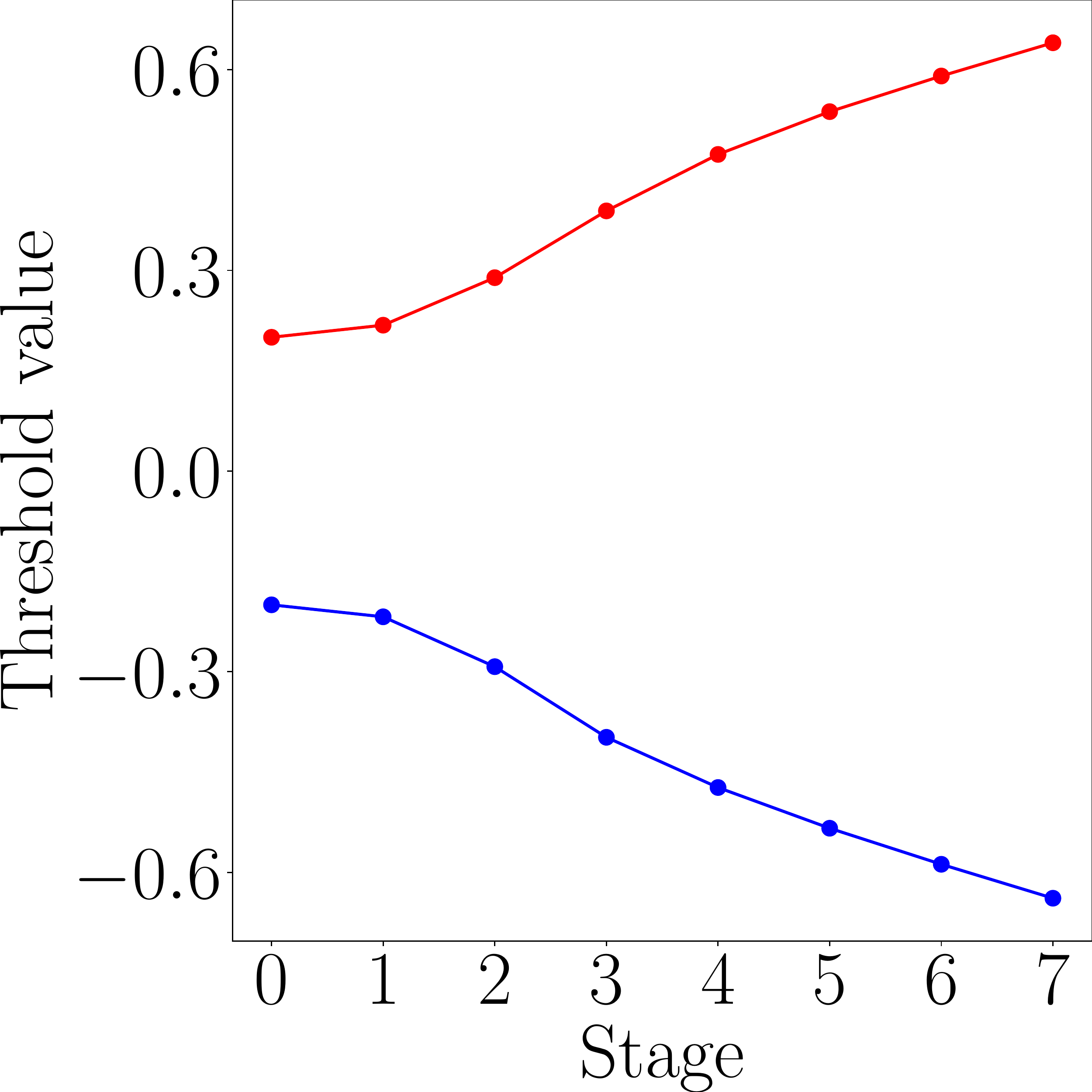}
        \caption{\footnotesize Policy parameters}\label{fig:pendulum_threshold_values}        
    \end{subfigure}
    \caption{\footnotesize (a), (b) The size of the RoA against the iterative stages (phases) of the Algorithms where (a) uses the RoA estimation method of~\cite{richards2018lyapunov} while (b) uses ours. The fraction is computed with respect to a rectangular domain around the equilibrium point that is large enough to enclose the RoA. Green: size of the true RoA. Each jump corresponds to a policy update sub-phase that increases the size of the true RoA. Red: The size of sublevel set that the RoA estimation sub-phase learns to approximate the RoA. After each policy update, the RoA estimation sub-phase takes multiple growth iterations to capture the true RoA as close as possible. (c) The trace of the level values corresponding to every iteration of the RoA estimation sub-phase. 
    (d) Red (Blue): The trace of the value of the upper (lower) threshold parameter of the policy during training. Each point corresponds to a policy update iteration.}\label{fig:pendulum_overall_plots}
\end{figure}

\begin{figure}[t!]
    \centering
    \begin{subfigure}[t]{0.23\textwidth}
        \centering
        \includegraphics[width=1\textwidth]{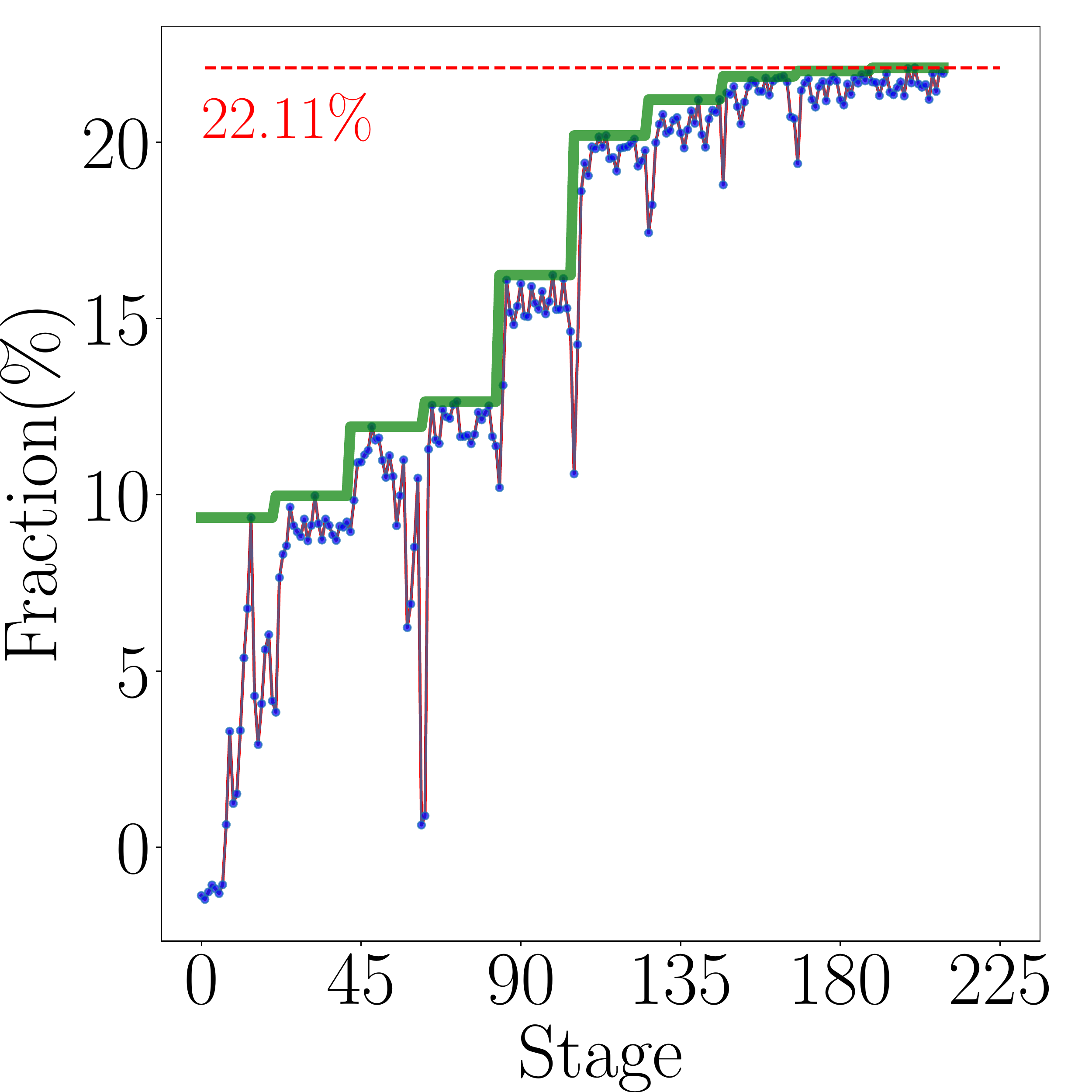}
        \caption{\footnotesize Unimproved process}\label{fig:pendulum_roa_and_policy_fractions_plots_trainable_slopes_original}        
    \end{subfigure}
    \hspace{\hspacesubfigs}
    \begin{subfigure}[t]{0.23\textwidth}
        \centering
        \includegraphics[width=1\textwidth]{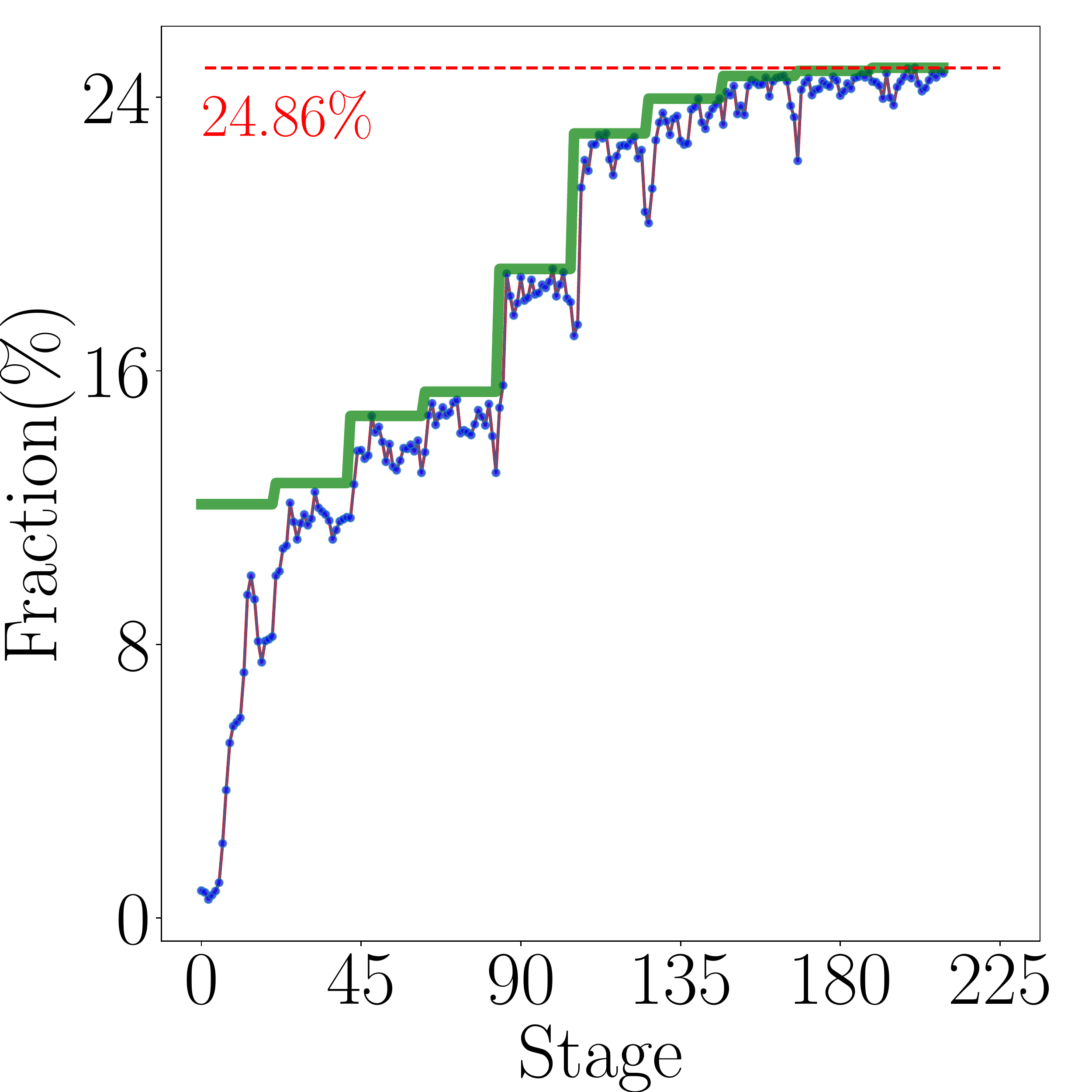}
        \caption{\footnotesize Improved process}\label{fig:pendulum_roa_and_policy_fractions_plots_trainable_slopes_monotone}        
    \end{subfigure}
    \hspace{\hspacesubfigs}
    \begin{subfigure}[t]{0.23\textwidth}
        \centering
        \includegraphics[width=1\textwidth]{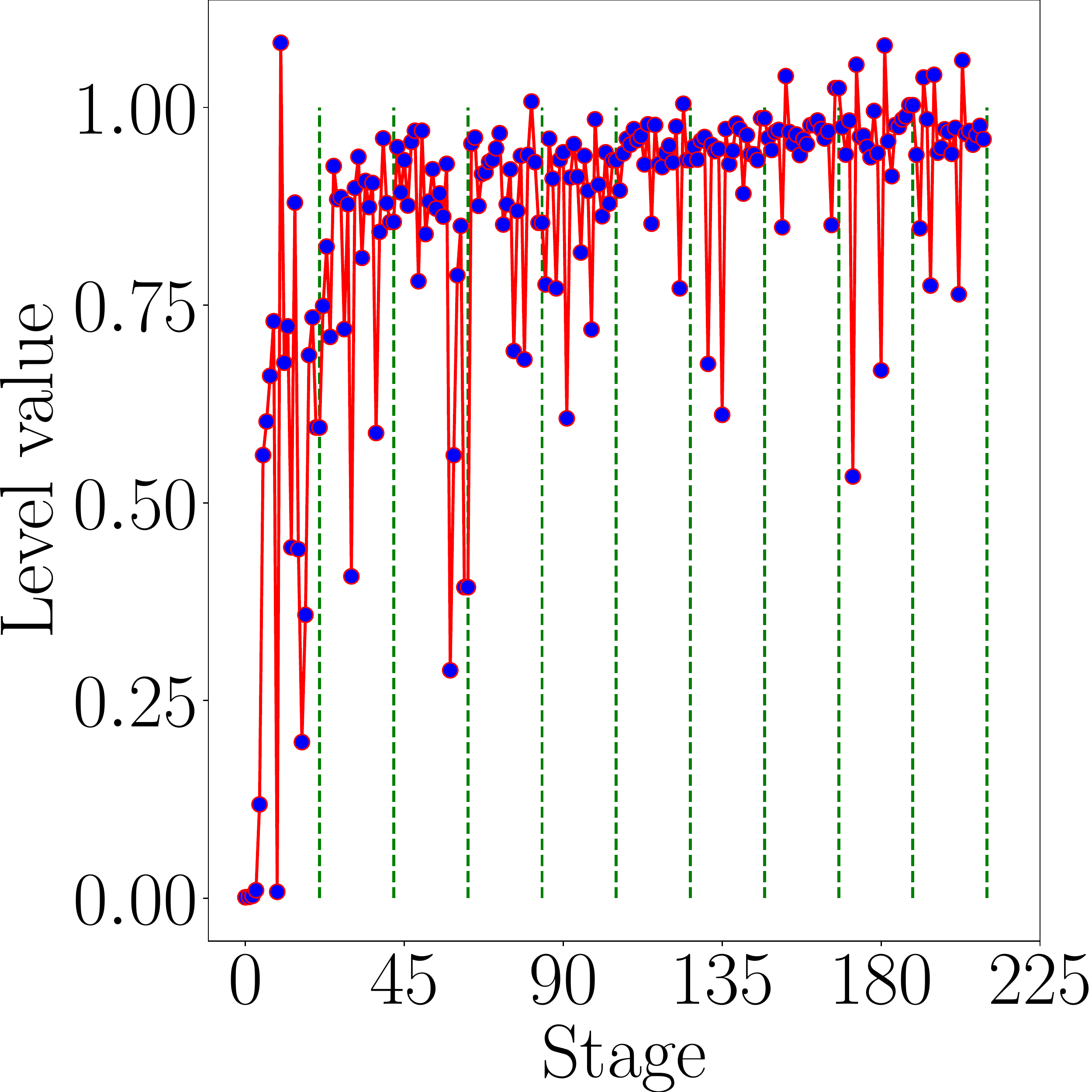}
        \caption{\footnotesize Level values}\label{fig:pendulum_level_values_trainabl_slopes}       
    \end{subfigure}
    \hspace{\hspacesubfigs}
    \begin{subfigure}[t]{0.23\textwidth}
        \centering
        \includegraphics[width=1\textwidth]{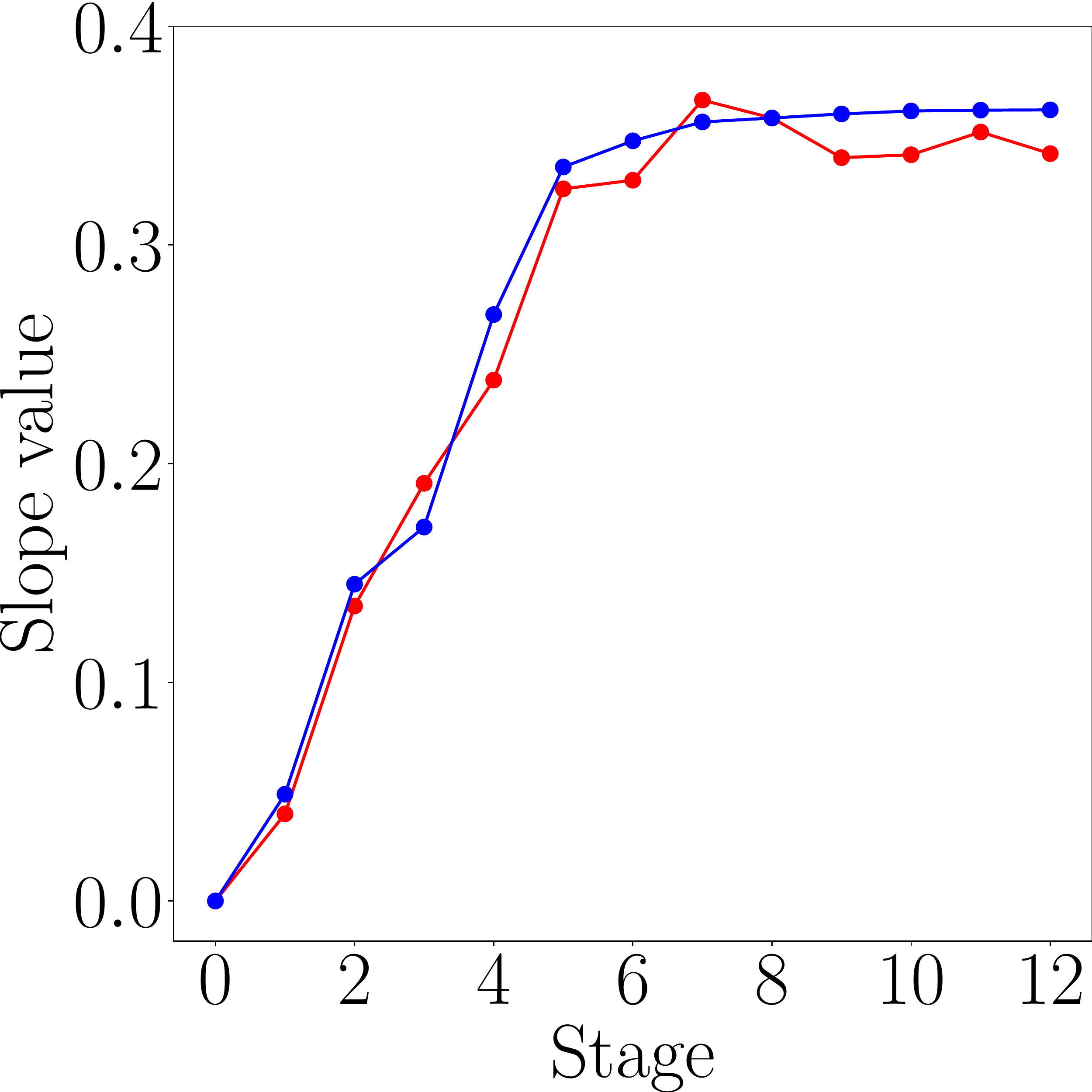}
        \caption{\footnotesize Policy parameters}\label{fig:pendulum_slope_values}        
    \end{subfigure}
    \caption{\footnotesize The same plots as~\Cref{fig:pendulum_overall_plots} with the exception that, here the threshold parameters are kept fixed at $a=0.2$ and $b=-0.2$ while the slopes $m_a$ and $m_b$ are trainable parameters.}.\label{fig:pendulum_overall_plots_trainable_slopes}\vspace{-2em}
\end{figure}

\begin{wrapfigure}{r}{0.3\textwidth}
    \centering
    \vspace{-2ex}
    \includegraphics[width=0.3\textwidth]{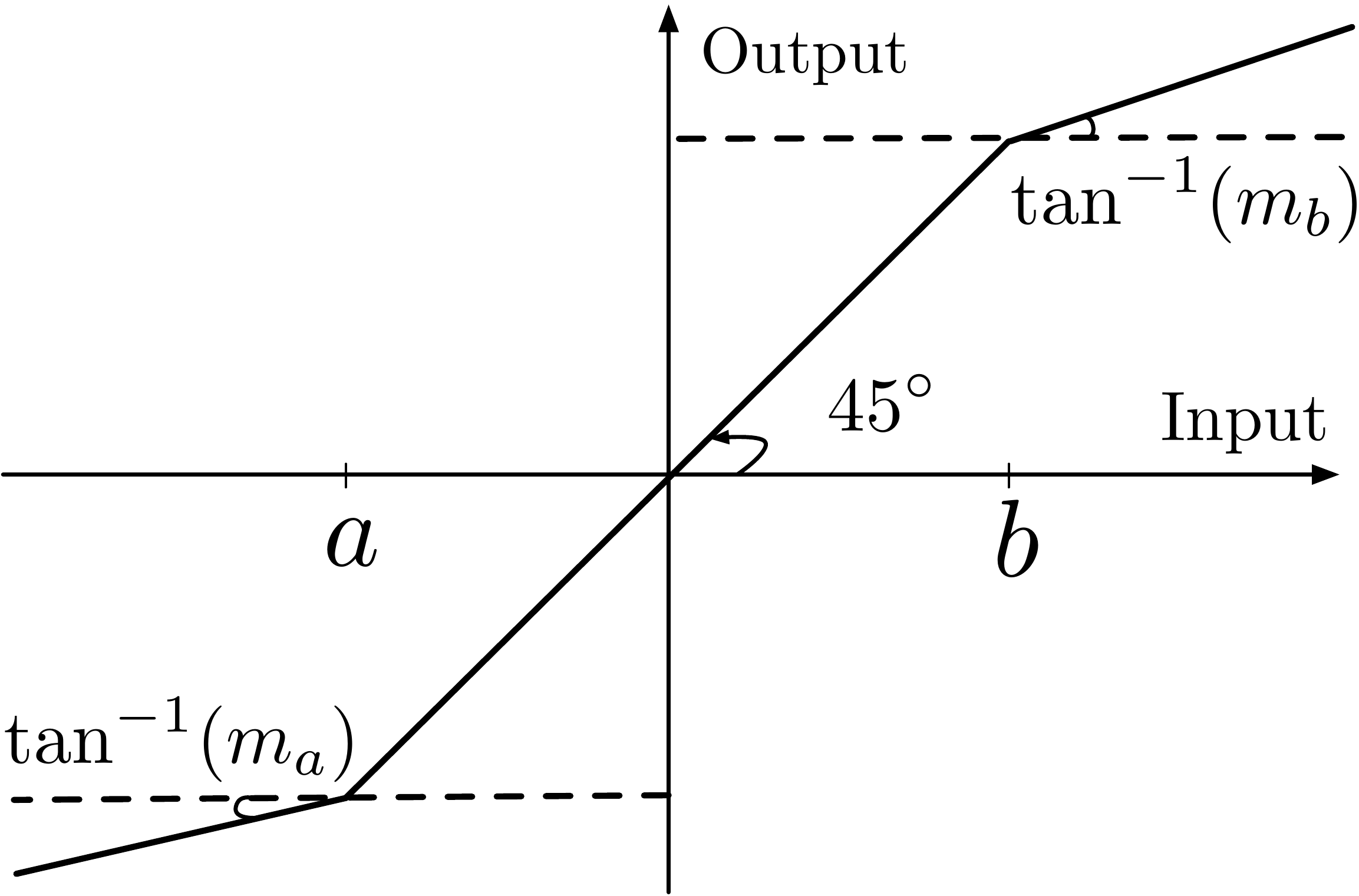}
    \caption{\footnotesize Loose saturation\vspace{-1ex}}\label{fig:sat_function}
\end{wrapfigure}
We consider an inverted pendulum system defined as $\dot{\theta} = \omega$ and $\dot{\omega} = \frac{g}{l}\sin(\theta) + \frac{u}{I} - \mu_f\frac{\omega}{I}$ where the state vector $\xb=(\theta, \omega)$ consists of the angle and angular velocity. Moreover $(g=0.81, l=0.5, I={\rm mass}\times l^2=0.25, \mu_f=0)$ are the acceleration of gravity, length, inertia, and friction coefficient. The scalar $u$ is the input force. The open-loop system (with $u=0$) has equilibrium points at $(\theta,\omega)=(k \pi, 0)$ with $k\in\ZZ$. We focus on the equilibrium point $(0, 0)$ in the frictionless setting where the system shows oscillatory behavior and consequently is not asymptotically stable (see~\Cref{sec:discretized_statespace} for system and modeling details). First, an LQR controller $K$ is designed for the linearized system around the origin (see the vector field and the initial RoA in~\Cref{fig:experiment_system_overall} in the Appendix). The control signal passes through a loose saturation function as $u=\pi_0(\xb;\psi) = \sat_\psi([\theta, \omega]\tran K)$. The function $\sat$ is parameterized by $\psi=(a, b, m_a, m_b)$ as illustrated in~\Cref{fig:sat_function}. In the first experiment, the slopes $m_a=0$ and $m_b=0$ are kept fixed where $a$ and $b$ are trainable parameters of the policy, i.e., $\psi=(a, b)$. The Lyapunov function $V(\cdot;\theta)$ is realized by a $3-$layer neural network parameterized by $\theta$. Each layer has $64$ neurons with a special architecture (see~\Cref{sec:lyapunov_architecture}) inspired by~\citep{richards2018lyapunov} followed by tanh activation function that imposes the positive definiteness of the entire network as is required by the Lyapunov conditions in~\labelcref{eq:lyapunov_conditions}. The chosen hyper-parameters for the SGD training can be found in~\Cref{sec:training_hyperparameters}.

The initial controller gives a small RoA since it is designed for the local linear approximation of the system. As the initial policy is LQR designed for the locally linearized model, $V(\cdot;\theta)$ is pre-trained by the quadratic function $0.1\theta^2 + 0.1\omega^2$.

After pre-training, sub-phases of the algorithm of~\Cref{sec:algorithm} are run alternately to capture the RoA and improve the policy. The green step-like plot in~\Cref{fig:pendulum_roa_and_policy_fractions_plots_original} and~\Cref{fig:pendulum_roa_and_policy_fractions_plots_monotone} shows the true size of the RoA. Each jump in the green plots shows one iteration of the policy update algorithm resulting in an increased RoA. The fluctuating red plot in~\Cref{fig:pendulum_roa_and_policy_fractions_plots_original} shows the size of the estimated RoA without our improvement over the RoA estimation algorithm of~\cite{richards2018lyapunov} while~\Cref{fig:pendulum_roa_and_policy_fractions_plots_monotone} shows the outcome of the presence of our proposed additional term in the loss function~\labelcref{eq:roa_update_objective}. It shows that the added term results in a less fluctuating estimate of the RoA, and when combined with the policy update sub-phase, gives a faster convergence to a larger RoA ($35.68\% \text{ vs } 27.50\%$ fraction of the domain volume after $7$ policy updates).

As stated in~\Cref{sec:algorithm}, $\bar{c}=1$ is not necessarily equal to $c_n$. After learning $V_n$ with $\bar{c}$ in~\labelcref{eq:roa_update_objective}, the algorithm searches for a value of $c_n$ such that the Lyapunov decrease condition is met for all states within the sublevel set $\Scal_{c_n}(V_n)$. It can be seen in~\Cref{fig:pendulum_level_values} that these values converge to $\bar{c}=1$ that can be perceived as a sign of the stable training of the algorithm.

The trace of the parameters of the policy is shown in~\Cref{fig:pendulum_threshold_values}. As these policy's trainable parameters represent the upper and lower limits of the loose threshold function, the policy learning algorithm updates them in the directions that decrease their suppressing effect. This is what we also expect from the physics of the system. Graphical visualization of the policy update and RoA estimation sub-phase is shown in~\Cref{fig:cartoon_and_graphical_RoA}(Right). Each row shows one phase of the algorithm. The policies are updated along the rows from top to bottom. Within one row, the policy is fixed and RoA is estimated from left to right (See~\Cref{fig:pendulum_nested_roas_full} in the appendix for a larger visualization).

In the second experiment, the threshold limits $a=-0.2$ and $b=0.2$ are kept fixed while the slopes $\psi=(m_a, m_b)$ are trainable parameters. The rest of the training setting remains the same as the previous experiment. Figures~\labelcref{fig:pendulum_roa_and_policy_fractions_plots_trainable_slopes_original} and~\labelcref{fig:pendulum_roa_and_policy_fractions_plots_trainable_slopes_monotone} shows that the policy update sub-phase enlarges the RoA (green plot) of the system while the RoA estimation sub-phase (red plot) manages to follow the new RoA after each policy update. Our improved RoA estimation algorithm results in a more monotonic convergence of the estimated RoA that ultimately learns a controller that induces a larger RoA ($24.86\% \text{ vs } 22.11\%$). Similar to~\Cref{fig:pendulum_level_values}, convergence of $c_n$ values to $\bar{c}=1$ can be seen in~\Cref{fig:pendulum_level_values_trainabl_slopes}. The trace of the upper and lower slopes are shown in~\Cref{fig:pendulum_slope_values}. Almost equal learned values for upper and lower slopes are expected due to the structural symmetry of the saturation function (\Cref{fig:sat_function}) that appears in the closed-loop system.

%% file: related_works_arxiv.tex
\vspace{-1ex}
\section{Related Work and Conclusions}

\label{sec:related_works}
We have proposed a two-player collaborative and iterative algorithm that iterates over two sub-phases that learn the Lyapunov function and use it to learn a controller to enlarge the RoA of the system. 

The existing approaches that are close to the purpose of this paper are those that simultaneously synthesize a controller and maximize the stability region. Our work does not put any limit on the class of considered systems except the generic condition of hyperbolicity. For restricted classes such as polynomial systems, Sums-of-Squares (SOS) method leads to a bilinear optimization that is solved by some form of alternation~\citep{jarvis2003some,majumdar2013control}. As a dual to Lyapunov-based methods,~\cite{majumdar2014convex} uses the notion of occupation measure to optimize a feedback controller for a polynomial system, but it has scalability issues due to its reliance on SDP optimization toolbox. Our work is different from this class of methods as our method is not limited to polynomial systems. Moreover, our method uses automatic differentiation that is naturally combined with neural networks to enjoy their superior scalability combined with SGD optimization. 

Among data-driven approaches, ~\cite{berkenkamp2017safe} use statistical models of the system to learn a controller with the assumption that the Lyapunov function is given. Our method is different as we learn the Lyapunov function and the controller together in an alternating fashion. Neural Lyapunov Control by~\cite{chang2019neural} is closer to our work while their approach is different from our method in multiple ways: They need to solve a costly global optimization problem for a component called~\emph{falsifier} to find the states on which the Lyapunov conditions are violated. We rather use the geometric properties of the level sets of the Lyapunov function to find the potentially unstable states that are chosen from a ring around the current stable sublevel set. In addition, Unlike~\cite{chang2019neural}'s method that needs to be done until the end to be usable on the system, the growing nature of our work allows the system be in action while the controller keeps improving.

In the context of reinforcement learning, our method can be seen as an actor-critic approach~\citep{grondman2012survey, bhasin2013novel,lillicrap2015continuous} where the actor tries to stabilize the system while the critic estimates the size of the RoA induced by the controller. The actor is the policy network and the critic is the network that implements the Lyapunov function.

In control theory, as the title of our work suggests, the proposed algorithm can be seen as an automatic version of the celebrated Lyapunov redesign method~\citep{khalil2002nonlinear, hwang2013robust} where the Lyapunov function and the closed-loop system are re-designed together iteratively with the purpose of enlarging the stability region.

In this work, we assume the model of the system is given. This condition can be relaxed as suggested by~\Cref{rem:model_based_assumption}. Investigating this relaxation can be considered for the future. For example, learning a local model that adapts at each iteration could be one solution that adds another component to the algorithm and turns it into a three-player collaborative game.

\clearpage

%% file: supplement_arxiv.tex
\newpage
\appendix

\section{Control Lyapunov Function and Maximal Stabilizable Set}
\label{sec:control_lyapunov_function}

The use of the Lyapunov theory to guide designing the input of a system has been made precise with the introduction of~\emph{control Lyapunov function (clf)}. A clf for a system of the form $\xb_{k+1}=f(\xb_k, \ub_k)$ is a $C^1$, radially unbounded function $V:\Xcal\to\RR_+$ if
\begin{align}
    &V(\mathbf{0})=0\quad\text{and}\quad V(\xb) > 0\quad&\forall \xb\in\Dcal\backslash \{\mathbf{0}\}\\
    &\inf_{\ub\in U} [V(\xb) - V(f(\xb, \ub))]\leq 0 \quad&\forall \xb\in\Dcal\backslash \{\mathbf{0}\}
\end{align}
Just as the existence of a Lyapunov function is necessary and sufficient conditions for the stability of an autonomous system, the existence of a clf is a necessary and sufficient condition for~\emph{stabilizability} of a system with control input. In other words, the existence of clf guarantees the existence of a controller that stabilizes the system for all initial states within a neighborhood around the equilibrium point. 

According to the definition of Lyapunov function in~\Cref{sec:preliminaries} and clf above, Lyapunov function assess the stability of a closed-loop system for a fixed controller while clf investigates the~\emph{existence} of a control signal that stabilizes the system in a domain $\Dcal$. In the definition of clf, no functional limitation is assumed for the control signal $\ub_k$. In practice, $\ub_k$ is produced by a state-feedback controller via the policy function $\pi\in\Pi$. Moreover, due to the implementation constraints, only a subset $\tilde{\Pi}\subseteq\Pi$ of these functions can be realized. Therefore, it is quite likely that $\mu(\Rcal_{\pi^*})<\mu(\Dcal)$ with $\mu$ be the Lebesgue measure, i.e., the best feasible controller cannot expand the RoA of the system to the entire $\Dcal$. Let $U(\xb)$ be the values that the control signal can take at state $\xb$. Then, we define
\begin{equation}
  \label{eq:maximal_stabilizable_set}
  \bar{\Rcal} = \sup_{\Bcal\subseteq \Dcal} \Bcal\quad\text{such that}\quad \exists~\text{clf $V$ on $\Bcal$ and a $\pi\in\tilde{\Pi}$ \text{materializes the RoA} $\bar{\Rcal}$}
\end{equation}
where $\bar{\Rcal}$ is the maximal stabilizable set as was used in~\Cref{sec:problem_statement}.

\section{Proofs}
In this section, a more detailed theoretical exposition of some of the material that is dropped from the main text due to space limitation is presented.

The following lemma shows that the derivative of a Lyapunov function vanishes at the equilibrium point. One can see the Lyapunov redesign method as an actor-critic algorithm where the Lyapunov function plays the role of the critic. As the learning signal for updating the actor (policy) passes through the derivative of the critic (Lyapunov function) due to the chain rule, the following lemma implies that getting closer to the equilibrium will weaken the information content of the signal for learning the policy.
\begin{lemma}\label[lemma]{lem:dv_dx_at_equilibrium}Derivative of a Lyapunov function at the origin:
  Let $\Xcal\subseteq \RR^d$ be a $d-$dimensional vector space and $V:\Xcal\to \RR$ be a continuous positive definite function, i.e., $V(\xb)>0$ for $\xb\neq \mathbf{0}$ and $V(\mathbf{0})=0$. Then, $\nabla_\xb V(\xb)|_{\xb=\mathbf{0}}=0$.
\end{lemma}
\begin{proof}
  We use the technique of proof by contradiction. Let $\gb = [g_1, g_2,\ldots, g_d]\tran= \nabla_\xb V(\xb)|_{x=\mathbf{0}}\neq 0$. Suppose there exists an index $i\in\{1,2,\ldots, d\}$ such that $g_i\neq 0$. Due to the continuity of $V$, we expand $V(\xb)$ at $\xb=\mathbf{0}$ in the direction of $g_i$ as
  \begin{equation*}
    V(0, \ldots, x_i=c, \ldots, 0) = V(\mathbf{0}) + c \frac{\partial V(\xb)}{\partial x_i}|_{\xb=\mathbf{0}}+ o(x_i)
  \end{equation*}
for an arbitrary value of $c$ close to $0$. Since $c$ is arbitrary, we choose $c=-\epsilon \frac{\partial V(\xb)}{\partial x_i}|_{\xb=\mathbf{0}} = -\epsilon g_i$. As we assumed $V(\mathbf{0})=0$, for $\epsilon$ sufficiently close to $0$, we can write
\begin{equation*}
  V(0, \ldots, x_i=c, \ldots, 0) = -\epsilon g_i^2<0\;\; \text{for}\;\; g_i\neq 0\
\end{equation*}
which is in contrast with the positive definiteness of $V$. Therefore, $g_i$ cannot be nonzero. As $i$ is chosen arbitrarily from $\{1, 2, \ldots, d\}$, the derivative of $V$ with respect to any of its arguments is zero at the origin, meaning that, $\nabla_\xb V(\xb)|_{\xb=\mathbf{0}}=0$.
\end{proof}

Each iteration of the RoA estimation algorithm expands the level set of the estimated Lyapunov function to sample from the gap $\Gcal=S_{\alpha c}(V)\backslash S_{c}(V)$ surrounding the current estimate of the RoA for a $\alpha>1$. Both too small and too large gaps are harmful to the stability of the growing RoA estimation algorithm. A too small gap results in too few samples and prolongs the number of growth phases. Moreover, if the gap is too small, it is more likely that all initial states taken from the gap either converges to the equilibrium or diverges. Therefore, the dataset for the optimization problem~\labelcref{eq:roa_update_objective} will be highly skewed that slows down the learning process even further. A too large gap is also harmful as it may advance far beyond the true RoA of the system and many sampled initial states can diverge to unknown and potentially dangerous regions of the state space. As a result, investigating the growth rate of the gap $\Gcal$ as a function of the properties of $V$ is desirable to regularize or prevent harmful sampling behaviors. The following theorem sheds light on this matter.

\begin{theorem} 
  \label{growth_rate_of_level_sets}
  Growth rate of sublevel sets: Assume $V:\Xcal\to\RR$ is a positive definite Lipschitz continuous function on $\Xcal\subseteq\RR^d$. Let $\Scal_c(V)=\{\xb\in \Xcal:V(\xb)\leq c\}$ be the region enclosed by the level set $\partial \Scal_c(V) = \{\xb\in \Xcal:V(\xb)= c\}$ at the level value $c$. If $G\leq\lVert \nabla_\xb V(\xb)\rVert$ for $G\in \RR^{>0}$ and $\xb\in \Scal_c(V)$, then $\partial \mu(\Scal_c(V))\backslash \partial c\propto G^{-1}$.
\end{theorem}  
  
\begin{proof}
  Let $\zb = z \nabla_{\xb} V(\xb)/ \lVert \nabla_{\xb} V(\xb) \rVert$ be a tiny perturbation in the direction of the normal to the level set. $V$ is expanded around $\xb\in \partial \Scal_c(V)$ as
  \begin{align*}
    V(\xb+\zb)  &= V(\xb) + \nabla_\xb V(\xb)\tran \zb + O(z^2)\\
                &= V(\xb) + z \lVert \nabla_\xb V(\xb) \rVert + O(z^2)
  \end{align*}
  where $O(z^2)$ can be ignored for sufficiently small $z = \lVert \zb\rVert$. 
  For $\alpha\in \RR^{>1}$ and sufficiently close to $1$,
  \begin{align*}
    \Scal_{\alpha c}(V) &= \{\xb + \zb: \xb\in \partial\Scal_c(V)\; \text{and}\; V(\xb) + \nabla_\xb V(\xb)\tran \zb \leq \alpha c\}\\
                    &= \{\xb + \zb: \xb\in \partial\Scal_c(V)\; \text{and}\; z\lVert \nabla_\xb V(\xb) \rVert \leq (\alpha-1) c\}.
  \end{align*}
As $G$ is assumed to be a lower bound of $\lVert \nabla_\xb V(\xb)\rVert$, $z\lVert \nabla_\xb V(\xb)\rVert\leq (\alpha-1) c$ implies $z \leq c(\alpha-1) \lVert \nabla_\xb V(\xb)\rVert^{-1} \leq c(\alpha-1)G^{-1}$.

Notice that $\partial \Scal_c(V)$ is a $(d-1)-$dimensional surface that encloses $\Scal_c(V)$ an $d-$dimensional volume that are both embedded in a $d-$dimensional embedding space $\Xcal$.
  We are interested in the volume of $\Gcal:=\Scal_{\alpha c}(V) \backslash \Scal_c(V)$. Assume $\d\omegab$ is the differential form for $\Gcal$. We can write $\d\omegab = \d \sbb \lVert \zb \rVert= z \d \sbb$ where $\d \sbb$ is the surface differential form for $\partial \Scal_c(V)$. Hence,
\begin{equation}
  \mu(\Gcal) = \d\mu(\Scal_c(V)) = \int_\Gcal \d\omegab=\int_{\partial\Scal_c(V)} z \d \sbb \leq c(\alpha -1)G^{-1}\int_{\partial \Scal_c(V)} \d \sbb.
\end{equation}
where $\int_{\partial \Scal_c(V)} \d \sbb$ does not depend on $\alpha$ or $G$. Hence, by pushing $\alpha\to 0$, $\partial \mu(\Scal_c(V)) \backslash \partial \alpha = c \mu(\partial \Scal_c(V)) G^{-1} \propto G^{-1}$ that completes the proof.
\end{proof}

As a result of this theorem, in some applications, one may need to control $\lVert\nabla_\xb V(\xb)\rVert$ for $\xb\in \partial \Scal_c(V)$ to prevent sampling from a too large or too small gap.

\section{Theoretical Motivation of $[V(\xb) - V_{\pi_{n-1}}(f_\pi(\xb))]^2$ in~\Cref{eq:roa_update_objective}}
\label{sec:improved_roa_estimation_constructive_term}
In this section we discuss the theory behind the term $[V(\xb) - V_{\pi_{n-1}}(f_\pi(\xb))]^2$ in~\Cref{eq:roa_update_objective} that we added as an improvement to~\citep{richards2018lyapunov} to facilitate learning the RoA. The objective function for the RoA estimation sub-phase is restated here:

{\small
\begin{equation*}
  \Lcal(V) = \sum_{\xb\in\XX^\text{IN}}[V(\xb) - \bar{c}] - \sum_{\xb\in\XX^\text{OUT}}[V(\xb) - \bar{c}] +\lambda_\text{RoA}\sum_{\xb\in\XX^\text{IN}} \Delta  V(\xb) + \lambda_\text{monot}\sum_{\xb\in\XX^\text{IN}} [V(\xb) - V_{\pi_{n-1}}(f_{\pi_{n-1}}(\xb))]^2.
\end{equation*}
}

The first two terms construct a classifier objective. The third term is added to conform with the Lyapunov decrease conditions. Here, we focus on the last term, i.e., $[V(\xb) - V_{\pi_{n-1}}(f_\pi(\xb))]^2$. The motivation behind adding this term comes from the constructive method proposed by~\cite{chiang1989stability}.

The idea of the constructive methodology of~\cite{chiang1989stability} is to construct a sequence of functions $V_0, V_1, \ldots$ for the autonomous dynamical system $\xb_{k+1}=f(\xb_k)$ in order to use the level sets of the accumulating function of this sequence for estimating the RoA of the vector field $f$. The theory is developed for a class of functions more general than Lyapunov functions that are called energy-like functions. An energy-like function decreases over the trajectories of the system (see~\cite{chiang2015stability} for a precise definition). Given an energy-like function $V_0$ for the vector field $f$, the following sequence of functions is constructed
\begin{align}
  V_1(\xb) &= V_0(\xb + \epsilon_1 f(\xb))\label{eq:constructive_process}\\
  V_2(\xb) &= V_1(\xb + \epsilon_2 f(\xb))\nonumber\\
  &\ldots\nonumber\\
  V_n(\xb) &= V_{n-1}(\xb + \epsilon_n f(\xb))\nonumber
\end{align}

where $d_i, i=1,2,\cdots,n$ are positive numbers. Two facts about this sequence must be proved: {\bf 1)} All functions produced in this sequence are energy-liked functions. {\bf 2)} For a fixed positive $c$, $S_c(V_i)\subset S_c(V_{i+1})$. The following two theorems guarantee these points.

\begin{theorem}[Energy-like functions, lemma $4.2$ in~\cite{chiang1989stability}]
  \label{thm:energy_like_functions}
  Let $V:\RR^d\to \RR$ be an energy-like function for the nonlinear autonomous system $\xb_{k+1}=f(\xb_k)$ with the equilibrium point $\bar{\xb}=\mathbf{0}$. Let $\Dcal$ be a compact set around $\bar{\xb}$ that contains no other equilibrium points. Then, there exists an $\tilde{\epsilon}>0$ such that for $\epsilon<\hat{\epsilon}$, the function $V_1 = V(\xb + \epsilon f(\xb))$ is also an energy-like function on the compact set $\Dcal$ for the vector field $f$.
\end{theorem}
This theorem guarantees that all functions in the constructive process~\labelcref{eq:constructive_process} are energy-like functions.

\begin{theorem}[Monotonic sublevel sets, lemma $4.1$ in~\cite{chiang1989stability}]
  \label{thm:monotonic_levelsets}
  Let $V:\RR^d\to \RR$ be an energy-like function for the nonlinear autonomous system $\xb_{k+1}=f(\xb_k)$ with the equilibrium point $\bar{x}=\mathbf{0}$. Let $\Dcal$ be a compact set around $\bar{\xb}$ that contains no other equilibrium points. Assume the set $S_c(V):=\{\xb: V(\xb)\leq c\text{ and }\xb\in\Dcal\}$ is non-empty for some constant $c$. Then there exists an $\tilde{\epsilon}>0$ such that for the set characterized by $S_c(V_1):=\{\xb: V_1(\xb)\leq c\text{ and }\xb\in\Dcal\}$ where $V_1(\xb)=V(\xb+\epsilon f(\xb))$ and $\epsilon<\tilde{\epsilon}$, the following holds
  \begin{equation}
    S_c(V)\subset S_c(V_1).
  \end{equation}
\end{theorem}
This theorem guarantees that the sequence built by the constructive process~\labelcref{eq:constructive_process} gives a monotonically increasing sequence of sublevel sets.

Now, we get back to the addded term to the objective function, i.e., $[V(\xb) - V_{\pi_{n-1}}(f_{\pi_{n-1}}(\xb))]^2$. Notice that $V_{\pi_{n-1}}(\cdot)$ is fixed and the minimization is performed with respect to $V(\cdot)$. Hence, minimizing the above term at each phase of the algorithm emulates the above constructive process and encourages the learned Lyapunov functions to be monotonic in the sense that their sublevel sets $\Scal_c(V_n)$ tend to become a monotonically increasing sequence that covers more and more space of the true RoA.

\section{Algorithms}
\label{sec:algorithms}
The proposed algorithms in this work are verbally described in~\Cref{sec:algorithm}. To facilitate implementation, the detailed pseudo-code of the algorithms come here. The RoA estimation sub-phase is realized as~\Cref{alg:roa_estimation} and the policy update sub-phase is realized as~\Cref{alg:policy_update}.
\RestyleAlgo{boxed}

\begin{algorithm}[h!] 
  \caption{RoA estimation: Learn $\Rcal_{\pi_{n}}$ from $(\Rcal_{\pi_{n-1}}, f_{\pi_{n}})$}
  \label{alg:roa_estimation}
  \SetKwInOut{Input}{input}\SetKwInOut{Output}{output}
  \Input{$(V_{\pi_{n-1}}, c_{n-1})$: The Lyapunov function and level value of phase $n$ where | $\Rcal_{\pi_{n-1}}=\Scal_{c_{n-1}}(V_{\pi_{n-1}})$ | $f_{\pi_{n}}$: Closed-loop system vector field  | $\gamma_r>1$: Level value multiplicative factor | $N\in \NN$: Number of sampled states | $M\in\NN$: Number of phases | $0\leq\beta_r\leq 1$: Mixture parameter | $L_r\in \NN$: Trajectory length  | $\Dcal$: Domain | $\lambda_{\rm RoA}$: Negative definiteness weighting factor | $\lambda_{\rm monot}$: Monotonicity weighting factor}

  \Output{$(V_{\pi_n}, c_n)$}
  Init $\hat{V}$ to $V_{\pi_{n-1}}$ and $\hat{c}$ to $c_{n-1}$\\
  Init the sampling distribution $p_r$ to $U(\Dcal)$, i.e., uniform distribution over the domain $\Dcal$\\
  \For{$m=1,\ldots,M$}{
        $\Gcal \leftarrow \Scal_{\gamma_r \hat{c}}(\hat{V})\backslash \Scal_{\hat{c}}(\hat{V})$\\
        $p_r\leftarrow \beta_r U(\Gcal) + (1-\beta_r) U(\Dcal)$\\
        $\XX_0\leftarrow$ Generate $N$ samples from $p_r$\\
        $\XX_{L_r}\leftarrow$ Run $f_{\pi_{n}}$ on $\XX_0$ for $L_r$ steps\\
        $\XX_0^\text{IN}\leftarrow \{(\xb, 1): \xb\in\XX_0, \Phi(\xb,L_r)\in  \Scal_{\hat{c}}(\hat{V})\}$\\
        $\XX_0^\text{OUT}\leftarrow \{(\xb, 0): \xb\in\XX_0, \Phi(\xb,L_r)\notin  \Scal_{\hat{c}}(\hat{V})\}$\\
        $V^*\leftarrow$ Optimize for $V$ in the objective function~\labelcref{eq:roa_update_objective} using the dataset $\{\XX_0^\text{IN}, \XX_0^\text{OUT}\}$\\
        $\hat{c}\leftarrow\argmax_c \{c\in\RR: \Delta f_{\pi_{n-1}}(\xb) < 0\;{\rm for}\;\xb\in\Scal_c(V^*)\}$\\
        $\hat{V}\leftarrow V^*$
    }
$(V_{\pi_n}, c_n)\leftarrow (\hat{V}, \hat{c})$
\end{algorithm}

\begin{algorithm}[t!] 
  \caption{Policy update: Learn $\pi_{n+1}$ from $(\Rcal_{\pi_n}, f_{\pi_{n}})$}
  \label{alg:policy_update}
  \SetKwInOut{Input}{input}\SetKwInOut{Output}{output}
  \Input{$(V_{\pi_n}, c_n)$: The Lyapunov function and level value of phase $n$ where $\Rcal_{\pi_n}=\Scal_{c_n}(V_{\pi_n})$ | Closed-loop system vector field $f_{\pi_n}$ | $\gamma_p>1$: Level value multiplicative factor | $N\in \NN$: Number of sampled states | $0\leq\beta_p\leq 1$: Mixture parameter  | $L_p\in \NN$: Trajectory length | $\lambda_u$: Unstable states weighting factor}
  \Output{$\pi_{n+1}$}
  Init sampling distribution $p_p$ to $U(\Scal_{c_n}(V_{\pi_n}))$, i.e., uniform distribution over $\Scal_{c_n}(V_{\pi_n})$\\
  $\Gcal \leftarrow \Scal_{\gamma_p c_n}(V_{\pi_n})\backslash\Scal_{c_n}(V_{\pi_n})$\\
  $p_p\leftarrow \beta_p U(\Gcal) + (1-\beta_p) U(\Scal_{c_n}(V_{\pi_n}))$\\
  $\XX_0\leftarrow$ Generate $N$ samples from $p_p$\\
  $\pi_{n+1}\leftarrow$ Optimize for $\pi$ in the objective function~\labelcref{eq:policy_update_objective} using the dataset $\{\XX_0\}$
\end{algorithm}

\section{Model-Based Assumption}
\label{sec:model_based_assumtion}
The model of the system is used to produce the trajectories of the system which is needed in evaluating both objective functions in~\labelcref{eq:policy_update_objective} and~\labelcref{eq:roa_update_objective}. However, the information needed to compute~\labelcref{eq:roa_update_objective} can also be obtained by experimentation with the physical system without requiring the dynamics function of the system. In the following, we show that the knowledge of the model of the system can be relaxed in both sub-phases of the algorithm.
\subsection{RoA Estimation Sub-Phase}
\label{sec:model_based_assumption_in_roa_estimation_phase}
As can be seen in~\Cref{eq:roa_update_objective} and~\Cref{alg:roa_estimation}, in this sub-phase, the model is used to produce trajectories starting from the sampled initial states from the gap around the current estimate of the RoA or from withing the RoA. In both cases, as also suggested by~\cite{richards2018lyapunov}, instead of the model, the real system can be used to produce the trajectories. Assume $\gamma_r$ in~\Cref{alg:roa_estimation} satisfies the conditions of~\Cref{growth_rate_of_level_sets} of the Appendix. Therefore, the gap from which the initial states are sampled is not too large and a hyperbolic physical system does not exhibit a drastically different and potentially dangerous behavior when it is launched from initial states sampled from this surrounding gap. Hence, in an approach similar to active learning, the real system can be used to produce the trajectories and label them based on whether they enter the current target level set of the estimated Lyapunov function or not. In conclusion, in this sub-phase, there will be no need to know the model of the system or estimate it.

\subsection{Policy Update Sub-Phase}
\label{sec:model_based_assumption_in_policy_estimation_phase}
The model requirement is a bit different in this sub-phase compared with the RoA estimation sub-phase. As we need to update the policy, information on the way the behavior of the system changes with respect to a change in the policy is required. However, looking at~\labelcref{eq:policy_update_objective}, it is observed that the behavior of the system influences the objective function only via the Lyapunov function $V(\Phi(\xb,L_p))$. Suppose the policy is parameterized as $\pi(\xb;\phi)$. Then, the closed-loop system becomes $\xb_{k+1}=f(\xb_k, \pi(\xb_k))$. The required gradient to update the policy parameters contains the term $\partial V(\xb_k) / \partial \psi$ where the information of $\psi$ is encoded in $\xb_k=\Phi(\xb_0, k)$. The derivative decomposes as
\begin{equation}
  \frac{\partial V(\xb_k)}{\partial \psi} = \frac{\partial V(\xb_k)}{\partial \xb_k}\frac{\partial \xb_k}{\partial \psi}.
\end{equation}
where the knowledge of the model is required to compute $\partial \xb_k / \partial \psi$ that will be a matrix of size $\dim(\Xcal)\times \dim(\psi)$. When the model is not given, one must instead try to estimate the model of the system $\hat{f}:\RR^d\to\RR^d$. However, this vector-valued function is difficult to estimate unless in very limited cases. Instead, one can try to directly estimate $V$ rather than $f$ as a function of $\psi$. As $V$ is a scalar-valued function, it takes fewer trajectories to obtain a decent estimate of $\partial V / \partial \psi$. Hence, even though the model of the system is needed for this sub-phase, there are two factors that relaxes this requirement: {\bf 1)} As the trajectories are integrated forward only for $L_p$ steps in~\Cref{alg:policy_update}, a local estimation would be sufficient. {\bf 2)} Even in the local estimatation regime, one does not need to estimate the nonlinear vector field as an $\RR^d$ to $\RR^d$ function. What matters is how the vector field looks like through the lens of the Lyapunov function that is a scalar-valued function. Hence, an $\RR^d\to \RR^d$ estimation problem can be replaced by an $\RR^d\to\RR$ estimation problem.

\section{Weak Learning Signal}
\label{sec:weak_learning_signal}
In this section, we take a closer look at the occasions that the learning signal for the policy update sub-phase of the algorithm (\Cref{sec:algorithm}) is weak. We discuss the problem for a general setting and show that the setting of this paper is a special case.
\paragraph{Problem Statement.} Assume the discrete-time time-invariant non-autonomous dynamical system $\xb_{k+1}=f(\xb_k, \ub_k)$. The goal is to design $\ub_k$ for $0\leq k\leq T$ such that $\xb_k$ meets some specified conditions for $0\leq k\leq T$. In the Lyapunov stability analysis, these conditions are assessed by a function $V:\Xcal\to \RR^{\geq0}$, i.e., after rolling out the starting state $\xb_0$ for $T$ steps by the dynamics $f$ controlled by $\ub_k$, $V(\xb_{-T})\in \Scal$ where $\Scal$ encapsulates the desired conditions for the trajectory $\xb_{-T}=(\xb_0, \xb_1, \ldots, \xb_T)$ of the system. Assume the deviation of $V(\xb_{-T})$ from its desired set $\Scal$ is measured by a deviation metric $\Lcal(V(\xb_{-T}), \Scal)$. In many control tasks such as \emph{tracking}, the entire trajectory matters, and $\Lcal$ will be a function of every state in the trajectory. However, in control tasks such as \emph{reaching}, only the final state $\xb_T$ matters; consequently the objective function $\Lcal$ only depends on $\xb_T$. Stability, in the presence of a Lyapunov function, can be seen as an example of the second class of tasks where the relative position of $\xb_T$ compared to the level sets of the Lyapunov function is sufficient to decide the convergence of that trajectory. In practice, the control signal $\ub_k$ is produced as a parametric function, i.e. $\ub_k = \pi(\xb_k;\psi)$. Therefore, the entire trajectory $\xb_{-T}(\psi)$ is now parameterised by $\psi$ and so is the loss function $\Lcal(\psi)$. The class of algorithms known as \emph{policy gradient} uses $\partial \Lcal(\psi) / \partial \psi$ to learn the controller $\pi(\cdot;\psi)$. In this section, we study the condition of this gradient and show under which circumstances it vanishes and results in slow convergence.

In a general case where $V$ in $\Lcal(V(\xb_{-T}), \Scal)$ is a function of the entire trajectory $\xb_{-T}$, the gradient w.r.t. the controller parameters is written as:
\begin{align}
  \frac { \partial \Lcal } { \partial \psi } &= \sum _ { 1 \leq p \leq T } \frac { \partial \Lcal _ { p } } { \partial \psi }\label{eq:total_loss_derivative}\\
  \frac { \partial \Lcal _ { k } } { \partial \psi } &= \sum _ { 1 \leq p \leq k } \left( \frac { \partial \Lcal} { \partial \mathbf { x } _ { k } } \frac { \partial \mathbf { x } _ { k } } { \partial \mathbf { x } _ { p } } \frac { \partial ^ { + } \mathbf { x } _ { k } } { \partial \psi } \right)\label{eq:pertime_loss_derivative}\\
  \frac { \partial \mathbf { x } _ { k } } { \partial \mathbf { x } _ { p } } &= \prod _ { p < i \leq k } \frac { \partial \mathbf { x } _ { i } } { \partial \mathbf { x } _ { i - 1 } }
  = \prod _ { p < i \leq k } \left(\frac {\partial f} {\partial \xb}|_{\xb = \xb_{i-1}} + \frac {\partial f} {\partial \ub}|_{\ub = \pi(\xb_{i-1})}\frac{\partial \pi(\xb)}{\partial \xb}|_{\xb=\xb_{i-1}}\right)\label{eq:dxk_dxp}
\end{align}
The gradients are derived as sum-of-products and $\frac { \partial ^ { + } \mathbf { x } _ { k } } { \partial \psi }$ refers to the \emph{immediate} partial derivate of state $\xb_k$ with respect to $\psi$ when $\xb_{k-1}$ is considered as a constant.

In the special case where $V$ is the Lyapunov function in $\Lcal(V(\xb_{-T}), \Scal)$ and the concern is the stability of the system, $V(\xb_{-T})=V(\xb_T)$, meaning that the sum on the r.h.s of~\labelcref{eq:total_loss_derivative} will only have one term $\partial \Lcal_T/ \partial \psi$ and~\labelcref{eq:pertime_loss_derivative} will transform to
\begin{equation}
  \frac { \partial \Lcal _ { T } } { \partial \psi } = \frac { \partial \Lcal  } { \partial \mathbf { x } _ { T } } \sum _ { 1 \leq k \leq T } \left(  \frac { \partial \mathbf { x } _ { T } } { \partial \mathbf { x } _ { k } } \frac { \partial ^ { + } \mathbf { x } _ { k } } { \partial \psi } \right).\label{eq:finaltime_loss_derivative}
\end{equation}
In the following, we analyse the role of each term that contributes to this gradient.

\subsection{The Component $\frac { \partial \Lcal  } { \partial \mathbf { x } _ { T } }$} 

This term concerns the differential condition of $\Lcal(\xb)=\Lcal(V(\xb), \Scal)$ at $\xb=\xb_T$ as
\begin{equation}\label{eq:dl_dxT}
  \frac{ \partial \Lcal(\xb) }{\partial \xb}|_{\xb = \xb_T} = \frac{\partial \Lcal(V)}{\partial V}|_{V=V(\xb_T)} \frac{\partial V(\xb)}{\partial \xb}|_{\xb = \xb_T}
\end{equation}
If any of these terms gets too smalll, the overall gradient gets small too resulting in a vanishing gradient issue. To avoid this, the algorithm must ensure that these terms remain sufficiently large. For the first term $\partial \Lcal(V) / \partial V$, this condition is normally fulfilled if the deviation metric $\Lcal(\cdot, \cdot)$ is designed properly. For example, one cadidate function for $\Lcal$ would be a signed distance function that measures how far $\xb_T$ is from the maximal stable sublevel set $\partial S_{\cmax}(V)$, e.g., $\Lcal(V(\xb_T), S_{\cmax}(V))= V(\xb_T) - \cmax$. Another candidate deviation metric could be $\Lcal(V(\xb_T), S_{\cmax}(V))= \max((V(\xb_T) - \cmax), 0)$ that does not care about the actual value of $V(\xb_T)$ as long as $\xb_T$ lives within the sublevel set $S_{\cmax}(V)$. For the former case, $\partial \Lcal / \partial V = V$ meaning that the gradient does not vanish as long as $\xb_T\neq \bar{\xb}$.

The second term of the r.h.s of~\labelcref{eq:dl_dxT} depends on the slope of the function $V$ evaluated at the final state of the trajectory. If $V$ is a candidate Lyapunov function, $\nabla_\xb V(\xb)$ vanishes at the equilibrium as we proved in~\Cref{lem:dv_dx_at_equilibrium}. Therefore, if $\xb_T$ is too close to the equilibrium (i.e., the system converges), the learning signal to update the policy will vanish. This point is formalized in the following remark.

\begin{remark}\label{rem:dv_dxT_for_stable_and_unstable_trajectories}
Consider the closed-loop system $\xb_{k+1}=f(\xb_k, \pi(\xb_k))$. Let $\xb_0$ be the initial state of a trajectory that starts from within the RoA of the equilibrium point $\bar{\xb}$ of the closed-loop system denoted by $\Rcal^{\bar{\xb}}_\pi$. Let $V$ be the Lyapunov function and the objective function $\Lcal(V_\pi(\xb_{-T}), \Scal)$ is optimized to update the policy $\pi$. One must be careful not to let the trajectories roll out for too long ($T\gg 1$). As $\xb_T$ gets too close to the equilibrium, the information of the trajectory degenerates (see~\Cref{thm:dv_dxT_for_stable_trajectories}) and the gradient to update the controller vanishes (see~\Cref{lem:dv_dx_at_equilibrium}). On the other hand, If $\xb_0$ does not belong to the RoA of $\bar{\xb}$, it can go too far from $\Rcal^{\bar{\xb}}_\pi$ and may escape the validity domain of the Lyapunov function $V$. Hence, in either case, the length $T$ of the trajectory influences the information content of the trajectory for updating the policy. Both too large and too small values of $T$ must be avoided when the trajectories pass through the Lyapunov function and the Lyapunov function acts as a critic in the algorithm for learning the policy. Too small values of $T$ are non-informative due to the continuity of $V$. Too large values of $T$, on the other hand, is non-informative as the trajectory enters the flat regions of $V$ or escape the domain of validity of $V$. 
\end{remark}

\begin{theorem}\label{thm:dv_dxT_for_stable_trajectories}
  Consider the dynamical system $\xb_{k+1}=f(\xb_k, \ub_k)$ where the control signal is issued by the policy function $\ub_k = \pi(\xb_k;\psi)$ that is parameterised by $\psi$. Let $\bar{\xb}=\mathbf{0}$ be an asymptotic equilibrium point of this system. If $\xb_0\in\Rcal_\pi^\mathbf{0}$ and $V:\Xcal\to\RR^{0\geq}$ is a $C^r$ Lyapunov function with $r\geq 1$, then
  \begin{equation}
    \forall \epsilon > 0, \exists T>0\;\; \text{such that}\;\; \lVert\frac{\partial V(\xb)}{\partial \psi}|_{\xb=\xb_T}\rVert < \epsilon
  \end{equation}
  \begin{proof}
    It can be seen in~\labelcref{eq:finaltime_loss_derivative} that $\nabla_{\xb} V$ appears multiplicatively in $\nabla_{\psi} V$. As stated in~\Cref{lem:dv_dx_at_equilibrium}, $\nabla_{\xb} V(\xb)=0$ at $\bar{\xb}=\mathbf{0}$ that makes $\nabla_{\psi} V$ vanish as well at the equilibrium point. Now suppose $V\in C^r (r\geq 1)$ with respect to both its arguments $\xb$ and $\ub$. Moreover, $\pi(\cdot, \psi)$ is assumed to be smooth with respect to its parameters $\psi$. Therefore, the map $\psi\to\nabla_\psi V$ is continuous. Furthermore, $\xb_k\to \mathbf{0}$ as $k\to\infty$ because $\xb_0\in \Rcal_\pi^\mathbf{0}$. The continouity of the map $\psi\to\nabla_\psi V$ together with the above result on vanishing $\nabla_{\psi} V$ at the equilibrium point completes the proof.
  \end{proof}
\end{theorem}

The effect of the other terms of~\Cref{eq:finaltime_loss_derivative} on the learning signal is discussed below.

\subsection{The Component $\frac{\partial \xb_k}{\partial \xb_p}$}
\label{sec:dxk_dxp}

This term determines how state information propagates forward through the trajectory. More precisely, it shows how perturbing a state at time $p$ affects the downstream state $\xb_k$ after $k-p$ time steps when the parameters of the system are kept fixed. In the following, we take a closer look at the constituent terms of~\labelcref{eq:dxk_dxp}. We first define the function
\begin{equation}
  \label{eq:single_Jacobian_expansion}
  J_i(\xb_i, \ub_i) = [\frac {\partial f} {\partial \xb}|_{\xb = \xb_{i}},  \frac {\partial f} {\partial \ub}|_{\ub = \ub_{i}}]\tran
\end{equation}
that is different form~\labelcref{eq:dxk_dxp} in the sense that $\ub_i$ is not necessarily a function of $\xb_i$. Due to the fact that $\partial \xb_k / \partial \xb_p$ equals the product of $J_i = J_i(\xb_i, \ub_i)$ along the trajectoy $\{(\xb_i, \ub_i)\}_{i=1}^{i=k-1}$, a measure of the size of $J_i$ would be informative about the influence of $\xb_p$ on $\xb_k$. First, we consider the vanishing gradient issue when the influence becomes too small. Notice that the size of $J_i$ as defined in~\labelcref{eq:single_Jacobian_expansion} is determined by two terms as (we drop the index $i$ for convenience)
\begin{equation}
  \label{eq:single_Jacobian_expansion_norm}
  \lVert J(\xb, \ub) \rVert \leq \lVert \frac {\partial f} {\partial \xb} \rVert + \Vert \frac {\partial f} {\partial \ub} \rVert.
\end{equation}

Let $f$ be the dynamics of the closed-loop system, i.e., $f=f_\pi$. If $f$ is Lipschitz continuous (as it is assumed in~\Cref{sec:preliminaries} to ensure the existence and uniqueness of the solution), both terms on the r.h.s. of the inequality~\labelcref{eq:single_Jacobian_expansion_norm} will be bounded. However, the condition of~\Cref{eq:single_Jacobian_expansion} is a more general case when $\ub$ is not necessarily a function of the states. In this case, $f$ must be Lipschitz in both $\xb$ and $\ub$, that is, there exists positive constants $L_{f_\xb}$ and $L_{f_\ub}$ such that:
\begin{align}
  \lVert \partial f(\xb, \ub) / \partial \xb\rVert \leq L_{f_\xb}, \; \forall \ub\in \Ucal\;\text{and}\; \forall \xb\in \Xcal\\
  \lVert \partial f(\xb, \ub) / \partial \ub\rVert \leq L_{f_\ub}, \; \forall \ub\in \Ucal\;\text{and}\; \forall \xb\in \Xcal.
\end{align}
Intuitively, this means that the open-loop dynamics $f(\xb, \ub)$ rolls out smoothly and does not respond too harshly to the changes in the control input. If a component of the system breaks down under some control input, $\ub\in\Ucal$, the above conditions do not hold. In addition, it might be the case that the dynamics of the system show high-frequency vibrations under some specific control inputs (e.g., when the controller excites the natural frequency of the system). These conditions occur rarely in physical systems when the controller remain within a reasonable working regime but may happen frequently under \emph{adversarial} regimes when the system is intentionally attacked by un unauthorized user. Such regimes are beyond the scope of this paper. Hence, we can reasonably assume that $J_i$ is upper bounded by $2\times\max(L_{f_{\xb_i}}, L_{f_{\ub_i}})$.

A more general control signal consists of two parts. The first part is a function of states and the second part is open-loop. Therefore, we have:
\begin{equation*}
  \ub_i = \pi(\xb_i;\psi) + \tilde{\ub}_i\;\; \text{with}\;\; \frac{\partial \tilde{\ub}_i}{\partial \xb_i}=0.
\end{equation*}
We can write the overall dynamics as $\xb_{i+1} = f(\xb_i, \ub_i)=f(\xb_i, (\pi(\xb_i), \tilde{\ub}_i))=f(\xb_i, \tilde{\ub}_i)$ which transforms to the case of~\labelcref{eq:single_Jacobian_expansion} with the difference that the policy function $\pi(\xb)$ is now absorbed in the first component of~\labelcref{eq:single_Jacobian_expansion}. Hence, $J_i(\xb_i, \tilde{\ub}_i)$ decomposes as
\begin{equation*}
  J_i(\xb_i, \tilde{\ub}_i) = [\frac {\partial f} {\partial \xb}|_{\xb = \xb_{i}} + \frac {\partial f} {\partial \pi(\xb)}|_{\pi(\xb) = \pi(\xb_i)},  \frac {\partial f} {\partial \tilde{\ub}}|_{\tilde{\ub} = \tilde{\ub}_{i}}]\tran\label{eq:single_Jacobian_expansion-x}
\end{equation*}
Observe that~\Cref{eq:dxk_dxp} depends only on the first block of $J_i(\xb_i, \tilde{\ub}_i)$ and $\partial f/\partial \tilde{\ub}$ does not affect the gradient even though it affects the trajectory of the system. As a result, the influence of the open-loop control signal $\tilde{\ub}_t$ on $\partial \Lcal / \partial \psi$ is via the final states of the trajectory $\xb_T$ as well as $\partial \Lcal / \partial \xb_T$ in~\labelcref{eq:finaltime_loss_derivative}.

\subsection{The Component $\frac { \partial ^ { + } \mathbf { x } _ { k } } { \partial \psi }$}
\label{sec:dxkplus_dpsi}

This component captures the effect of the policy parameters on the next state of the system when the current state is kept fixed. Using chain rule we have

\begin{equation}
  \label{eq:dxkplus_dpsi}
  \frac { \partial ^ { + } \mathbf { x } _ { k } } { \partial \psi } = \frac{\partial f(\xb_k, \pi)}{\partial \pi}|_{\pi=\pi(\xb_k;\psi)}\times \frac{\partial \pi(\xb_k;\psi)}{\partial \psi}.
\end{equation}

The first term in the r.h.s. depends on how sensitive is the dynamics function $f$ with respect to its control argument. Lower bounding this sensitivity can ensure that the effect of the controller remains visible throughout the trajectory. More rigorously, we assume there exists sufficiently large $\kappa>0$  such that
\begin{equation*}
  \lVert \frac{\partial f(\xb, \pi)}{\partial \pi}|_{\pi=\pi(\xb;\psi)} \rVert>\kappa, \quad \text{ for every $x\in\Dcal$} 
\end{equation*}
where $\Dcal$ is the domain of interest in which the system operates.

The second term on the r.h.s of~\Cref{eq:dxkplus_dpsi} is independent from the rest of the system and shows the sensitivity of the control signal with respect to the parameters of the controller. For example, if the controller is implemented by a neural network with tanh activation functions in the hidden layers, this term can vanish if the weights diverge and push the activation values towards the saturation regimes of tanh.

\section{Training Details}
\label{sec:training_details}
\subsection{Lyapunov Function}
\label{sec:lyapunov_architecture}
A parametric candidate Lyapunov function $V(\cdot;\theta)$ must satisfies the conditions of~\labelcref{eq:lyapunov_conditions}. It must be positive definite on a domain $\Dcal$ and its value must decrease along the trajectories of the system. As the second condition depends on the system, it is embedded in the optimization loss function~\labelcref{eq:roa_update_objective}. The first condition (positive definiteness) though does not depend on the system. It must be enforced for every $\xb\in\Dcal$. Rather than including it in the loss function, we restrict the hypothesis set $V(\cdot;\theta)\in\Hcal$ to the class of positive definite functions. In kernel methods, this property can be achieved by a proper choice of kernels. With neural networks, $V(\cdot;\theta)$ can be represented by $V(\xb;\theta) = v(\xb;\theta)\tran v(\xb;\theta)$ as a Lyapunov candidate function where $v(\xb;\theta)$ is a multilayer perceptron. This guarantees the non-negativeness of $V(\cdot;\theta)$. To ensure $V(\xb;\theta)$ does not vanish at any point other than the origin, first suppose $\zb_{\ell}$ and $\zb_{\ell+1}$ are the input and output of the $\ell^{\rm th}$ layer respectively, i.e., $\zb_{\ell+1} = \zeta_\ell(W \zb_{\ell})$ where $W$ is the weight matrix and $\zeta_\ell(\cdot)$ is the activation function. To guarantee the strict positiveness of $v(\xb;\theta)$ for $\xb\neq \mathbf{0}$, both activation function $\zeta_\ell$ and weight matrix $W_\ell\in\RR^{d_\ell \times d_{\ell+1}}$ must have trivial nullspaces. Activation functions such as tanh and ReLU meet this condition. The weight matrix can be constructed as the following:

\begin{equation}
  \mathbf{W}_{\ell}=\left[\begin{array}{c}
  \mathbf{G}_{\ell 1}^{\top} \mathbf{G}_{\ell 1}+\varepsilon \mathbf{I}_{d_{\ell-1}} \\
  \mathbf{G}_{\ell 2}
  \end{array}\right]
\end{equation}
where $\mathbf{G}_{\ell 1} \in \mathbb{R}^{q_{\ell} \times d_{\ell-1}}$ for some $q_{\ell} \in \mathbb{N}_{\geq 1}, \mathbf{G}_{\ell 2} \in \mathbb{R}^{\left(d_{\ell}-d_{\ell-1}\right) \times d_{\ell-1}}, \mathbf{I}_{d_{\ell-1}} \in \mathbb{R}^{d_{\ell-1} \times d_{\ell-1}}$ is the identity matrix, and $\epsilon\in\RR_+$ is a positive constant to ensure the upper block has full rank. See Remark 1 in the appendix of~\citep{richards2018lyapunov} for further details.

In our experiments, $v(\xb;\theta)$ is represented by a $3$-layer multilayer perceptron each layer of dimension $64$ followed by tanh activation functions and a linear last layer.

\subsection{Discretized Time and State Space}
\label{sec:discretized_statespace}
We discretize the dynamics of the inverted pendulum with time resolution of $\Delta T = 0.01$. The state space is also discretized in the rectangle $[\theta_{\min}, \theta_{\max}] = [-\pi/2, \pi/2], [\omega_{\min}, \omega_{\max}] = [-2\pi, 2\pi]$. Each dimension is divided equally into $100$ sections to form a $2-$dimensional grid. Therefore, all the computations of~\Cref{sec:experiments} takes place with a finite set of states. For example, the true RoA denoted by the green plot in~\Cref{fig:cartoon_and_graphical_RoA}(Right) is computed by integrating forward all states of the state space grid. If the resolution of the grid is too coarse, the bound on the negative definiteness of $\Delta V(\xb)$ must change from $0$ to a more conservative negative value in order to make sure the critical Lyapunov decrease condition will not be violated. The new bound depends on the Lipschitz constants of $V$ and $f$ (see~\cite{richards2018lyapunov} for the detailed derivation of the bound). 

\subsection{Hyper-parameters}
\label{sec:training_hyperparameters}
The pre-training phase is performed with the learning rate $0.001$ and $10k$ training steps. Each run of the RoA estimation algorithm is performed with the learning rate $0.01$ and $10k$ training steps. Each run of the policy update sub-phase is performed with the learning rate $0.01$ with $100$ steps. The number of learning steps in the policy update sub-phase is a proxy to the distance between the current policy and the updated policy. Hence, by keeping this number fairly small, we make sure the condition of~\Cref{assumption:continuously_changing_roa_with_policies} is likely to be fulfilled. 

In~\Cref{eq:roa_update_objective}, $\lambda_{\rm RoA}$ is set to $1000$ to enforce the Lyapunov decrease condition. In the same equation, $\lambda_{\rm monot}$ is set to 0.01 to encourage to monotonicity of the RoA estimation algorithm. In~\Cref{eq:policy_update_objective}, $\lambda_u$ is set to $10$ to put more emphasis on stabilizing the unstable states compared with keeping stable the already stable states.

The sampling mixture parameters $\beta_r$ and $\beta_p$ are both set to $0.6$ in~\Cref{alg:roa_estimation,alg:policy_update} respectively.

The length of the integrated trajectories ($L_r, L_p$) is set to $10$ for both RoA estimation (see \Cref{alg:roa_estimation}) and policy update (see~\Cref{alg:policy_update}) sub-phases. 

The multiplication constants $\gamma_r$ and $\gamma_p$ for both RoA estimation and policy update sub-phases is set to $4$.

The number of RoA estimation sub-phases $m$ in~\Cref{alg:roa_estimation} is set to $20$. The total number of policy update sub-phases is also set to $20$. One can alternatively use a context-aware stopping criterion. For instance, updating the policy can be stopped when no significant change in the RoA or the policy parameters is observed.

The number of sampled initial states $N$ is initialized to $10$ in~\Cref{alg:roa_estimation,alg:policy_update} and increases by $10$ after each update of the policy. The heuristic reason is that after each policy update, the RoA enlarges and it takes more samples to obtain a good representative of the gap surrounding the RoA.

\clearpage
\section{More Experimental Results}
\label{sec:more_experimental_results}

\begin{figure}[ht!]
  \centering
  \begin{subfigure}[t]{0.45\textwidth}
      \centering
      \includegraphics[width=1\textwidth]{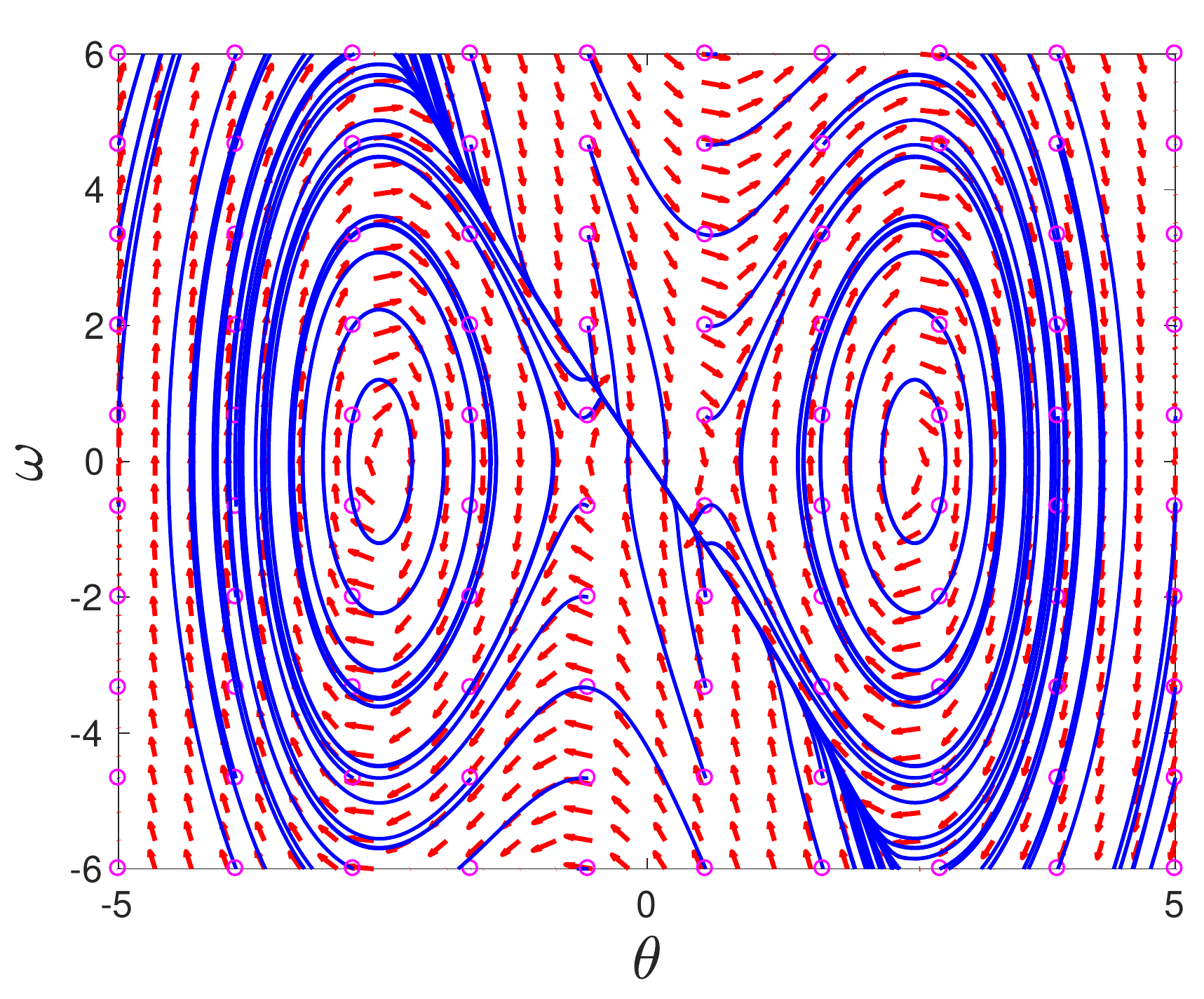}
      \caption{Phase portrait}\label{fig:vectorfield_pendulum}        
  \end{subfigure}
  \begin{subfigure}[t]{0.35\textwidth}
      \centering
      \includegraphics[width=1\textwidth]{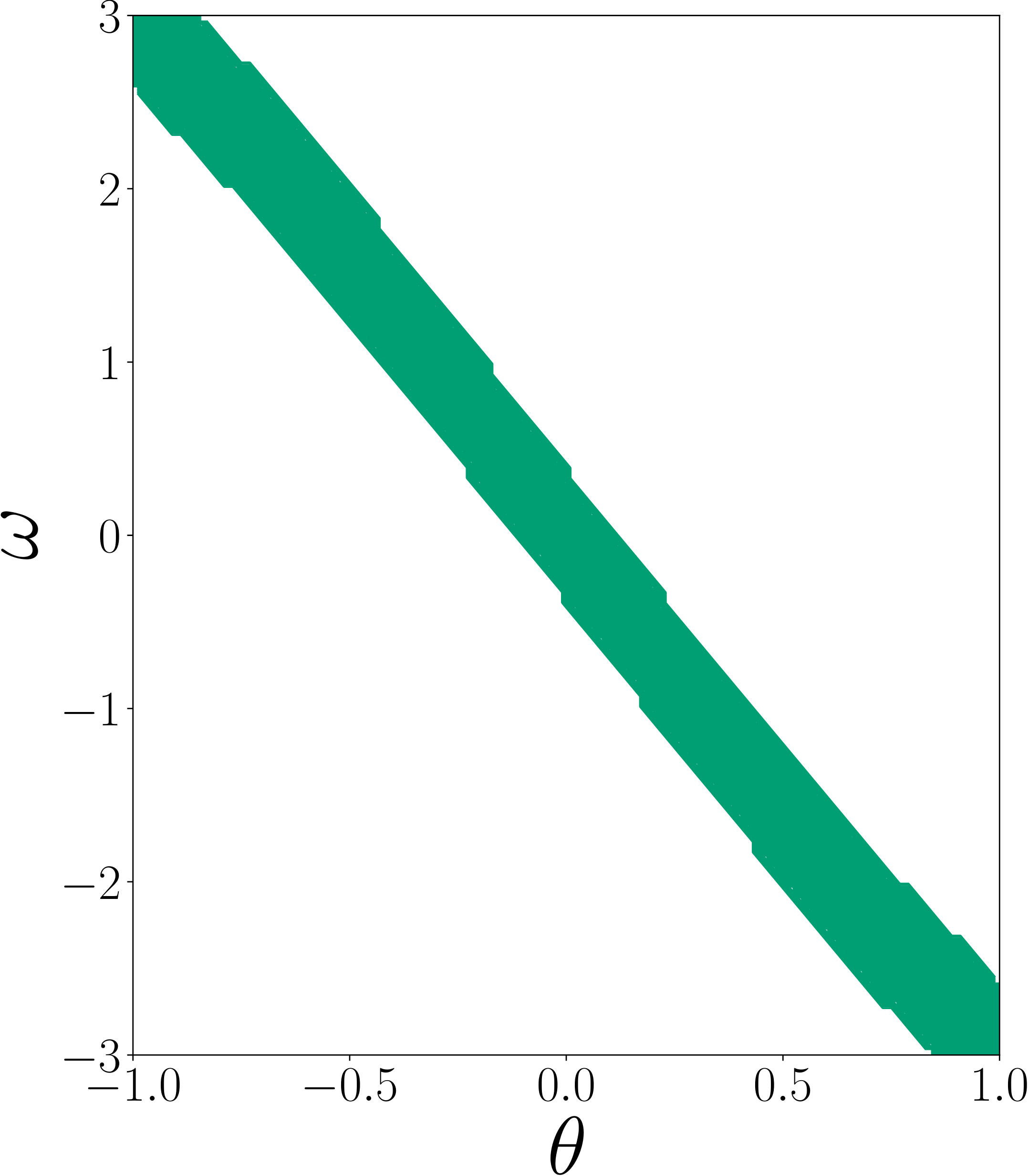}
      \caption{RoA}\label{fig:roa_pendulum_stage_0}        
  \end{subfigure}
  \caption{\footnotesize Dynamics of the inverted pendulum and its initial RoA for an LQR controller}\label{fig:experiment_system_overall}
\end{figure}

\begin{figure}[ht!]
  \label{fig:pretraining_lyapunov_function}
  \centering
  \begin{subfigure}[t]{0.31\textwidth}
      \centering
      \includegraphics[width=1\textwidth]{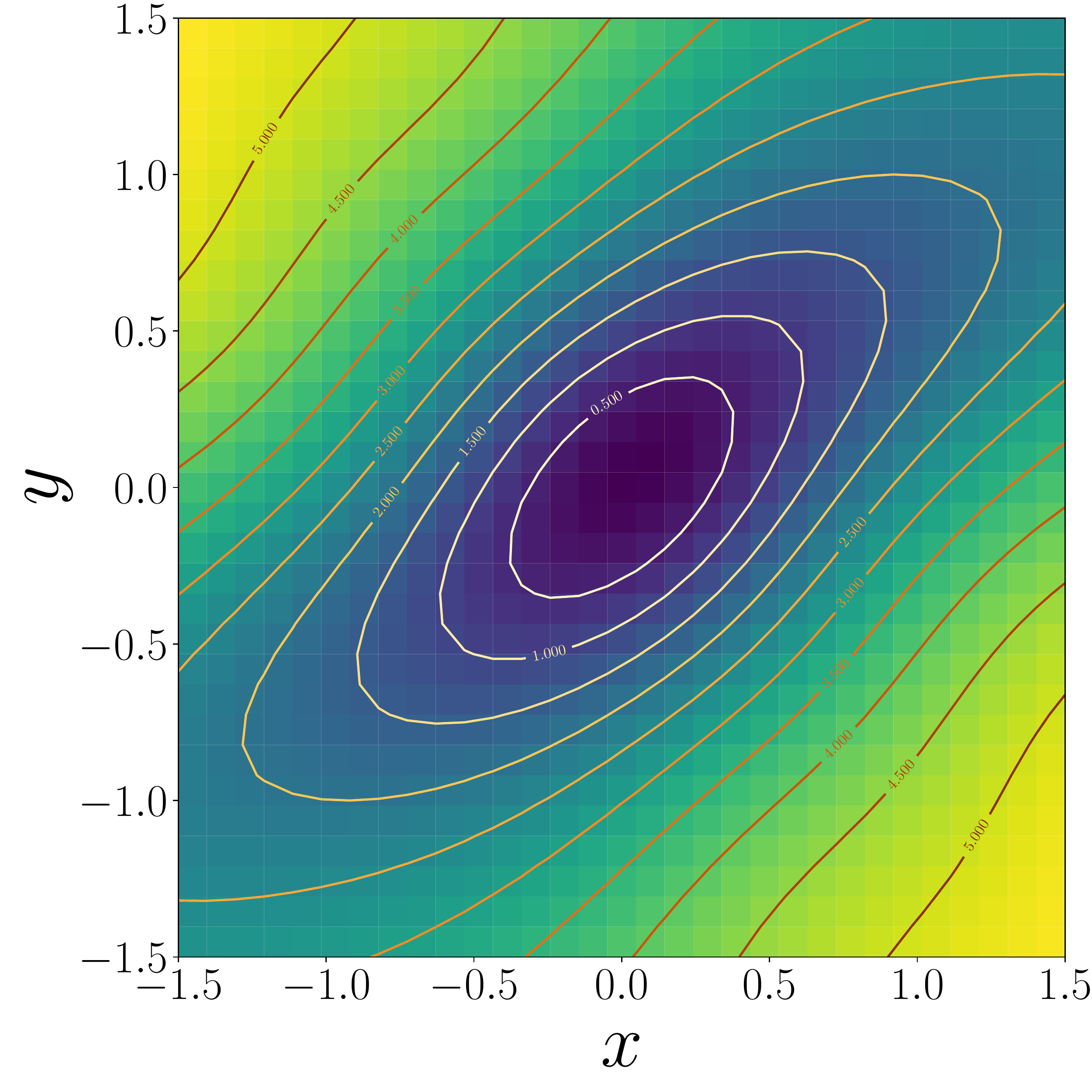}
      \caption{\footnotesize Untrained}\label{fig:levelsets_nn_pre_init}        
  \end{subfigure}
  \begin{subfigure}[t]{0.31\textwidth}
      \centering
      \includegraphics[width=1\textwidth]{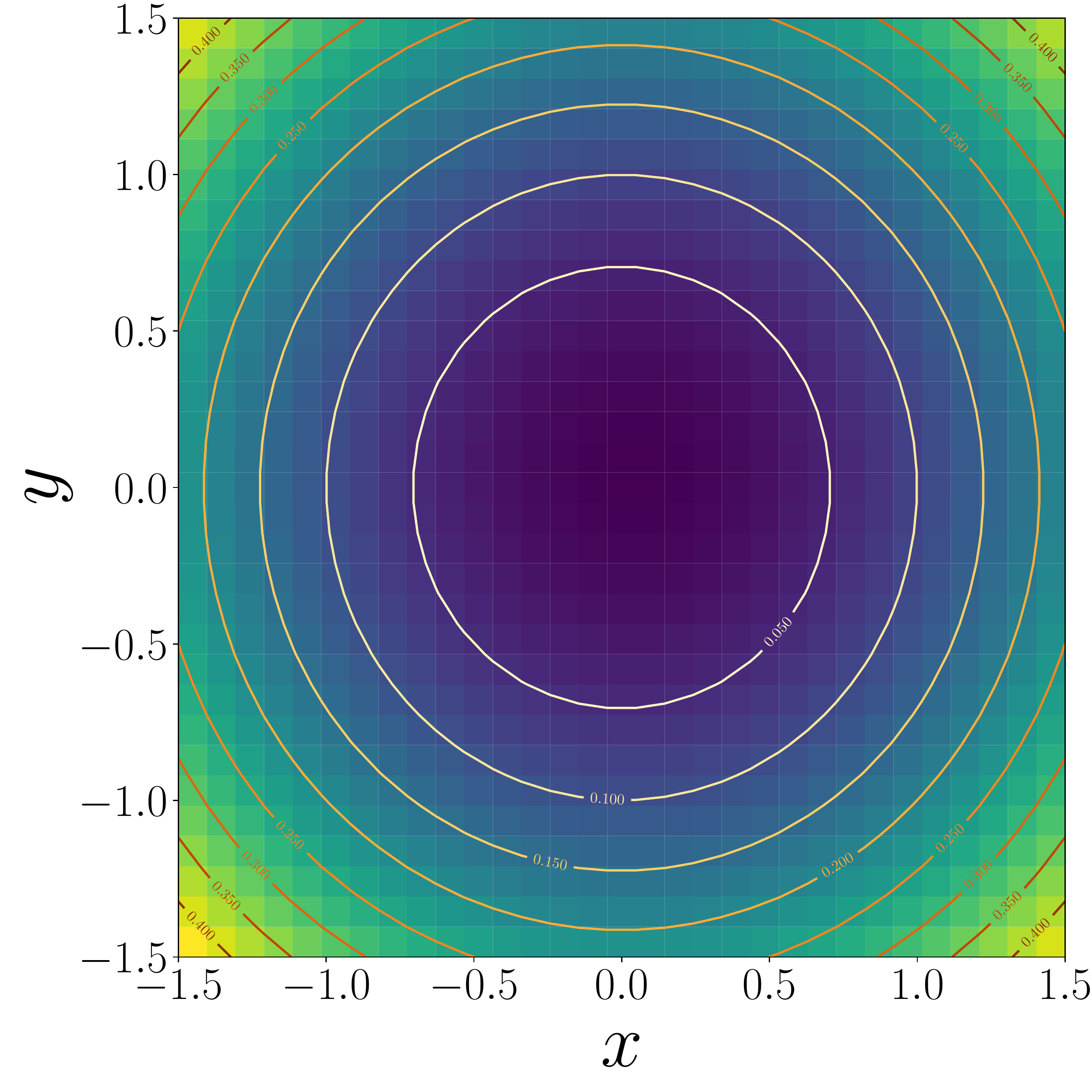}
      \caption{\footnotesize Quadratic}\label{fig:evelsets_quadratic_pre}        
  \end{subfigure}
  \begin{subfigure}[t]{0.31\textwidth}
      \centering
      \includegraphics[width=1\textwidth]{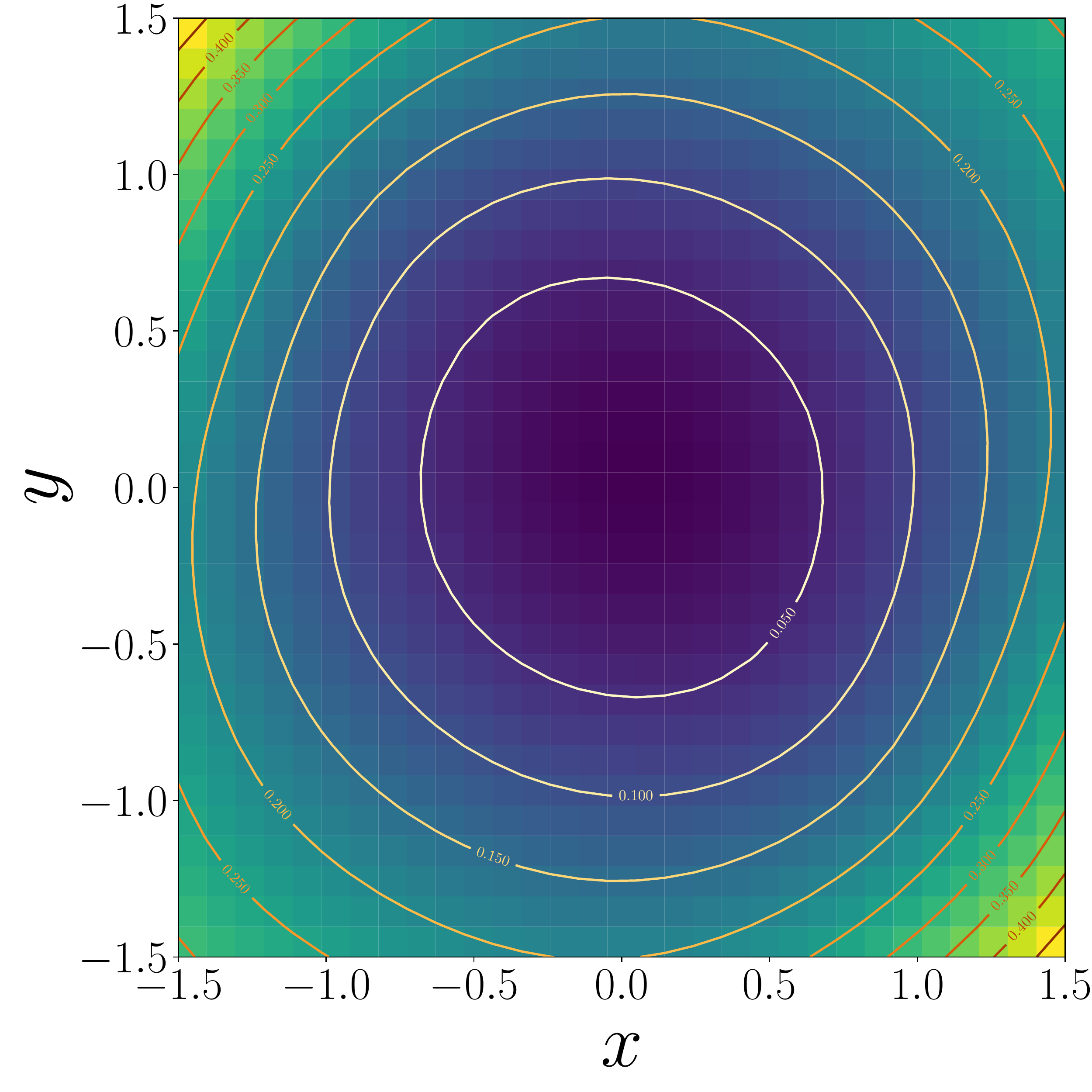}
      \caption{\footnotesize Trained}\label{fig:levelsets_nn_pretrained}    
  \end{subfigure}
  \caption{\footnotesize Pre-training the neural network with a quadratic Lyapunov function. The background heatmap represents the value of the underlying function. Lighter regions correspond to larger values.(a) The level sets of the untrained initialized neural network. (b) The target quadratic function (2) The level sets of the neural network pre-trained with the quadratic function of~\Cref{fig:evelsets_quadratic_pre}.}
\end{figure}

\newcommand{\checkfor}[1]{%
  \ifcsname#1\endcsname%
    \renewcommand{\sfsize}{0.14\textwidth}
  \else%
    \newcommand{\sfsize}{0.14\textwidth}
  \fi%
}

\checkfor{sfsize}

\begin{figure}[h!]
  \centering
  \input{./figs/increasing_roas/tex_files/increasing_nested_roa_figure_stage_1.tex}\\
  \input{./figs/increasing_roas/tex_files/increasing_nested_roa_figure_stage_2.tex}\\
  \input{./figs/increasing_roas/tex_files/increasing_nested_roa_figure_stage_3.tex}\\
  \input{./figs/increasing_roas/tex_files/increasing_nested_roa_figure_stage_4.tex}\\
  \input{./figs/increasing_roas/tex_files/increasing_nested_roa_figure_stage_5.tex}\\
  \input{./figs/increasing_roas/tex_files/increasing_nested_roa_figure_stage_6.tex}\\
  \caption{Visualizing the true ROA which is enlarged by the improved policy and is chased by a learned Lyapunov function. Each row corresponds to a policy update sub-phase and each column corresponds to a RoA estimation sub-phase. At each row, from left to right, the policy is fixed that results in a fixed true RoA (green plot). The columns from left to right are the internal iterations of the RoA estimation sub-phase (see~\Cref{alg:roa_estimation}). The blue color is the estimated RoA $S_{c_n}(V_{\pi_n})$ and the pink color shows the gap $\Gcal=S_{\gamma c_n}(V_{\pi_n})\backslash S_{c_n}(V_{\pi_n})$ that is used in Algorithms~\labelcref{alg:roa_estimation,alg:policy_update} for sampling the initial states. After the RoA estimation sub-phase is done (the rightmost figure of each row), the policy update sub-phase is performed. The leftmost figure in the next row shows that the true RoA enlarges as a result of the policy update. The RoA estimation sub-phase continues from its latest iteration which is a decent initial approximate for the enlarged RoA. As a result of the alternate application of RoA estimation and policy update sub-phases, the bottom rightmost figure shows a significantly larger RoA compared with the top leftmost figure (see the green boundary of the true RoA around the blue region that is the estimated RoA by the learned Lyapunov function). Note that the Lyapunov function is also learned such that its sublevel set (blue region) almost perfectly matches the true RoA (green boundary) in the bottom rightmost figure.}\label{fig:pendulum_nested_roas_full}
\end{figure}
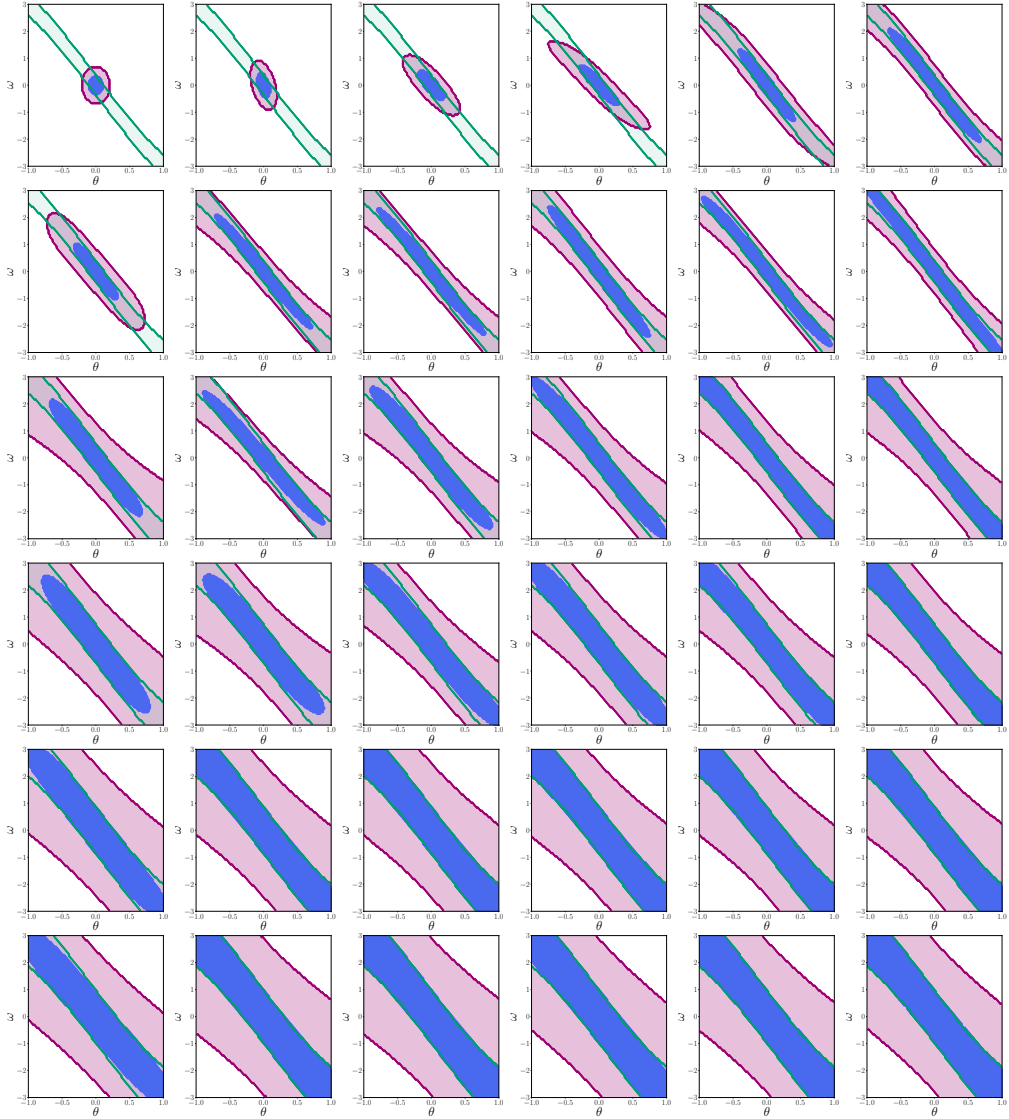

%% file: figs/increasing_roas/tex_files/increasing_nested_roa_figure_stage_5.tex
\begin{subfigure}[t]{\sfsize}
    \centering
    \includegraphics[width=1\textwidth]{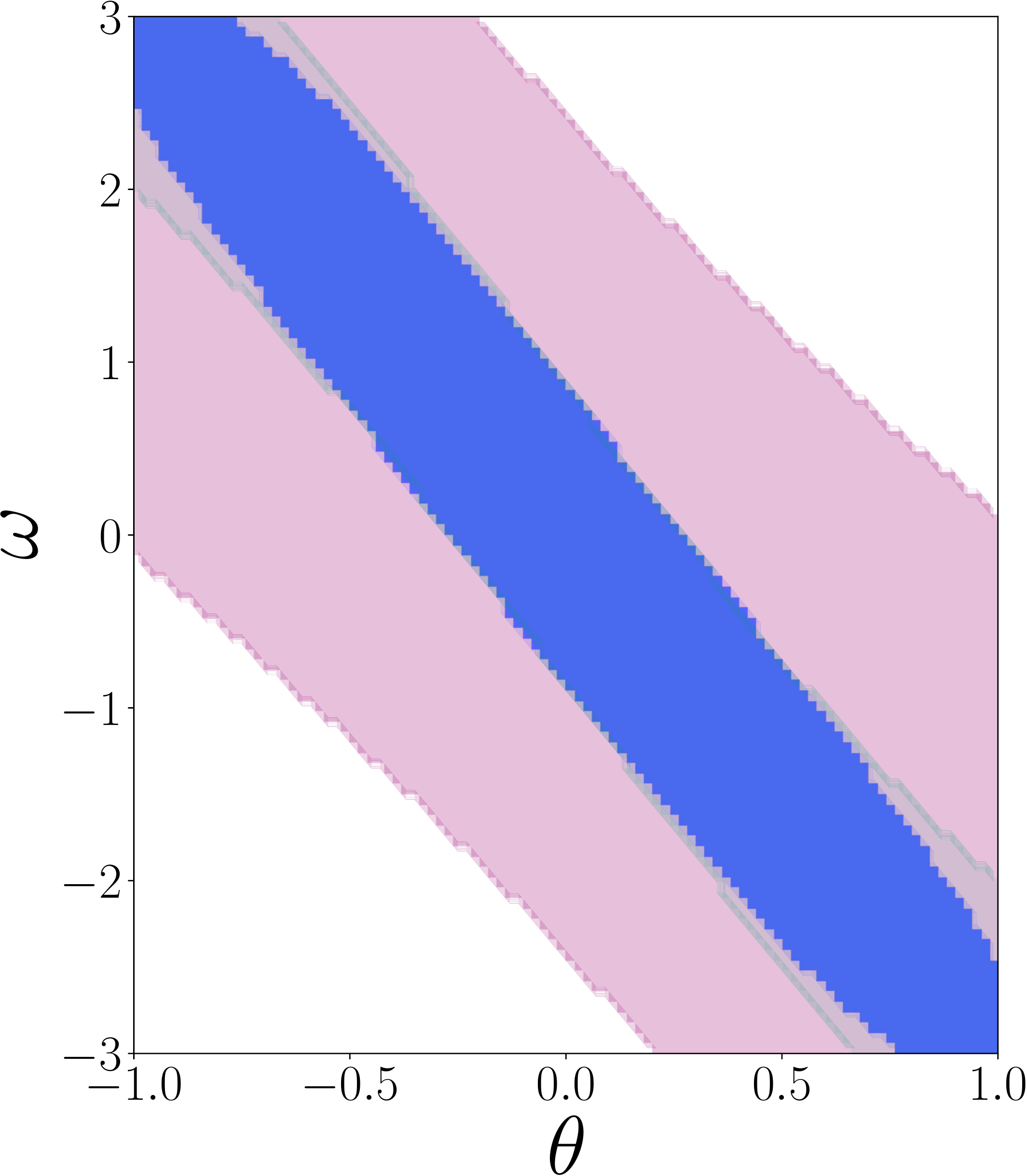}
\end{subfigure}
\begin{subfigure}[t]{\sfsize}
    \centering
    \includegraphics[width=1\textwidth]{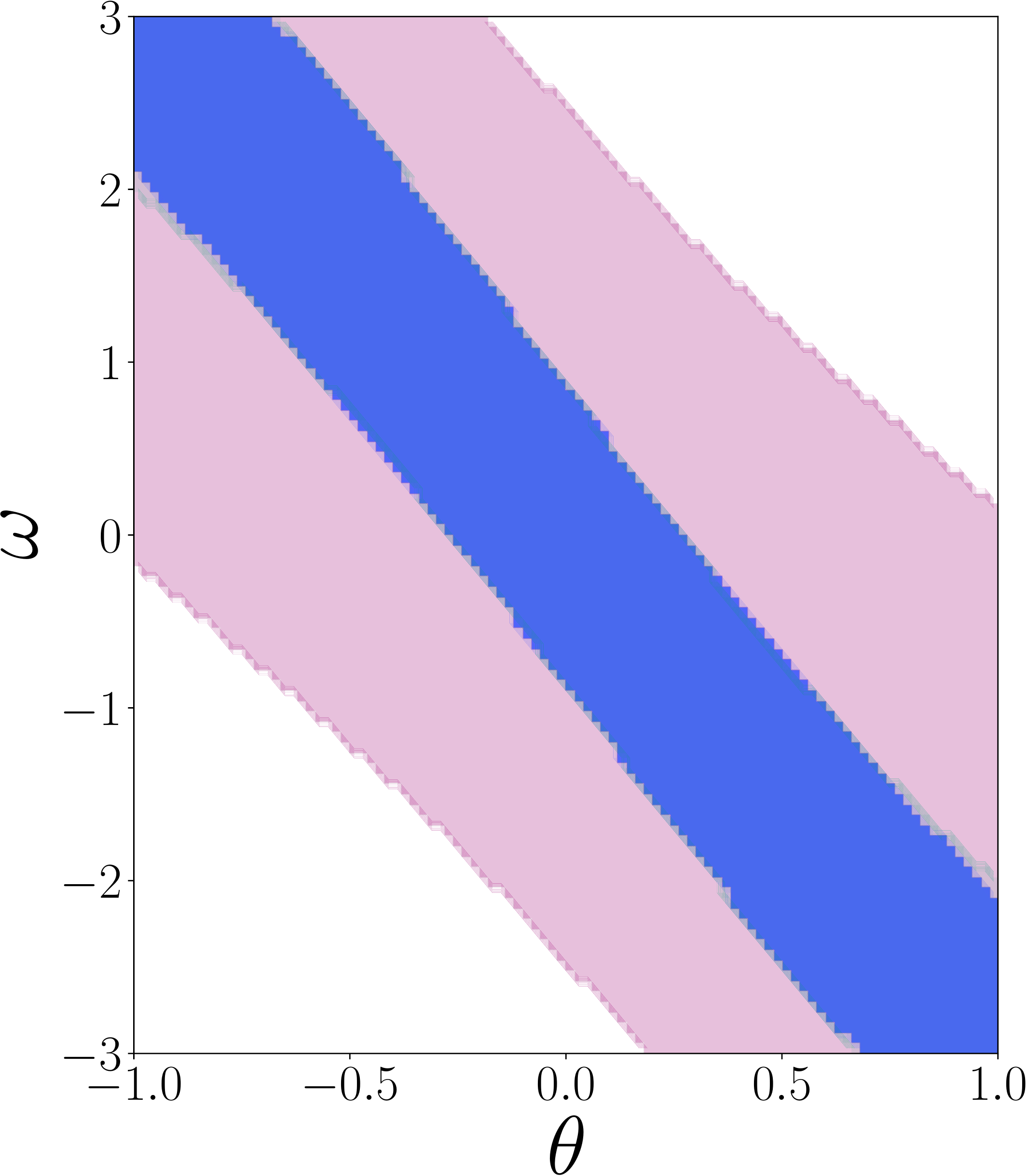}
\end{subfigure}
\begin{subfigure}[t]{\sfsize}
    \centering
    \includegraphics[width=1\textwidth]{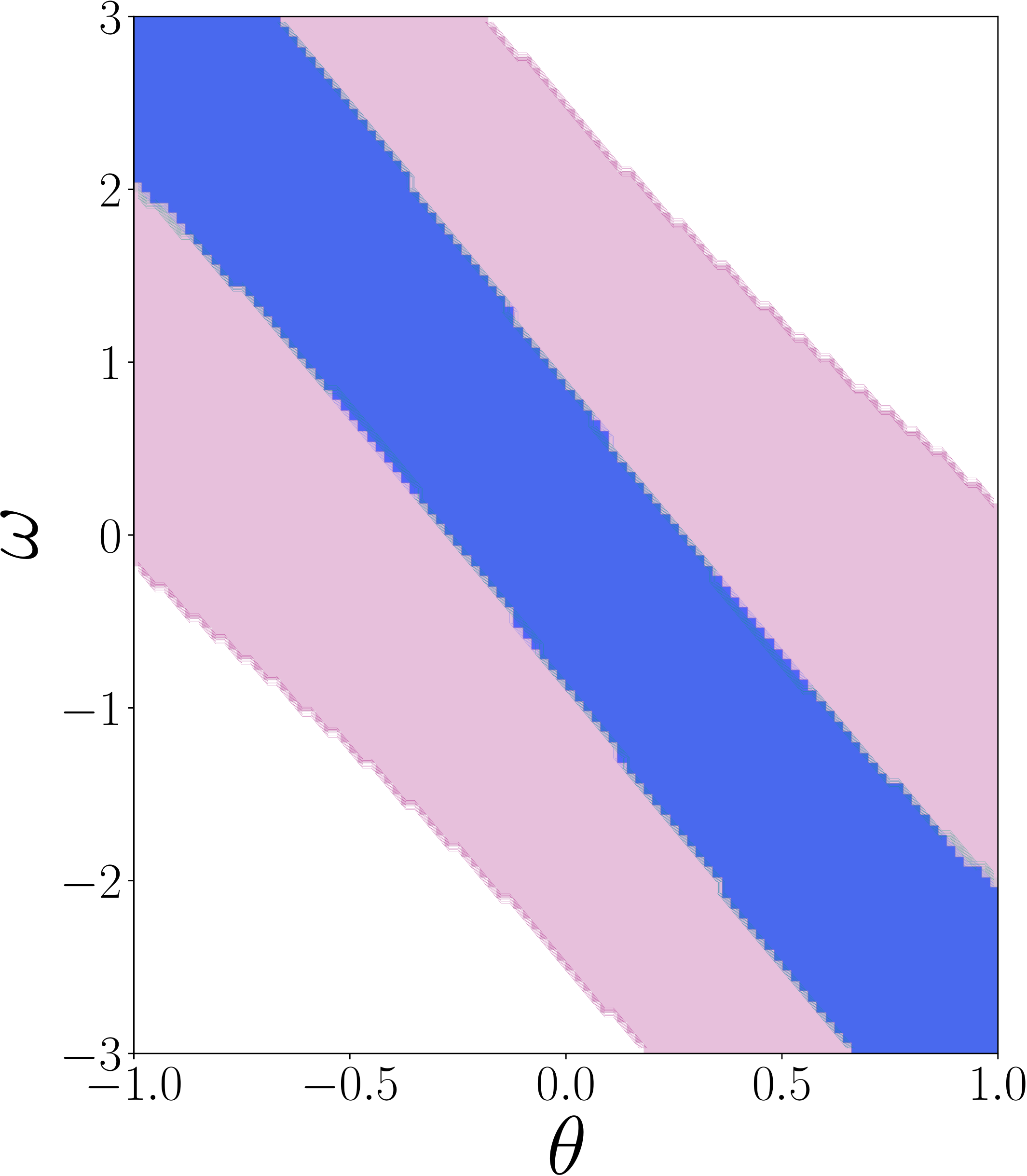}
\end{subfigure}
\begin{subfigure}[t]{\sfsize}
    \centering
    \includegraphics[width=1\textwidth]{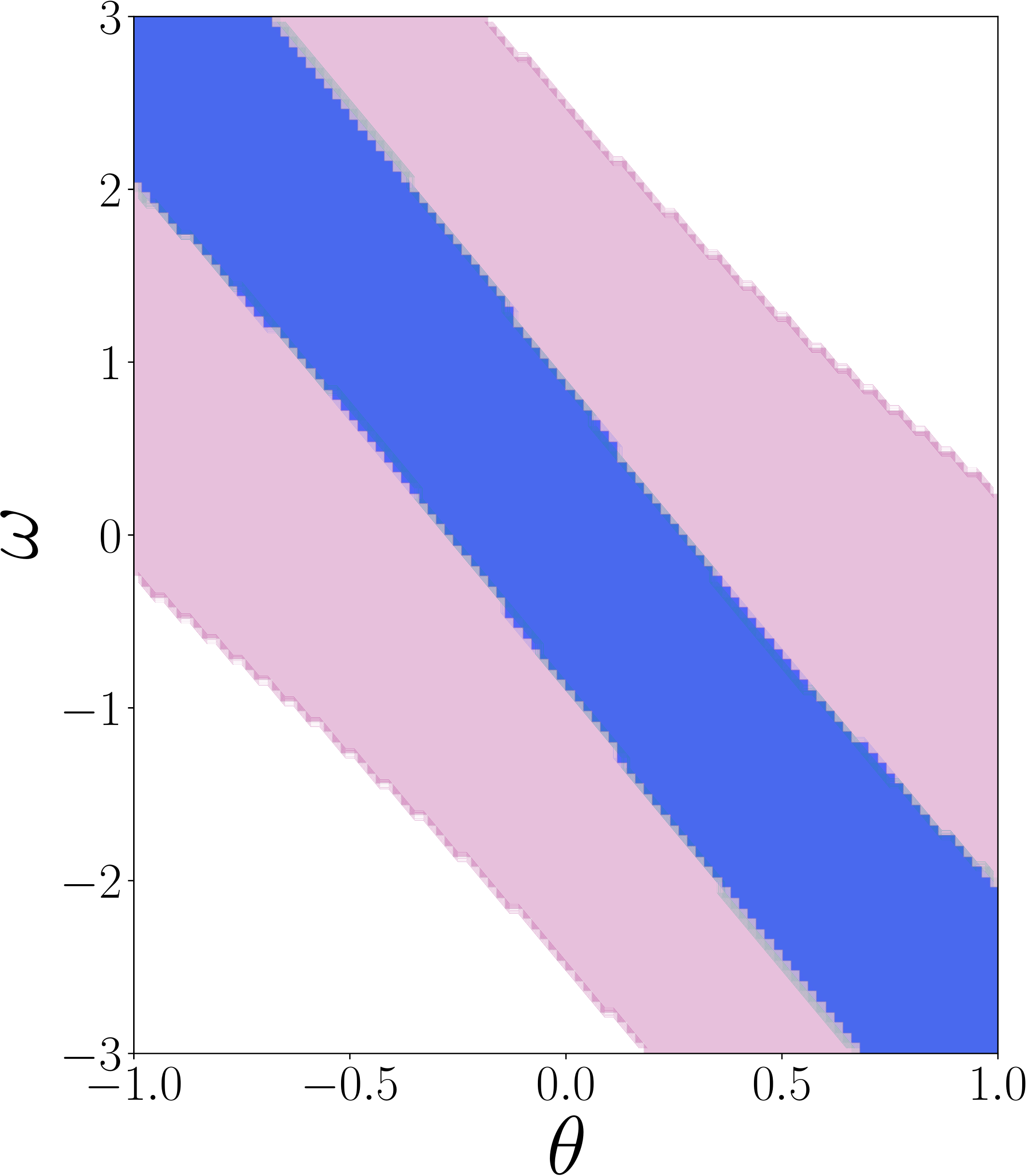}
\end{subfigure}
\begin{subfigure}[t]{\sfsize}
    \centering
    \includegraphics[width=1\textwidth]{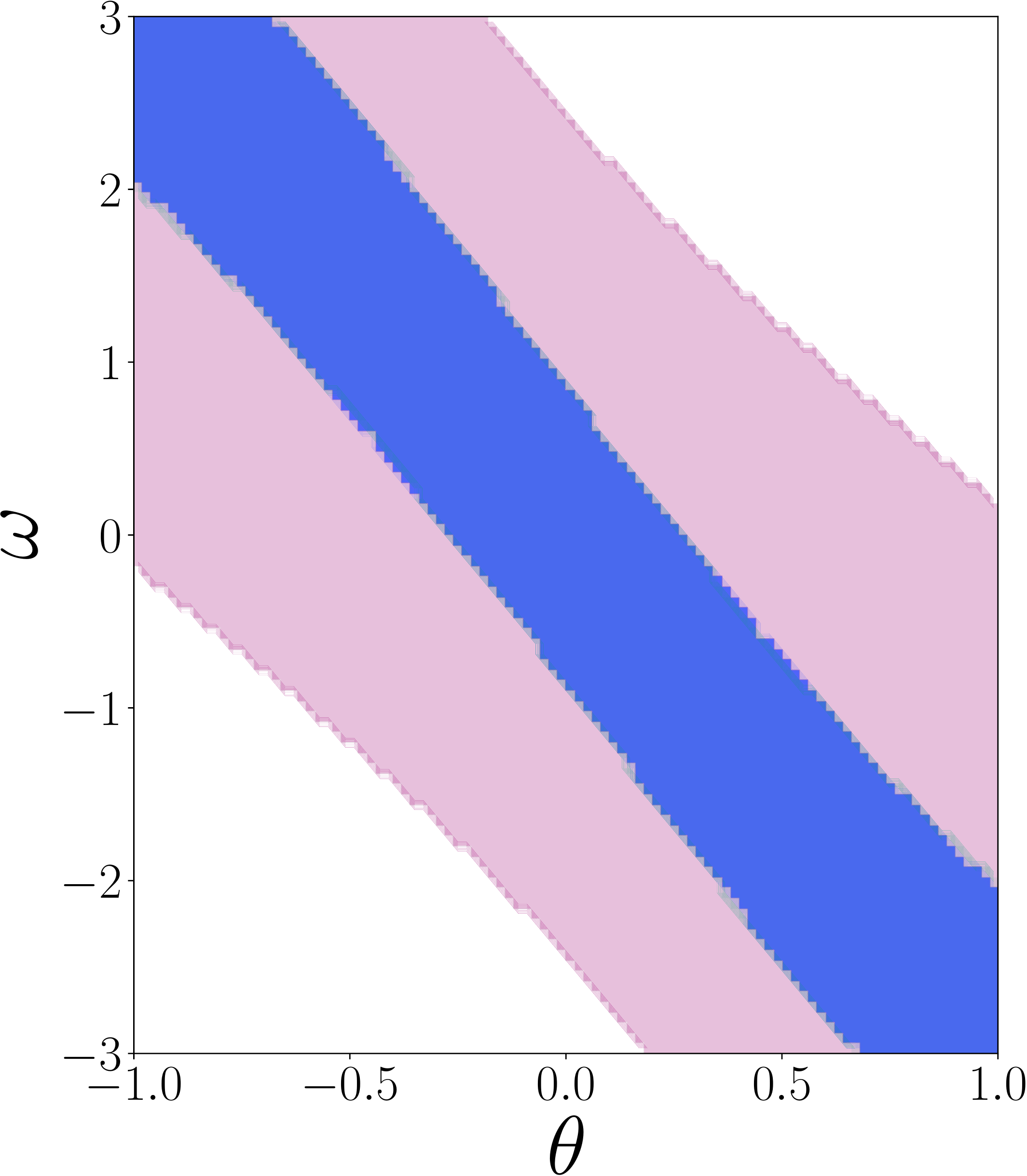}
\end{subfigure}
\begin{subfigure}[t]{\sfsize}
    \centering
    \includegraphics[width=1\textwidth]{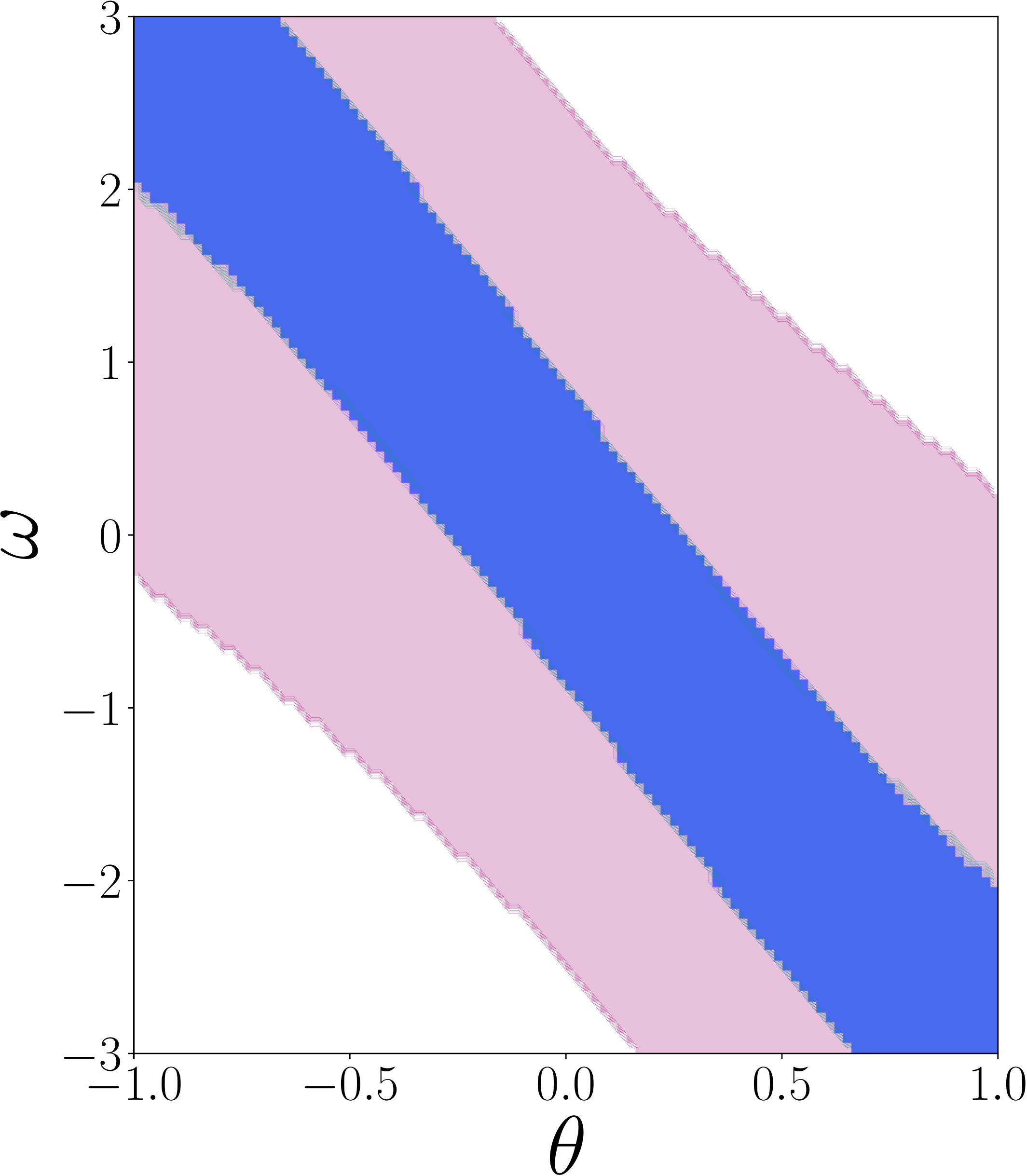}
\end{subfigure}

%% file: figs/increasing_roas/tex_files/increasing_nested_roa_figure_stage_6.tex
\begin{subfigure}[t]{\sfsize}
    \centering
    \includegraphics[width=1\textwidth]{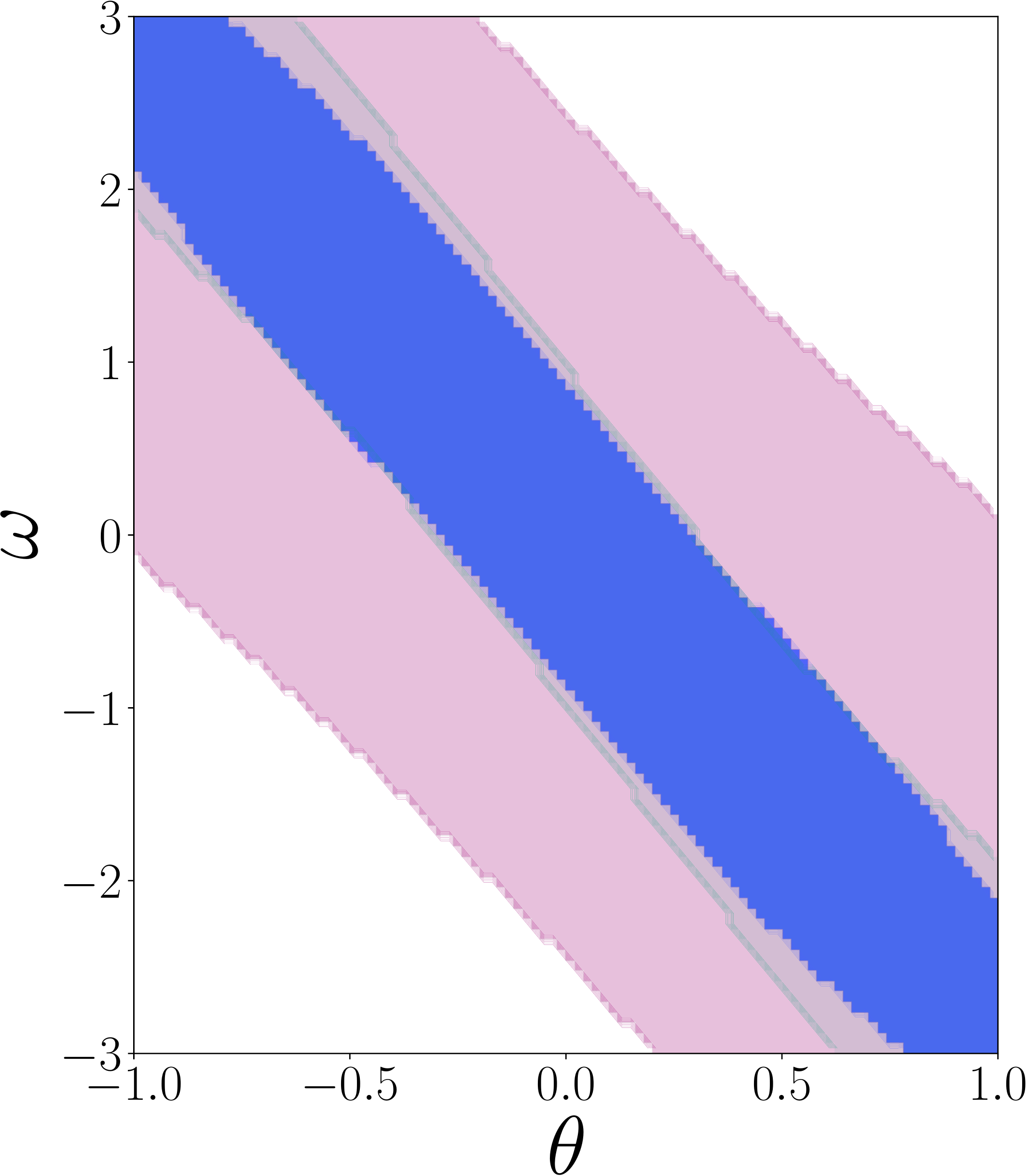}
\end{subfigure}
\begin{subfigure}[t]{\sfsize}
    \centering
    \includegraphics[width=1\textwidth]{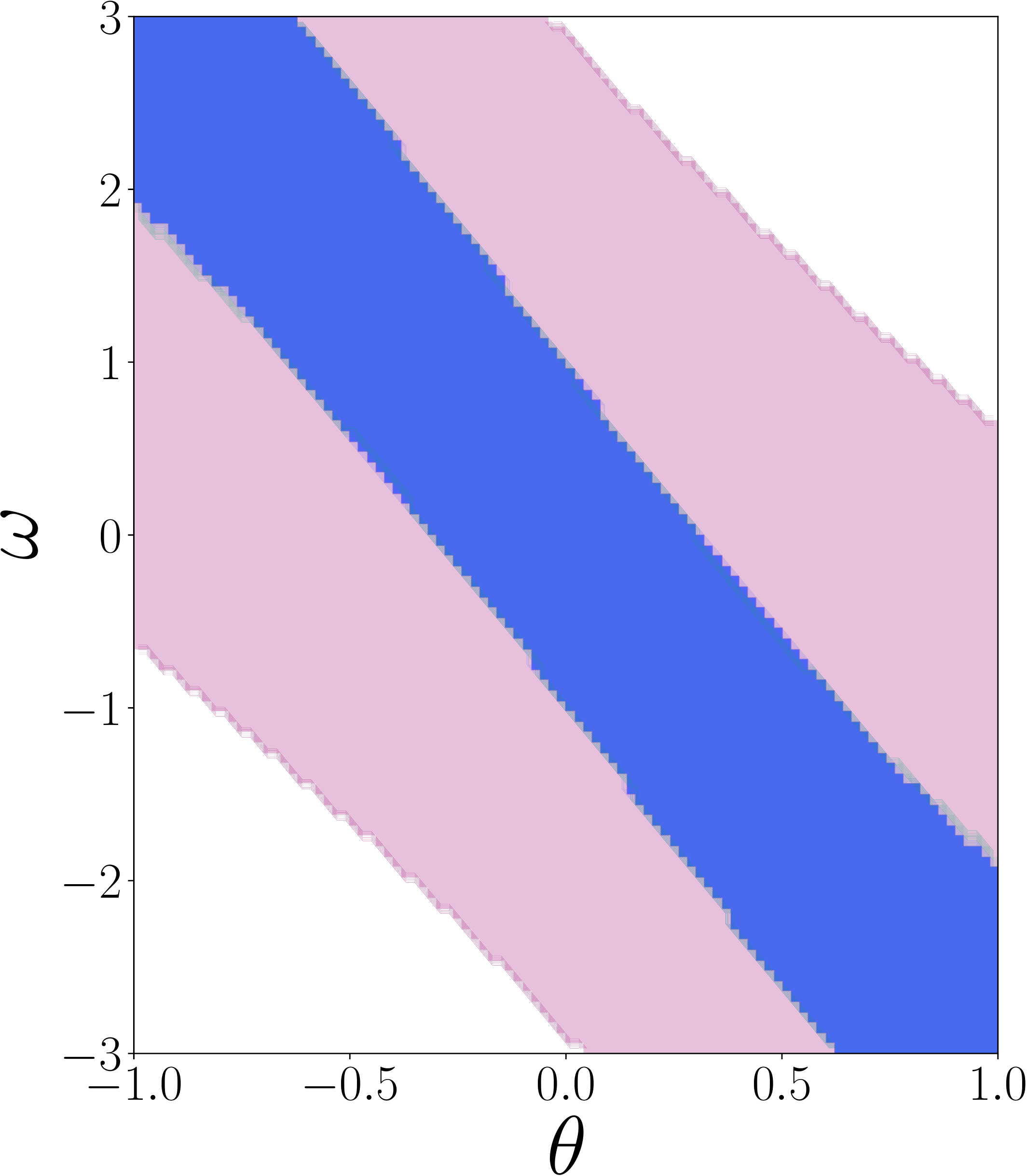}
\end{subfigure}
\begin{subfigure}[t]{\sfsize}
    \centering
    \includegraphics[width=1\textwidth]{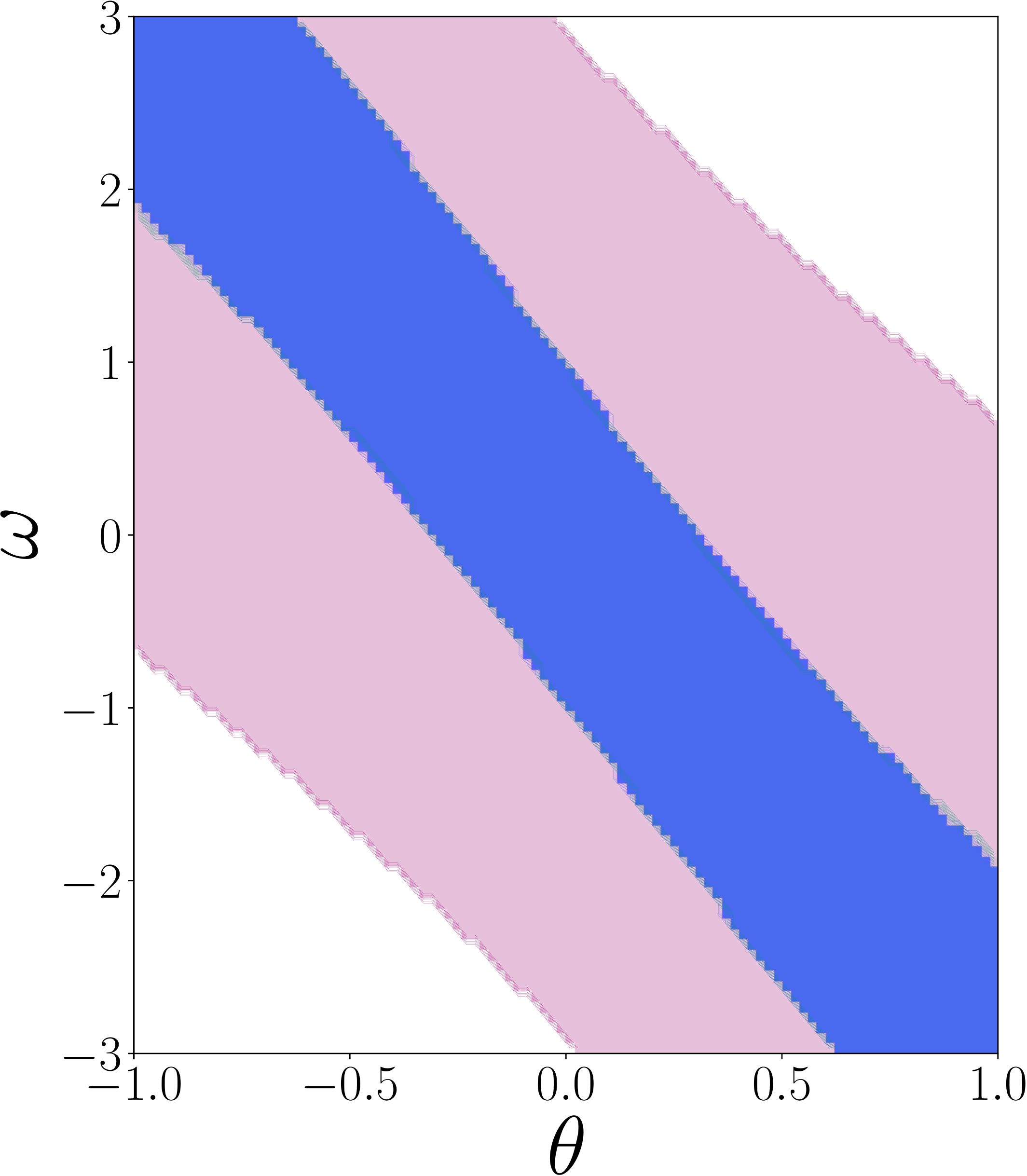}
\end{subfigure}
\begin{subfigure}[t]{\sfsize}
    \centering
    \includegraphics[width=1\textwidth]{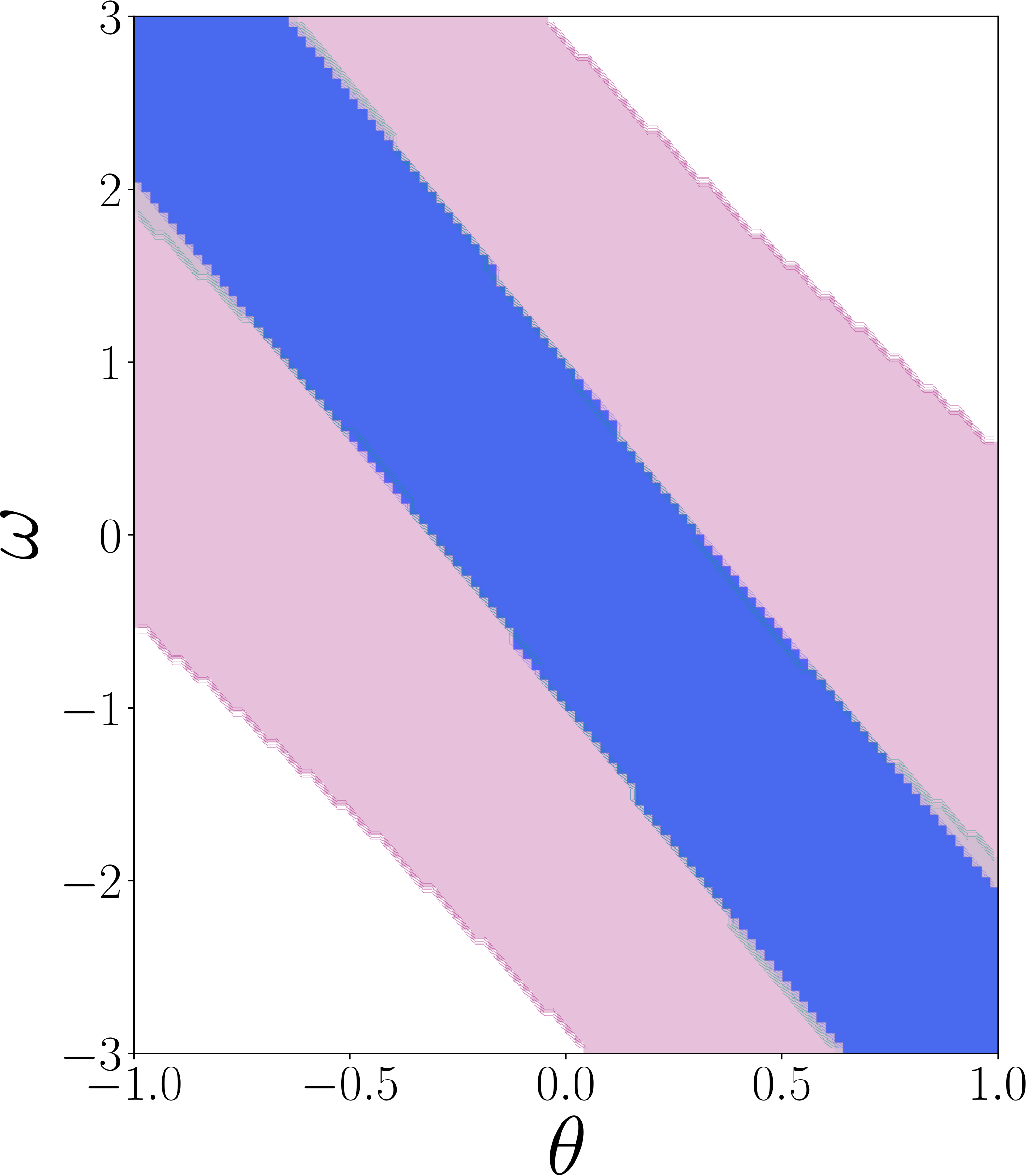}
\end{subfigure}
\begin{subfigure}[t]{\sfsize}
    \centering
    \includegraphics[width=1\textwidth]{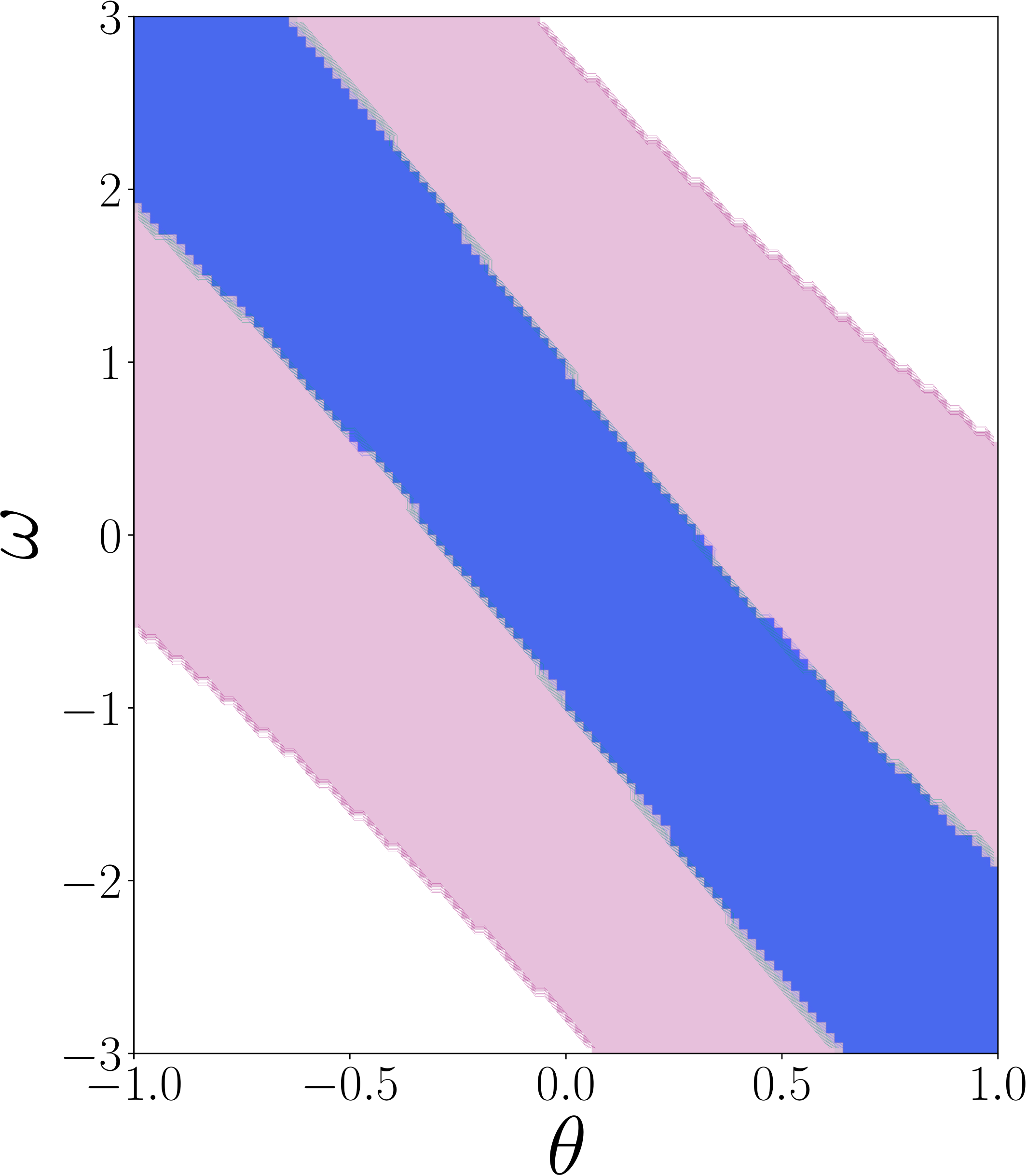}
\end{subfigure}
\begin{subfigure}[t]{\sfsize}
    \centering
    \includegraphics[width=1\textwidth]{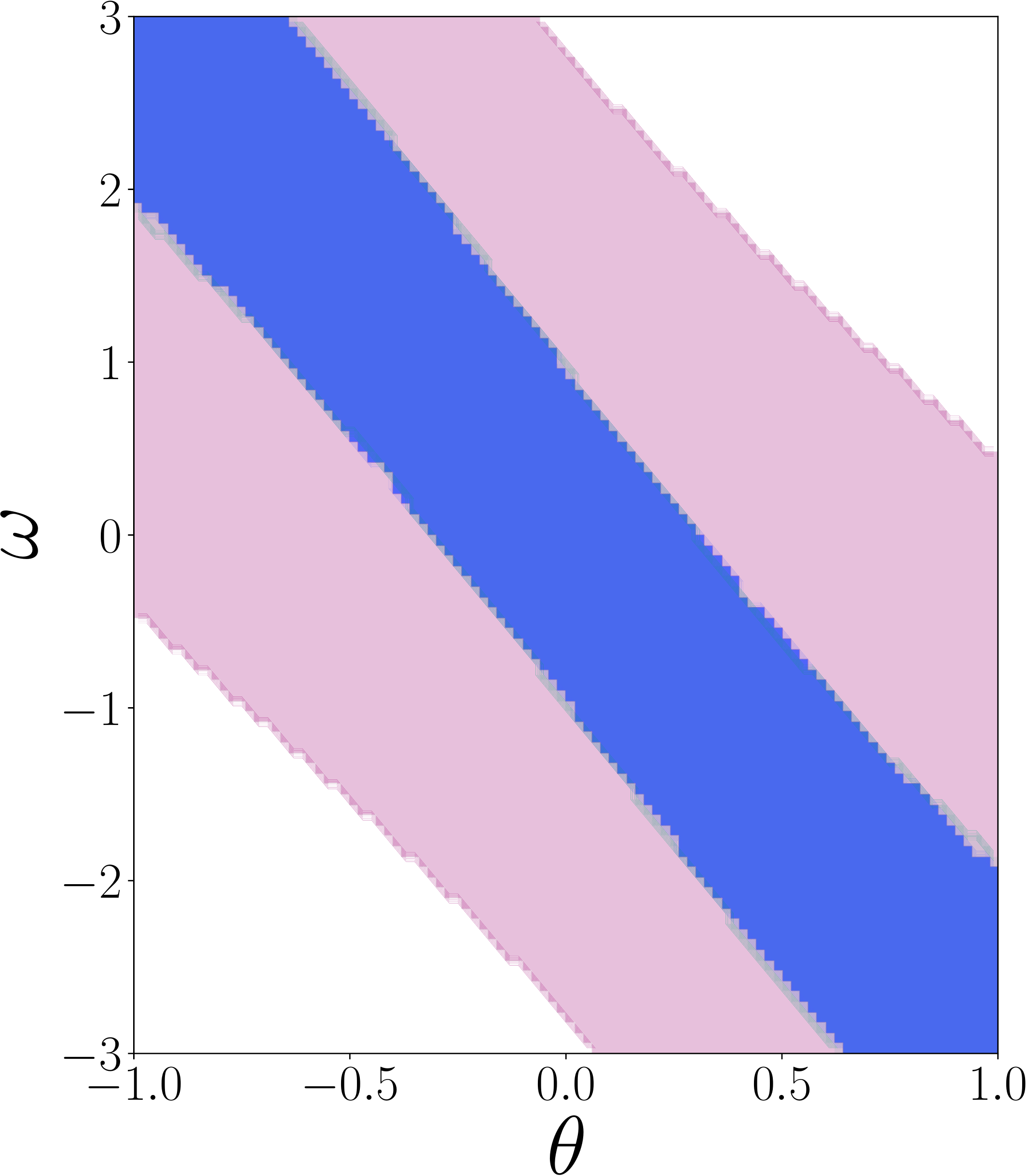}
\end{subfigure}